%% file: acmsmall.tex
\documentclass[acmsmall]{acmart}
\input{usepackage}

\newboolean{shortver}
\setboolean{shortver}{false}% for long version

\AtBeginDocument{%
  }

\setcopyright{cc}
\setcctype{by-nc-nd}
\acmDOI{10.1145/3839461}
\acmYear{2026}
\acmJournal{PACMPL}
\acmVolume{10}
\acmNumber{OOPSLA2}
\acmArticle{329}
\acmMonth{10}
\acmSubmissionID{oopslab26main-p402-p}
\received{2026-03-17}
\received[accepted]{2026-06-10}

\begin{document}

\title{EUF$^n$: A Decidable Extension to the Theory of Equality with Uninterpreted Functions}
% \title{EUFⁿ: A Decidable Extension to the Theory of Equality with Uninterpreted Functions}


\author{Yide Du}
\orcid{0009-0000-0199-5642}
\authornotemark[2]
\affiliation{%
  \department{College of Computer Science and Technology}
  \institution{National University of Defense Technology}
  \city{Changsha}
  \state{Hunan}
  \country{China}
}
\email{dyd1024@nudt.edu.cn}

\author{Zhenbang Chen}
\orcid{0000-0002-4066-7892}
\authornote{Zhenbang Chen is the corresponding author.}
\authornotemark[2]
\affiliation{%
  \department{College of Computer Science and Technology}
  \institution{National University of Defense Technology}
  \city{Changsha}
  \state{Hunan}
  \country{China}
}
\email{zbchen@nudt.edu.cn}

\author{Weijiang Hong}
\orcid{0000-0002-7092-3658}
\authornotemark[2]
\affiliation{%
  \department{College of Computer Science and Technology}
  \institution{National University of Defense Technology}
  \city{Changsha}
  \state{Hunan}
  \country{China}
}
\email{hongweijiang17@nudt.edu.cn}

\author{Wei Dong}
\orcid{0000-0002-8033-7943}
\affiliation{%
 \department{College of Computer Science and Technology}
 \institution{National University of Defense Technology}
 \city{Changsha}
 \state{Hunan}
 \country{China}
}
\authornote{Also with the affiliation: State Key Laboratory of Complex \& Critical Software Environment, National University of Defense Technology, Changsha, China.}
\email{wdong@nudt.edu.cn}

\renewcommand{\shortauthors}{Yide Du, Zhenbang Chen, Weijiang Hong, and Wei Dong}

\begin{abstract}

  The theory of Equality with Uninterpreted Functions (EUF) is fundamental to constraint solving and program verification. Uninterpreted functions abstract concrete implementations, enabling generalization and simplification of theorems and proofs. However, standard EUF restricts function composition to fixed finite depths (\emph{e.g.}, $f^k(x)$ where $k$ is constant). This work extends EUF to EUF$^n$, supporting \emph{parametric composition depth} for unary functions (\emph{e.g.}, $f^n(x)$ where $n$ is a natural number variable).

  % An EUF$^n$ formula can be viewed as a disjunction of infinitely many EUF formulas instantiated by natural number assignments. It is satisfiable if at least one of these EUF formulas
  % is satisfiable. We prove the decidability of EUF$^n$ satisfiability via a \emph{conditional congruence graph (CCG)} algorithm. This approach generalizes congruence
  % closure by maintaining conditional equivalence relationships between terms, reducing satisfiability to deciding existential sentences in Presburger arithmetic
  % with divisibility, a problem decidable in NEXPTIME~\cite{lechner2015complexity}.

  An EUF$^n$ formula can be viewed as a disjunction of infinitely many EUF formulas, each instantiated by an assignment of natural numbers. Its satisfiability is defined by the satisfiability of at least one such instantiated EUF formula. We establish the decidability of the EUF$^n$ satisfiability problem via a \emph{conditional congruence graph (CCG)} algorithm. This approach generalizes the standard congruence closure procedure by maintaining conditional equivalence relations between terms. The algorithm reduces the satisfiability problem to deciding existential sentences in Presburger arithmetic with divisibility, which is a decidable problem, thereby yielding a decision procedure for the quantifier-free fragment of EUF$^n$ with a 2NEXPTIME complexity upper bound.

  The enhanced expressiveness of EUF$^n$ enables new applications: (1) Encoding a decidable subclass of interleaved Dyck reachability problems where existing over/under-approximations produce false positives/negatives, and (2) Encoding a new decidable subclass of uninterpreted program verification problems.

\end{abstract}





\maketitle

\input{section/sec1.tex}

\input{section/sec2.tex}
\input{section/sec3.tex}

\input{section/sec4.tex}

\input{section/sec5.tex}
\input{section/sec6.tex}

\ifthenelse{\boolean{shortver}}{}{\appendix \input{section/appendix}
}

\bibliographystyle{ACM-Reference-Format}
\bibliography{ref}

\end{document}

%% file: usepackage.tex
\usepackage{amssymb}
\usepackage{stmaryrd}
\usepackage{amsmath}
\usepackage{mathrsfs}

\newtheoremstyle{mydefstyle}
  {6pt}
  {6pt}
  {\normalfont}
  {}
  {\bfseries}
  {.}
  { }
  {\thmname{#1}\thmnumber{ #2}\thmnote{ (#3)}}
\theoremstyle{mydefstyle}
\newtheorem{myDef}{Definition}
\newtheorem{myExample}{Example}
\newtheorem{myThe}{Theorem}
\newtheorem{myLem}{Lemma}

\usepackage{tcolorbox}

\usepackage{microtype}       % 更智能的微调
\usepackage{hyphenat}        % 改进连字符规则
\usepackage{ragged2e}        % 优化对齐控制
\usepackage{adjustbox}      % 灵活调整框架和图像

\usepackage[ruled, vlined, linesnumbered]{algorithm2e}
\SetAlgoSkip{0pt}
\SetAlgoNoEnd

\usepackage{subcaption}
\usepackage{graphicx}
\usepackage{tikz}
\usetikzlibrary{positioning, backgrounds, shadows, decorations.pathmorphing}
\usepackage{lstautogobble}  % Fix relative indenting
\usepackage{color}          % Code coloring
\usepackage{zi4}
\usepackage{enumitem}
\usepackage{wrapfig}

\definecolor{bluekeywords}{rgb}{0.13, 0.13, 1}
\definecolor{greencomments}{rgb}{0, 0.5, 0}
\definecolor{redstrings}{rgb}{0.9, 0, 0}
\definecolor{silvergrey}{RGB}{180,175,170}

\usepackage{listings}
\lstdefinelanguage{upl}%
{morekeywords={%
		assume,assert,%
		while,if
	},%
	numbers=none, %
	mathescape=true,%
	sensitive,%
	commentstyle=\color{blue},
	morecomment=[l]//,%
	morecomment=[s]{/*}{*/},%
	morestring=[b]",%
	morestring=[b]',%
	showstringspaces=false%,
}[keywords,comments,strings]%
\newcommand{\Kleene}{$^*$\xspace}  % Kleene star operator
\newcommand\stdef[1]{\langle {#1} \rangle}

%% file: section/sec1.tex
\section{Introduction}
The theory of Equality with Uninterpreted Functions (EUF) establishes a logical framework combining equality reasoning with
uninterpreted function symbols, where variables range over infinite domains and functions retain only their \emph{functional congruence}
property. Formally, this congruence axiom requires that identical inputs produce identical outputs:
$\forall x,y.\ x = y \rightarrow f(x) = f(y)$. A canonical demonstration of this principle appears in the EUF formula:
$x = y \land f(x) \neq f(y)$ which is inherently unsatisfiable due to the contradiction between the assumed equality $x=y$ and
the required congruence consequence $f(x)=f(y)$.

As a fundamental component of Satisfiability Modulo Theories (SMT), EUF plays a pivotal role in formal verification methodologies~\cite{bueno2021software}.
Its applications encompass a wide range of areas, including circuit equivalence verification~\cite{1997Automatic}, the proof of program equivalence
before and after compiler optimizations~\cite{muller2005checking}, verification of uninterpreted programs~\cite{govind2021logical}, and model
checking~\cite{lee2014unbounded, fan2024leveraging}. The essential power of EUF lies in its abstraction mechanism: uninterpreted functions
obscure implementation details to simplify verification. When verification fails under this abstraction, practitioners employ a \emph{refinement process} that
incrementally incorporates specific function properties until either identifying genuine counterexamples or completing the proof.
This methodology strikes an optimal balance between abstraction precision and verification tractability.

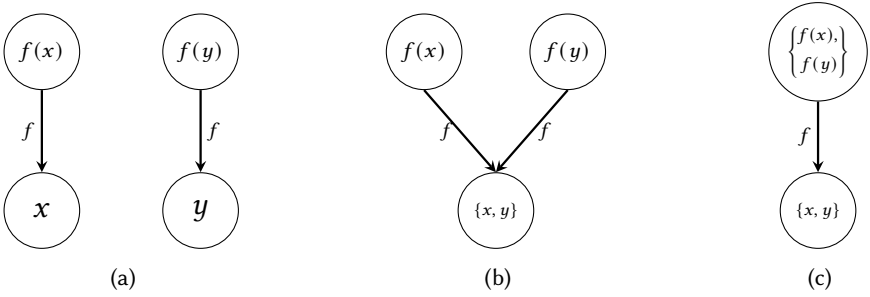
\begin{figure}[ht]
    \centering
    \input{figure/EUF_CG_solve.tex}
    \caption{Congruence Closure on an EUF Example.}
    \label{fig:EUF_solve}
    \vspace{-10pt}
\end{figure}

The satisfiability problem for EUF formulas is known to be decidable. The foundational work by \cite{congruence_closure3}, \cite{congruence_closure2}, and \cite{congruence_closure1} established the congruence closure algorithm,
which combines directed acyclic graph (DAG) representations with the \textit{Union-Find} data structure to maintain distinct equivalence
classes. Here, $\textit{Union}(x, y)$ merges the equivalence classes containing $x$ and $y$ while designating a new representative,
and $\textit{Find}(x)$ retrieves the representative of the equivalence class containing $x$. To illustrate, consider
the prototypical EUF formula: $$ x = y \land f(x) \neq f(y).$$
As illustrated in Figure~\ref{fig:EUF_solve_subfig1}, the initial equivalence classes are formed as distinct sets: $\{x\}$, $\{y\}$, $\{f(x)\}$,
and $\{f(y)\}$, which correspond to the individual nodes in the graph. The directed edges represent the superterm relationship between nodes;
for example, $f(x)$ is the term obtained by applying the function $f$ to $x$, and it is referred to as a superterm of $x$.
The algorithm then processes the equalities in the formula. Upon processing $x=y$, \textit{Union}($x$, $y$) is called to merge the
equivalence classes of $x$ and $y$ into $\{x, y\}$, as shown in Figure~\ref{fig:EUF_solve_subfig2}. This merger establishes $f(x)$ as
a superterm of $y$ and $f(y)$ as a superterm of $x$. Consequently, congruence propagation is triggered: since $x$ and $y$ are now
equivalent, functional congruence requires that $f(x)$ and $f(y)$ must also be equivalent, leading to the \textit{Union}($f(x)$, $f(y)$)
operation (Figure~\ref{fig:EUF_solve_subfig3}). Finally, when processing the disequality $f(x) \neq f(y)$, the condition
$\textit{Find}(f(x)) = \textit{Find}(f(y))$ indicates that $f(x)$ and $f(y)$ belong to the same equivalence class,
directly contradicting the disequality. This conflict, detected via congruence closure, confirms the formula's unsatisfiability.

% The algorithm then processes the equalities in the formula. When processing $x=y$, the algorithm first calls
% \textit{Union}($x$, $y$) to merge $\{x\}$ and $\{y\}$ into $\{x, y\}$, as shown in Figure~\ref{fig:EUF_solve_subfig2}. As a result of this merge,
% $f(x)$ becomes a superterm of $y$, and $f(y)$ becomes a superterm of $x$. The merged equivalence classes then trigger congruence
% propagation: since $x$ is equivalent to $y$, functional congruence requires that $f(x)$ and $f(y)$ must also be equivalent, which
% triggers the \textit{Union}($f(x)$, $f(y)$) operation, as shown in Figure~\ref{fig:EUF_solve_subfig3}. Finally, when processing the
% disequality $f(x) \neq f(y)$, the condition $\textit{Find}(f(x)) = \textit{Find}(f(y))$ indicates that $f(x)$ and $f(y)$ are
% equivalent, directly contradicting the asserted disequality. This conflict confirms the formula’s unsatisfiability through the
% combined mechanisms of congruence closure and conflict detection in the algorithm.

A fundamental limitation in standard EUF theory is its support for only \emph{finite-depth function composition}, as seen in terms like $f^c(x)$,
where $c$ is a constant (\emph{e.g.}, $f(f(x))$ and $f^6(x)$). To enhance the expressiveness of standard EUF, we propose EUF$^n$—a decidable extension
that allows \emph{parametric composition depth} for unary functions through terms of the form $f^n(x)$, where $n$ is a natural number variable.
The essential distinction lies in the finite notation $f^n(x)$, which has no direct analogue for multi-arity functions. Consequently,
our extension must prohibit parametrically nested compositions in the latter case. For example, a term such as $g(a, g(b, g(c, d)))$
has no concise parametric representation like $g^3(\cdot, \cdot)$. On the other hand, terms that mix parametric and
fixed-depth applications---such as $g(f^n(x), x)$ or $g(f^2(x), x)$---remain well-formed and admissible within EUF$^n$.
% The key is the finite representation $f^n(x)$, which has no equivalent for multi-arity functions; thus, the extension must
% prohibit parametric nesting in the latter case.
% % Crucially, this extension prohibits \emph{parametric composition depth} in multi-arity functions (\emph{e.g.}, $g(x, g(x, g(x, \dots)))$ with $g$ nested $n$ times),
% % as such terms cannot be finitely represented.
% For instance, a term like $g(a, g(b, g(c, d)))$ cannot be written as $g^3(c, d)$ or expressed in any
% similar compact form. In contrast, terms such as $g(f^n(x), x)$ or $g(f^2(x), x)$ are valid within EUF$^n$.
% Despite this restriction, EUF$^n$
% remains expressive enough for a wide range of verification tasks involving unary function abstractions.
Despite this restriction, EUF$^n$ remains sufficiently expressive for many verification tasks involving unary function abstractions.
For instance, it can model data-structure traversals---\emph{e.g.}, $read^n(arr)$ for array access, $next^n(head)$ for list
walking---where the parametric depth $n$ captures the access index or path length.
It is also applicable to control-property verification: by using a unary counting function $add^n(count)$, one can encode
constraints such as equality or inequality between the execution counts of two functions, which is useful for detecting anomalies
like use-after-free.
On the other hand, the term $f^n(x)$ can be expressed in first-order logic as $f'(n, x)$, where $n$ is a natural number variable and $x$ is
an uninterpreted variable. However, since the current combination of EUF and Presburger arithmetic does not support the derivation
of equalities such as $f'(1, f'(n, x)) = f'(n+1, x)$, it is necessary to extend the syntax and axioms of EUF and to develop specialized
algorithms for solving $\mathrm{EUF}^n$ formulas.

This paper introduces the syntax and axioms of EUF$^n$, along with a \emph{conditional congruence graph} (\emph{CCG}) algorithm for solving
quantifier-free EUF$^n$ formulas, establishing the decidability of their satisfiability problem. Analogous to the standard congruence
closure algorithm for EUF (Figure~\ref{fig:EUF_solve}), this approach represents terms as nodes and functional relationships as edges.
The key distinction lies in the \emph{CCG}, which maintains both conditional equivalence and conditional superterm relationships,
explicitly capturing the conditions under which terms become equivalent or functionally related. For example,
a conditional equivalence edge between terms $f^m(x)$ and $f^n(x)$ is annotated with the condition $m \text{=} n$. As a result, when
processing equalities in EUF$^n$ formulas, the algorithm updates only the associated conditional constraints, rather than merging the
corresponding nodes as in Figure~\ref{fig:EUF_solve}. Finally, the algorithm reduces satisfiability of EUF$^n$ formulas to the satisfiability problem
for the existential theory of Presburger arithmetic with divisibility, which is known to be decidable~\cite{lipshitz1978diophantine, bel1976decidability}
in NEXPTIME~\cite{lechner2015complexity} thereby establishing decidability.

EUF$^n$ extends the expressiveness of standard EUF, enabling broader applications. On one hand, the congruence property of functions closely parallels
bracket matching: terms obtained by applying the same function the same number of times are treated as equivalent, analogous to matched brackets
in balanced expressions. This correspondence allows EUF$^n$ to encode Dyck reachability problems over graphs, where two nodes are Dyck-reachable
if a path between them forms a well-balanced bracket string in a Dyck language~(\emph{i.e.}, matching brackets). Dyck reachability~\cite{zhang2013fast}
and its extension, interleaved Dyck reachability~\cite{zhang2017context,li2023single,conrado2024better}, play a central role in static analysis with
wide-ranging applications. Although interleaved Dyck reachability is not an equivalence relation and is generally undecidable~\cite{reps2000undecidability},
we show that in the single-source single-target setting, a decidable subclass can be encoded using EUF$^n$. In contrast, existing over- and under-approximation
techniques may still yield false positives or negatives even in this restricted case.

Additionally, uninterpreted functions abstract concrete implementations to simplify verification. Programs that use such abstractions,
which are referred to as uninterpreted programs, have undecidable verification in general~\cite{mathur2019decidable}. Unlike prior
approaches that reason over program states (\emph{e.g.}, equalities, disequalities and function relations)~\cite{mathur2019decidable,govind2021logical,hong2021trace},
EUF$^n$ supports encoding of certain \emph{execution traces}. For instance, a loop like \texttt{while(cond)\{x=next(x);\}}
yields access chains $next^n(x)$, where $n$ depends on iteration count. Standard EUF loses this parametric relation by requiring
concrete $n$~\cite{mathur2019deciding}, while EUF$^n$ retains it symbolically, which identifies a decidable subclass of uninterpreted programs
beyond the \emph{coherent} class~\cite{mathur2019decidable,govind2021logical}.

To summarize, this paper presents the following contributions:
\begin{enumerate}
    \item We propose EUF$^n$, an extension of EUF that enables the representation of terms with \emph{parametric composition depth} for unary functions.
    \item We establish the decidability of the quantifier-free fragment of EUF$^n$ via the \emph{CCG} algorithm with 2NEXPTIME complexity.
    \item We demonstrate the practical expressiveness of EUF$^n$ by encoding a subclass of single-source single-target interleaved Dyck reachability problems,
    as well as a subclass of verification problems for uninterpreted programs. In both cases, we identify new decidable subclasses enabled by this encoding.
\end{enumerate}

The structure of this paper is organized as follows: Section \ref{sec:section2} introduces the basic concepts of EUF$^n$,
Section \ref{sec:section3} provides the decision algorithm for EUF$^n$, Section~\ref{sec:section4} examines the practical
encoding capabilities of EUF$^n$, Section \ref{sec:section5} reviews related work,
and Section \ref{sec:section6} concludes the paper.

%% file: figure/EUF_CG_solve.tex
\tikzstyle{squarednode} = [rectangle, rounded corners, draw = gray, fill=gray!5, very thick, text centered, align=center]%, minimum size=5mm]
\tikzstyle{originnode} = [circle, rounded corners, draw = black, fill=white!5,  text centered, align=center, minimum size=10mm, font=\bfseries\Large]
\tikzstyle{arrow} = [->,>=stealth, draw=black, line width=0.3mm]
\tikzstyle{line} = [-,>=stealth, draw=black]

\begin{subfigure} {0.35\linewidth}
    \centering
    {\scriptsize
        \begin{tikzpicture}
            \node[originnode] (x){$x$};
            \node[originnode, right of=x, node distance=21mm] (y){$y$};
            \node[originnode, font=\footnotesize, above of=x, node distance=21mm] (fx){$f(x)$};
            \node[originnode, font=\footnotesize, above of=y, node distance=21mm] (fy){$f(y)$};

            \draw[arrow] (fx.south) --  node[left, align=center]{$f$} (x.north);
            \draw[arrow] (fy.south) --  node[right, align=center]{$f$} (y.north);
        \end{tikzpicture}
    }
    \caption{ }
    \label{fig:EUF_solve_subfig1}
\end{subfigure}
\begin{subfigure} {0.35\linewidth}
    \centering
    {\scriptsize
        \begin{tikzpicture}
            \node[originnode, font=\tiny] (x){$\{x,y\}$};
            \node[above of=x, node distance=21mm] (xx){};
            \node[originnode, font=\footnotesize, left of=xx, node distance=9.5mm] (fx){$f(x)$};
            \node[originnode, font=\footnotesize, right of=xx, node distance=9.5mm] (fy){$f(y)$};

            \draw[arrow] (fx.south) --  node[left, align=center]{$f$} (x.north);
            \draw[arrow] (fy.south) --  node[right, align=center]{$f$} (x.north);
        \end{tikzpicture}
    }
    \caption{ }
    \label{fig:EUF_solve_subfig2}
\end{subfigure}
\begin{subfigure}{0.25\linewidth}
    \centering
    {\scriptsize
        \begin{tikzpicture}
            \node[originnode, font=\tiny] (x){$\{x,y\}$};
            \node[originnode, font=\tiny, above of=x, node distance=21mm] (fx){
                $\left\{
                    \begin{aligned}
                        f(x), \\
                        f(y)
                    \end{aligned}
                \right\}$%
            };
            \draw[arrow] (fx.south) --  node[left, align=center]{$f$} (x.north);
        \end{tikzpicture}
    }
    \caption{}
    \label{fig:EUF_solve_subfig3}
\end{subfigure}

%% file: section/sec2.tex
\section{Logic of EUF$^n$}  \label{sec:section2}
This section presents the formal syntax and axioms of EUF$^n$, followed by a detailed example illustrating its decision procedure.

\subsection{Syntax of EUF$^n$} \label{section2:formal_syntax}
% The syntax of EUF$^n$ is defined over a finite signature $\Sigma = \mathsf{F} \cup \mathsf{V} \cup \mathsf{C}$, where $\mathsf{F}$ denotes
% function symbols (with $f = (f_{1} \circ \cdots \circ f_{n})^{lin} \in \mathsf{F}$ representing unary functions and $g \in \mathsf{F}$ denoting $n$-ary functions where
% $n \geq 2$), $\mathsf{V} = \mathsf{V}_{\mathsf{U}} \cup \mathsf{N}$ partitions variables into uninterpreted-type variables
% ($\mathit{var} \in \mathsf{V}_{\mathsf{U}}$) and natural number-type variables ($\mathit{int\_var} \in \mathsf{N}$), and $\mathsf{C}$ contains
% constant symbols ($c$).

The syntax of EUF$^n$ is defined over a finite signature $\mathsf{\Sigma}=\mathsf{V} \cup \mathsf{F} \cup \mathsf{C}$.
Variables are partitioned into $\mathsf{V} = \mathsf{U} \cup \mathsf{N}$, where $\mathsf{U}$ consists of uninterpreted-type
variables ($\mathit{var} \in \mathsf{U}$) and $\mathsf{N}$ includes natural number-type variables ($\mathit{int\_var} \in \mathsf{N}$).
Here, $F$ denotes the set of function symbols, which includes unary function symbols (\emph{e.g.}, the successor function $S$) and $k$-ary
function symbols for $k \geq 2$ (\emph{e.g.}, the binary symbol $+$). A composite unary term can be written as $f_1 \circ \cdots \circ f_n$,
where each $f_i \in F$ ($1 \leq i \leq n$) is a unary function symbol. Finally, $C$ denotes the set of integer constants, with
$0, c, c_1, c_n \in C$.
% % Here, $\mathsf{F} = \mathsf{F}_{\mathsf{U}} \cup \mathsf{F}_M$ denotes the set of function symbols. The set $\mathsf{F}_{\mathsf{U}}$
% % includes unary function symbols of the form $f = f_1 \circ \cdots \circ f_n$, where each $f_i: \mathsf{U} \to \mathsf{U}$ for $1 \leq i \leq n$,
% % as well as $k$-ary functions $g: \mathsf{U}^k \to \mathsf{U}$ for $k \geq 2$. The set $\mathsf{F}_M$ consists of binary functions of the form
% % $f': \mathbb{N} \times \mathsf{U} \to \mathsf{U}$.
% Here, $\mathsf{F}$ denotes a set of function symbols, which includes unary function symbols (\emph{e.g.}, $f$, $S$) and $k$-ary function
% symbols for $k \geq 2$ (e.g., the binary symbol $+$). A unary function symbol can be applied repeatedly to form a composite term
% of the form $f_1 \circ \cdots \circ f_n$, where each $f_i: \mathsf{U} \to \mathsf{U}$ for $1 \leq i \leq n$.
% % Here, $\mathsf{F}$ denotes a set of function symbols. It includes unary function symbols (\emph{e.g.}, $f$, $S$)  and $k$-ary function
% % symbols for $k \geq 2$ (\emph{e.g.}, the binary symbol $+$).
% % % A unary function symbol can be applied repeatedly to form a composite
% % % term $f_1(f_2(\cdots f_n(x)\cdots))$, where each $f_i \in \mathsf{F}$ is unary.
% And $\mathsf{C}$ denotes the set of integer constants, with $0, c, c_1, c_n \in \mathsf{C}$.

\begin{figure}[ht]
    \vspace{-10pt}
    \begin{center}
    \begin{align*}
        formula &::= \mathbf{true} \mid \mathbf{false} \mid term = term \\
        &\quad \ \ \mid formula \land formula \mid formula \lor formula \mid \lnot formula \\
        term &::= var \mid f^{lin} (term) \mid g (term, \cdots, term) \\
        lin &::= c \mid int\_var \mid c_1 \cdot lin+\cdots+c_n \cdot lin
        % \oplus
    \end{align*}
    \end{center}
    \vspace{-0.5cm}
    \caption{The syntax of $\text{EUF}^{n}$ formulas.}\label{fig:syntax}
    \vspace{-10pt}
\end{figure}

As formalized in Figure~\ref{fig:syntax}, the syntax employs a three-tiered structure:
\begin{itemize}
    \item   $formula$ represents an EUF$^n$ formula, constructed from Boolean literals ($\mathbf{true}$, $\mathbf{false}$), equality predicates over $term$s (atomic formulas),
            or logical combinations of subformulas using standard connectives.
    \item   $term$ denotes functional expressions within formulas, which may be uninterpreted-type variables ($var$) or function applications,
            where $f$ represents a unary function and $g$ denotes a multi-arity function.
    \item   $lin$ specifies the composition depth for unary function $f$, defined as either constants ($c$), natural number-type variables ($int\_var$),
            or a linear combination of existing $lin$ expressions. The value of any $lin$ expression is guaranteed to be non-negative during formula solving.
\end{itemize}

\subsection{The Axioms of EUF$^n$} \label{section2:EUFn_axiom}
To enable \emph{parametric composition depth} for unary functions, we extend the EUF theory into a many-sorted EUF$^n$ with two distinct
sorts: $\mathsf{U}$ (uninterpreted) and $\mathsf{N}$ (natural numbers). The extended theory preserves the four core axioms of EUF,
with $\vec{x}, \vec{y} \in \mathsf{U}^k$ denoting $k$-tuples:

\begin{enumerate}[label=(\arabic*),noitemsep]
    \item $ \forall x:\mathsf{U}. \ x = x   $   \hfill  (reflexivity)   \label{axiom1}
    \item $ \forall x, y:\mathsf{U}. \ x = y \rightarrow y = x   $  \hfill  (symmetry)   \label{axiom2}
    \item $ \forall x, y, z:\mathsf{U}. \ x = y \wedge y = z \rightarrow x = z   $  \hfill  (transitivity)      \label{axiom3}
    \item $ \forall \vec{x},\vec{y}:\mathsf{U}^k. \ {\bigwedge}_{i=1}^k x_i = y_i \rightarrow  f(\vec{x}) = f(\vec{y})$ \hfill (congruence)  \label{axiom4}
\end{enumerate}
\noindent This extension further requires the axioms of Presburger arithmetic (see Appendix~\ref{sec:appendix_presburger}) along with the following two novel axioms governing function composition depth:
\begin{enumerate}[label=(\arabic*),noitemsep,resume*]
    \item $ \forall x:\mathsf{U}. \ f^{0}(x) = x   $    \label{axiom_new1}
    \item $ \forall x:\mathsf{U},n:\mathsf{N}. \ f^{S(n)}(x) = f(f^{n}(x)) $   \label{axiom_new2}
\end{enumerate}

\noindent
where $x$ is a variable of uninterpreted type $(\mathsf{U})$ and $n$ is a natural number variable of type $(\mathsf{N})$,
while $S$ denotes the successor function.

The model of EUF$^n$ is defined as $\mathcal{M} = (\mathcal{U} \cup \mathbb{N}, \mathcal{I})$, where $\mathcal{U}$ denotes
the domain of uninterpreted elements, $\mathbb{N}$ denotes the domain of natural numbers. For every unary function
$f: \mathsf{U} \to \mathsf{U}$, there exists a corresponding binary
function $f': \mathsf{N} \times \mathsf{U} \to \mathsf{U}$ such that for any variable $n : \mathsf{N}$, $x : \mathsf{U}$,
and assignment functions $\sigma: \mathsf{N} \to \mathbb{N}$ and $\sigma': \mathsf{U} \to \mathcal{U}$,
we have
\[
\mathcal{I}(f')(\sigma(n), \sigma'(x)) = \underbrace{\mathcal{I}(f) \circ \cdots \circ \mathcal{I}(f)}_{\sigma(n)\ \text{times}} (\sigma'(x)).
\]
A natural way to satisfy this requirement is for $\mathcal{I}(f')$ to apply $\mathcal{I}(f)$ repeatedly $n$ times. In the axioms,
we use $f^n(x)$ as a shorthand for $f'(n,x)$. The interpretations of addition ($+$), and constants
in Presburger arithmetic follow the standard semantics of natural number addition theory.

\begin{myThe}[Generalized Congruence] \label{theorem:Gen_Congruence}
    $$
    \forall x,y:\mathsf{U},\ m,n:\mathsf{N}. \  x=y \land m=n \rightarrow  f^{n}(x)=f^{m}(y)
    $$
\end{myThe}
\begin{proof}
    The case where $m = n = 0$ is secured by Axiom \ref{axiom_new1}. The general case can be established by induction,
    where the inductive step would rely on Axioms~\ref{axiom_new2} and~\ref{axiom4}, thus allowing for the completion of the proof.
\end{proof}

Having defined the syntax and axioms of EUF$^n$ formulas, we now define their satisfiability in terms of the satisfiability
of EUF formulas~\cite{shostak1978algorithm}. For a given formula $\Psi$ and an assignment $\sigma: \mathsf{N} \to \mathbb{N}$,
we define $\Psi[\sigma]$ as the formula obtained by replacing every free occurrence of a variable $x$ in $\Psi$ with the
corresponding value $\sigma(x)$.

\begin{myDef}[EUF\textsuperscript{$n$} Satisfiability] \label{thm:eufn-sat-equivalence}
    A formula $\Psi$ in $\mathrm{EUF}^n$ over signature $\Sigma$ is satisfiable iff there exists an assignment $\sigma: \mathsf{N}\to\mathbb{N}$
    such that $\Psi[\sigma]$ is satisfiable in EUF.
\end{myDef}
% That is, for any model $\mathcal{M}$ satisfying Equation~\ref{equality:interpretation}, we have $(\mathcal{M}, \sigma) \models \Psi$.

\subsection{An Example of the Decision Procedure for EUF$^n$} \label{section2:example}
EUF$^n$ extends EUF by supporting terms with \emph{parametric composition depth} for unary functions, thereby transforming equivalence
relations between terms into \emph{conditional equivalence}: two terms become equivalent only when their composition depth satisfys
specific constraints. To demonstrate the decision procedure for EUF$^n$, we employ a concrete example to illustrate its
core mechanisms. Consider the following EUF$^n$ formula as a representative case study:
\begin{align}\label{equality:EUFn_ex}
 x = y \land z = f^m(x) \land k = f^n(y) \land z \neq k,
\end{align}
where $m,n \in \mathsf{N}$ denote the composition depth of the
uninterpreted unary function $f$. This formula can be regarded as the disjunction of EUF formulas generated by instantiating
$m$ and $n$ over the natural numbers $\mathbb{N}$, as shown below:
% \begin{equation}  \label{equality:EUFn_EUF}
%     \begin{aligned}
%     & x = y \land z = f^m(x) \land k = f^n(y) \land z \neq k \\
%     & \Leftrightarrow \\
%     & \left\{ x = y \land z = f^0(x) \land k = f^0(y) \land z \neq k \right\} \lor \\
%     & \left\{ x = y \land z = f^1(x) \land k = f^0(y) \land z \neq k \right\} \lor \\
%     & \left\{ x = y \land z = f^0(x) \land k = f^1(y) \land z \neq k \right\} \lor \\
%     & \ldots
%     \end{aligned}
% \end{equation}
\begin{equation}  \label{equality:EUFn_EUF}
    \begin{aligned}
    & x = y \land z = f^m(x) \land k = f^n(y) \land z \neq k \Longleftrightarrow
        (x = y \land z = f^0(x) \land k = f^0(y) \land z \neq k) \lor \\
    &   (x = y \land z = f^1(x) \land k = f^0(y) \land z \neq k) \lor
        (x = y \land z = f^0(x) \land k = f^1(y) \land z \neq k) \lor \cdots
    \end{aligned}
\end{equation}
The EUF$^n$ formula is satisfiable if and only if there exists an assignment $\sigma$ such that the corresponding EUF formula
$x = y \land z = f^{\sigma(m)}(x) \land k = f^{\sigma(n)}(y) \land z \neq k$ is satisfiable.

Figure~\ref{fig:solve_EUFn} presents the solving process of the formula based on the \emph{conditional congruence graph} (\emph{CCG}) algorithm.
In these graphs, nodes represent $term$s from the formula, and directed edges indicate superterm relationships. For example,
the $term$ $f^{m}(x)$ (\emph{i.e.}, $z$) is a superterm of $x$, with edge label $(m \text{>} 0, f, m)$ encoding the function $f$,
its composition depth $m$, and the condition $m > 0$ that validates the superterm relation. Undirected edges connect equivalent
$term$s, annotated with logical conditions governing their equivalence. Figure~\ref{fig:solve_EUFn_sub1} shows the initialized
\emph{CCG} for the example, which includes the conditional superterm relationships between the terms in the formula. Notably,
during resolution, the graph preserves all existing nodes and introduces only new edges.

The \emph{CCG} provides a symbolic representation of the family of concrete congruence graphs induced by the disjunctive EUF formulas
above. For any concrete assignment to parameters $m$ and $n$, the corresponding concrete congruence graph is obtained from the \emph{CCG}
via a two-step instantiation: first, evaluating the parametric conditions on its edges with the given values; second, keeping only
those edges whose conditions are satisfied. After merging equivalent nodes, the resulting graph coincides exactly with the congruence
graph of the instantiated EUF formula.

\input{figure/solve_EUFn.tex}

The decision procedure starts by processing the equation $x = y$, which introduces an undirected edge between nodes
$x$ and $y$, annotated with the constraint $\mathbf{true}$. To maintain correctness after inserting a conditional equivalence edge,
three operations are performed on the \emph{CCG}; collectively, they symbolically abstract the corresponding operations needed for
the entire family of concrete congruence graphs it represents.
% Following each conditional equivalence edge insertion, three operations are performed
% to maintain the correctness of the \emph{CCG}, which symbolically abstracts a family of concrete congruence graphs:
\begin{enumerate}
    \item \textbf{Equivalence propagation.} Let $[x]$ denote the set containing $x$ and all nodes connected to $x$ via conditionally
        equivalent edges (similarly for $[y]$). This operation corresponds to the \textit{Union} procedure in standard EUF congruence
        closure algorithms. The step establishes pairwise equivalence between all nodes in $[x]$ and $[y]$, updating their conditions
        by conjoining with the equivalence condition of $x$ and $y$. This operation is illustrated in Figure~\ref{fig:solve_EUFn_sub2}
        with the addition of the edge $x \stackrel{\mathbf{true}}{\rule[0.5ex]{2em}{0.5pt}} y$.
    \item \textbf{Superterm inheritance.} After adding the conditional equivalence between $x$ and $y$, their superterms/subterms
        consequently become shared under the given condition. Therefore, in this step, $x$ and $y$ inherit each other's superterm/subterm
        relations and update the associated conditions. In Figure~\ref{fig:solve_EUFn_sub3}, edges $k \xrightarrow{(n>0 \land \mathbf{true}, f, n)} x$ and
        $z \xrightarrow{(m>0 \land \mathbf{true}, f, m)} y$ are added. Moreover, new conditional superterm relationships can arise
        between terms such as $f^m(x)$ (\emph{i.e.} $z$) and $f^n(y)$ (\emph{i.e.} $k$). This follows from the extended formalism's support
        for parametric composition depth, which forces consideration of all possible functional relationships between terms.
        As shown in Figure~\ref{fig:solve_EUFn_sub3}, the edges $z  \xrightarrow{(m>n \land m>0 \land n>0, f, m-n)} k$ and
        $k \xrightarrow{(n>m \land m>0 \land n>0, f, n-m)} z$ are introduced.
    \item \textbf{Conditional congruence closure.} When solving EUF formulas, after adding equivalence relations, it is necessary to check
       for additional derivable equivalences based on congruence axioms. Similarly, when using the \emph{CCG} to handle EUF$^n$ formulas,
       conditional equivalence relations must also be processed to ensure (conditional) congruence closure. In Figure~\ref{fig:solve_EUFn_sub4},
       given the superterm relations $z = f^m(x)$ and $k = f^n(y)$, the edge
       $z \stackrel{m = n \land m > 0 \land n > 0}{\rule[0.5ex]{6em}{0.5pt}} k$ is added to the graph. In this case, there is no superterm
       relationship between $x$ and $y$; however, when such a relationship exists, the congruence closure process differs. This will be
       discussed in detail in Section~\ref{section3:congruence}.
\end{enumerate}

After adding the conditional equivalence between $z$ and $k$, the same three steps repeat. Since $z$ and $k$ have no additional
equivalent nodes (this simplified example omits edges such as the one from $x$ to $z$ with the constraint $m=0$),
\textbf{equivalence propagation} is unnecessary. The only conditional superterm of $k$ is $z$, but the superterm condition
$m=n \land m>n \land m<n$ is unsatisfiable, terminating \textbf{superterm inheritance} for $z$ (and similarly for $k$). Finally,
while $z$ and $k$ are mutual superterms, the superterm conditions remain unsatisfiable, so the \textbf{conditional congruence closure} terminates.

Finally, the decision procedure handles the disequality $z \neq k$ by examining the conditional equivalence relation between $z$ and $k$
in the \emph{CCG}, subject to the constraint $m = n \land m > 0 \land n > 0$. Since $\neg(m = n \land m > 0 \land n > 0)$ is satisfiable
(\emph{e.g.}, by assigning distinct positive integers to $m$ and $n$), it follows that in this case, no contradiction arises. Consequently,
the original EUF$^n$ formula is satisfiable for such assignments.

As shown in this example, for any assignment $\sigma$, substituting $\sigma(m)$ and $\sigma(n)$ for $m$ and $n$ in both the terms and
edge conditions of Figure~\ref{fig:solve_EUFn} determines which edges remain valid. The resulting congruence graph, which is obtained
by merging equivalence edges with $\mathbf{true}$ conditions and removing redundant superterm edges, matches the congruence graph
obtained by directly evaluating the EUF formula from Equation~\ref{equality:EUFn_EUF} under $\sigma$. Thus, the \emph{CCG} in
Figure~\ref{fig:solve_EUFn} precisely captures the congruence graphs of all the disjuncted EUF formulas: $x$ and $y$ are always
equivalent, while $z$ and $k$ are equivalent if and only if $\sigma(m) = \sigma(n)$.

%% file: figure/solve_EUFn.tex
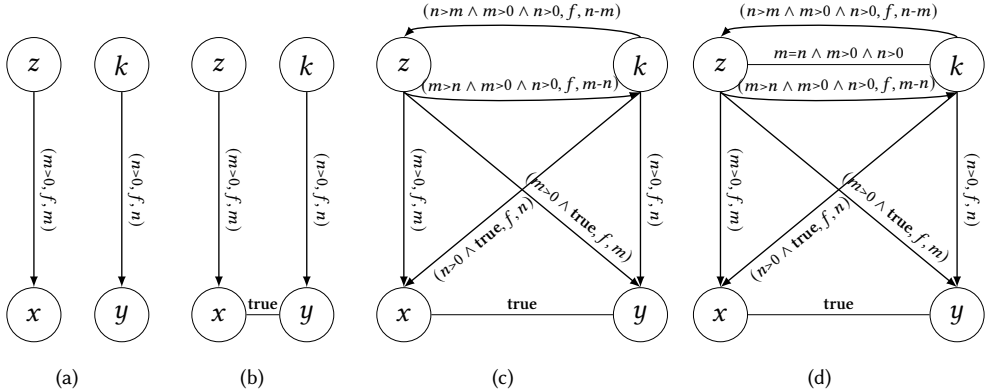
\begin{figure}[htbp]
    \centering
    \vspace{-8pt}
    \resizebox{0.9\textwidth}{!}{
    \begin{minipage}{\textwidth}
        \centering
        \begin{subfigure}[b]{0.13\textwidth}
            \centering
            \input{figure/solve_EUFn_1.tex}
            \vspace{-8pt}
            \caption{}
            \label{fig:solve_EUFn_sub1}
        \end{subfigure}
        \hfill
        \begin{subfigure}[b]{0.13\textwidth}
            \centering
            \input{figure/solve_EUFn_2.tex}
            \vspace{-8pt}
            \caption{}
            \label{fig:solve_EUFn_sub2}
        \end{subfigure}
        \hfill
        \begin{subfigure}[b]{0.27\textwidth}
            \centering
            \input{figure/solve_EUFn_3.tex}
            \vspace{-8pt}
            \caption{}
            \label{fig:solve_EUFn_sub3}
        \end{subfigure}
        \hfill
        \begin{subfigure}[b]{0.27\textwidth}
            \centering
            \input{figure/solve_EUFn_4.tex}
            \vspace{-8pt}
            \caption{}
            \label{fig:solve_EUFn_sub4}
        \end{subfigure}
    \end{minipage}}
    \vspace{-8pt}
    \caption{\emph{CCG}-Based Solution of an EUF$^n$ Formula Example}
    \label{fig:solve_EUFn}
    \vspace{-12pt}
\end{figure}

%% file: figure/solve_EUFn_1.tex
\tikzset{
    originnode/.style = {circle, rounded corners, draw = black, fill=white!5,  text centered, align=center, minimum size=8mm, font=\bfseries\Large},
    arrow/.style       = {->, >=latex , draw=black, line width=0.2mm},
    line/.style        = {-, >=stealth, draw=black},
    dashed arrow/.style = {arrow, dashed},
    dashed line/.style  = {line, dashed}
}

{\scriptsize
\begin{tikzpicture}

    \node[originnode] (x){$x$};
    \node[originnode, right of=x, node distance=13mm] (y){$y$};
    \node[originnode, above of=x, node distance=37mm] (z){$z$};
    \node[originnode, above of=y, node distance=37mm] (k){$k$};
    % \node[below of=z, node distance=5mm] (temp){};

    % \draw[arrow] (z.south) --  node[sloped, below, align=center]{$(m \text{>} 0,f,m)$} (x.north);
    \draw[arrow] (z.south) -- node[sloped, above, align=center] {$(m \text{>} 0,f,m)$} (x.north);
    \draw[arrow] (k.south) --  node[sloped, above, align=center]{$(n \text{>} 0,f,n)$} (y.north);

\end{tikzpicture}
}

%% file: figure/solve_EUFn_2.tex
\tikzset{
    originnode/.style = {circle, rounded corners, draw = black, fill=white!5,  text centered, align=center, minimum size=8mm, font=\bfseries\Large},
    arrow/.style       = {->, >=latex , draw=black, line width=0.2mm},
    line/.style        = {-, >=stealth, draw=black},
    dashed arrow/.style = {arrow, dashed},
    dashed line/.style  = {line, dashed}
}

{\scriptsize
\begin{tikzpicture}

    \node[originnode] (x){$x$};
    \node[originnode, right of=x, node distance=13mm] (y){$y$};
    \node[originnode, above of=x, node distance=37mm] (z){$z$};
    \node[originnode, above of=y, node distance=37mm] (k){$k$};
    % \node[below of=z, node distance=5mm] (temp){};

    % \draw[arrow] (z.south) --  node[sloped, below, align=center]{$(m \text{>} 0,f,m)$} (x.north);
    \draw[arrow] (z.south) -- node[sloped, above, align=center] {$(m \text{>} 0,f,m)$} (x.north);
    \draw[arrow] (k.south) --  node[sloped, above, align=center]{$(n \text{>} 0,f,n)$} (y.north);
    \draw[line] (x.east) --  node[above, align=center]{$\mathbf{true}$} (y.west);
\end{tikzpicture}
}

%% file: figure/solve_EUFn_3.tex
\tikzset{
    originnode/.style = {circle, rounded corners, draw = black, fill=white!5,  text centered, align=center, minimum size=8mm, font=\bfseries\Large},
    arrow/.style       = {->, >=latex , draw=black, line width=0.2mm},
    line/.style        = {-, >=stealth, draw=black},
    dashed arrow/.style = {arrow, dashed},
    dashed line/.style  = {line, dashed}
}

{\scriptsize
\begin{tikzpicture}

    \node[originnode] (x){$x$};
    \node[originnode, right of=x, node distance=35mm] (y){$y$};
    \node[originnode, above of=x, node distance=37mm] (z){$z$};
    \node[originnode, above of=y, node distance=37mm] (k){$k$};
    % \node[below of=z, node distance=5mm] (temp){};

    % \draw[arrow] (z.south) --  node[sloped, below, align=center]{$(m \text{>} 0,f,m)$} (x.north);
    \draw[arrow] (z.south) -- node[sloped, above, align=center] {$(m \text{>} 0,f,m)$} (x.north);
    \draw[arrow] (k.south) --  node[sloped, above, align=center]{$(n \text{>} 0,f,n)$} (y.north);
    \draw[arrow] (k.north) .. controls +(135:3mm) and +(45:3mm) .. node[above, sloped, align=center]{$(n \text{>} m \land m \text{>} 0 \land n \text{>} 0, f,n \text{-} m)$} (z.north);
    \draw[arrow] (z.south) .. controls +(-45:3mm) and +(200:3mm) .. node[above, sloped, align=center]{$(m \text{>} n \land m \text{>} 0 \land n \text{>} 0, f,m \text{-} n)$} (k.south);

    \draw[line] (x.east) --  node[above, align=center]{$\mathbf{true}$} (y.west);

    \draw[arrow] (k.south) -- node[sloped, below, pos=0.7]{$(n \text{>} 0 \land \mathbf{true},f,n)$} (x.north);
    \draw[arrow] (z.south) -- node[sloped, above, pos=0.7]{$(m \text{>} 0 \land \mathbf{true},f,m)$}  (y.north);

\end{tikzpicture}
}

%% file: figure/solve_EUFn_4.tex
\tikzset{
    originnode/.style = {circle, rounded corners, draw = black, fill=white!5,  text centered, align=center, minimum size=8mm, font=\bfseries\Large},
    arrow/.style       = {->, >=latex , draw=black, line width=0.2mm},
    line/.style        = {-, >=stealth, draw=black},
    dashed arrow/.style = {arrow, dashed},
    dashed line/.style  = {line, dashed}
}

{\scriptsize
\begin{tikzpicture}

    \node[originnode] (x){$x$};
    \node[originnode, right of=x, node distance=35mm] (y){$y$};
    \node[originnode, above of=x, node distance=37mm] (z){$z$};
    \node[originnode, above of=y, node distance=37mm] (k){$k$};
    % \node[below of=z, node distance=5mm] (temp){};

    % \draw[arrow] (z.south) --  node[sloped, below, align=center]{$(m \text{>} 0,f,m)$} (x.north);
    \draw[arrow] (z.south) -- node[sloped, above, align=center] {$(m \text{>} 0,f,m)$} (x.north);
    \draw[arrow] (k.south) --  node[sloped, above, align=center]{$(n \text{>} 0,f,n)$} (y.north);
    \draw[arrow] (k.north) .. controls +(135:3mm) and +(45:3mm) .. node[above, sloped, align=center]{$(n \text{>} m \land m \text{>} 0 \land n \text{>} 0, f,n \text{-} m)$} (z.north);
    \draw[arrow] (z.south) .. controls +(-45:3mm) and +(200:3mm) .. node[above, sloped, align=center]{$(m \text{>} n \land m \text{>} 0 \land n \text{>} 0, f,m \text{-} n)$} (k.south);

    \draw[line] (x.east) --  node[above, align=center]{$\mathbf{true}$} (y.west);
    \draw[line] (z.east) --  node[above, align=center]{$m \text{=} n \land m \text{>} 0 \land n \text{>} 0$} (k.west);

    \draw[arrow] (k.south) -- node[sloped, below, pos=0.7]{$(n \text{>} 0 \land \mathbf{true},f,n)$} (x.north);
    \draw[arrow] (z.south) -- node[sloped, above, pos=0.7]{$(m \text{>} 0 \land \mathbf{true},f,m)$}  (y.north);

\end{tikzpicture}
}

%% file: section/sec3.tex
\section{The EUF$^n$ Decision Algorithm}  \label{sec:section3}
This section first presents a formal definition of the \emph{conditional congruence graph (CCG)} and then describes the EUF$^n$
decision algorithm built upon it. The algorithm reduces the satisfiability problem of EUF$^n$ formulas to the
satisfiability of existential Presburger arithmetic with divisibility by leveraging the structure of the \emph{CCG}.
Since the satisfiability of such Presburger formulas is decidable~\cite{lechner2015complexity}, the satisfiability of
EUF$^n$ formulas is also decidable.

\subsection{Conditional Congruence Graph}
Similar to the congruence closure algorithm for determining the satisfiability of standard EUF, the decision procedure for EUF$^n$
also models terms in the formula as nodes and functional relationships between terms as directed edges.
% The algorithm first processes all equalities to compute a congruence closure, establishing equivalence relations of terms. Subsequently,
% it verifies whether any disequalities conflict with these derived equivalence relations to assess the satisfiability of the EUF$^n$ formula.
The primary distinction between the two approaches lies in the conditional labels attached to the edges of the
\emph{CCG}. These labels specify the conditions under which functional or equivalence relations hold.
Consequently, when handling equivalence relations, the algorithm maintains conditional equivalence edges rather than merging equivalent nodes.

For simplicity, our analysis in this section is limited to formulas involving unary functions. Nevertheless, the proposed approach can be
readily extended to handle multi-arity functions with minimal modification. Note that while EUF$^n$ does not support parametric
composition of multi-arity functions, constant-depth compositions of such functions, as allowed in standard EUF syntax, are still
supported in EUF$^n$. We begin by formally defining the \emph{CCG}.

\begin{myDef}[Conditional Congruence Graph] \label{def:CCG}
    Let $\mathbb{C}$ denote the set of conditions (logical formulas), $\mathbb{E}$ the set of linear expressions,
    and $\mathbb{F}$ the set of unary function symbols (including finite compositions).

    A \textbf{conditional congruence graph (CCG)} is a triple $(V, E_{\equiv}, E_{\succ})$, where:
    \begin{itemize}
        \item $V$ is the set of nodes representing terms.

        \item $E_{\equiv} \subseteq V \times \mathbb{C} \times V$ is the set of \textbf{conditional equivalence edges}.
        An edge $(u, \phi, v) \in E_{\equiv}$ (undirected) denotes that $u$ and $v$ are equivalent under condition $\phi \in \mathbb{C}$.

        \item $E_{\succ} \subseteq V \times (\mathbb{C} \times \mathbb{F} \times \mathbb{E}) \times V $ is the set of
        \textbf{conditional superterm edges}. An edge $(v, (\phi, f, k), u) \in E_{\succ}$ (directed) denotes that $v = f^k(u)$
        under condition $\phi \in \mathbb{C}$, where $f \in \mathbb{F}$ represents a functional relationship between terms
        and $k \in \mathbb{E}$ specifies the composition depth of the function.
    \end{itemize}
\end{myDef}

For any term $t \in V$, we define $[t]$ as the set containing $t$ itself and all nodes conditionally equivalent to $t$; that is,
for every $t' \in [t] \setminus \{t\}$, there exists an edge $(t, \phi, t') \in E_{\equiv}$. Note that when the directed edges and
conditions on equivalence edges are ignored, the subgraph induced by the nodes in $[t]$ forms a complete graph.

\begin{myDef}[$f \text{-}  term $]     \label{def:define_Fterm}
    Let $t$ be an arbitrary term and let $f \in \mathbb{F}$ be a unary function. The $f\text{-}term$ of $t$, denoted $f\text{-}term(t)$,
    is the set of all terms obtained by iteratively applying $f$ to $t$ (for any number of applications). Formally,
    \[
        f\text{-}term(t) = \bigcup_{k \in \mathbb{E}} \{ f^{k}(t) \}.
    \]
\end{myDef}
For brevity, we use $f\text{-}t$ to denote $f\text{-}term(t)$ in subsequent discussions.

\subsection{Overview of the Algorithm} \label{section3:algorithm}
As discussed earlier, the \emph{CCG} captures both conditional superterm relations and conditional equivalence relations between terms.
Therefore, the algorithm for solving EUF$^n$ formulas based on the \emph{CCG} first processes all equalities in the formula to derive
conditional equivalence relations among terms, maintaining conditional consistency after each equivalence relation update. It then
evaluates the disequalities to verify whether there exists a non-conflicting assignment. If such an assignment exists, instantiating
the original EUF$^n$ formula under this assignment yields a satisfiable EUF instance; otherwise, the EUF$^n$ formula is unsatisfiable.

When processing equations, the proposed algorithm does not immediately merge equivalent nodes, but instead updates the conditional
equivalence and superterm relations between them. As illustrated in Section~\ref{section2:example}, processing each equality in an
EUF$^n$ formula involves the following three steps:
(1) \textbf{Equivalence propagation}: Enforce the transitivity of equality, similar to the \textit{Union} operation in standard EUF solving algorithms.
(2) \textbf{Superterm inheritance}: After updating equivalence relations, propagate the corresponding superterm (and subterm) relationships.
(3) \textbf{Conditional congruence closure}: Update the conditional equivalence relations between congruence terms based on the congruence
axiom and generalized congruence rules (Theorem~\ref{theorem:Gen_Congruence}).
As illustrated in Section~\ref{section2:example}, the above procedure terminates because no term can be its own superterm, and is
correct because the \emph{CCG} completely and accurately captures the congruence graphs of the infinitely many EUF formulas
corresponding to the EUF$^n$ formula.

Given that an EUF$^n$ formula adheres to the syntax in Figure~\ref{fig:syntax}, it can be transformed into disjunctive normal form
(DNF), possibly incurring an exponential size increase \cite{kroening2008decision}. The validity of the formula is then determined
by the satisfiability of its DNF conjuncts, where each conjunct consists of $\mathrm{EUF}^n$-literals (atomic formulas or their negations).
For each conjunct, Algorithm~\ref{algo:solve} either produces a satisfying assignment for the natural number variables or concludes UNSAT.

\input{algorithms/algo_solve.tex}

Algorithm \ref{algo:solve} outlines the process for resolving the target conjunct. In Line \ref{algo:init}, the algorithm initializes
a \emph{CCG}; for example, the graph shown in Figure~\ref{fig:solve_EUFn_sub1} is the resulting initial \emph{CCG}.
The detailed functionality of the initialization procedure, \texttt{Initialize}, will be elaborated in Section \ref{section3:init}.
In Line~\ref{algo:worklist_init}, the queue $W$ is initialized to track conditional equivalence edges on which conditional congruence
closure has not yet been applied. In Lines \ref{algo:eq_start}-\ref{algo:eq_end}, the algorithm processes equalities in
$\Psi$ (\emph{e.g.}, $term_l \text{=} term_r$). Lines~\ref{algo:call_EqProp} perform \textbf{equivalence propagation} by invoking
the \texttt{EqProp} function. This function first executes Line~\ref{algo:call_AddEq}, which calls \texttt{AddEq} to insert a
conditional equivalence edge between $term_l$ and $term_r$. As illustrated in the example in Section~\ref{section2:example}, an
equivalence edge is added between $x$ and $y$ in Figure~\ref{fig:solve_EUFn_sub2}.

% Note that $W$ is updated exclusively within the
% \texttt{EqProp} function (Line~\ref{algo:push_W}), where the edge corresponding to the parameters of \texttt{EqProp} is pushed
% into $W$. The parameters for \texttt{EqProp} are sourced solely from the equations in formula $\Psi$ (Line~\ref{algo:eq_start})
% and the results of \texttt{CongruenceClosure} (Line~\ref{algo:Congruence}).

Next, Lines~\ref{algo:tran_start}--\ref{algo:tran_end} iteratively invoke \texttt{AddEq} for $term_1$ and $ term_2$ where $term_1$
and $term_2$ are connected to $term_l$ and $term_r$ via conditional equivalence edges, respectively,
\emph{i.e.}, $term_1 \in [term_l] \setminus \{term_l\}$ and $term_2 \in [term_r] \setminus \{term_r\}$. These calls update
equivalence edges between $term_1$ and $term_2$, establishing their equivalence under condition $\phi_{\text{new}} \land \phi_1 \land \phi_2$.
% Note that during loop execution, the pre-loop state of $E_{\equiv}$ can be stored in a separate variable, and only this preserved
% snapshot is utilized throughout the loop iterations.
Note that during loop execution, a pre-loop snapshot of $E_{\equiv}$ is saved and used exclusively throughout all iterations,
thereby ensuring termination. In the example from Section~\ref{section2:example}, since $[x] = \{x\}$ and $[y] = \{y\}$,
no new equivalence edges are added at this step.
The \texttt{AddEq} function (Lines~\ref{algo:AddEq}–\ref{algo:AddEq_end} in Algorithm~\ref{algo:solve}) incorporates a new
conditional equivalence edge $(term_l, \phi_{\text{new}}, term_r)$ into the \emph{CCG}. It first checks whether a
conditional equivalence edge already exists between $term_l$ and $term_r$ (Line~\ref{algo:find}); if not, $\phi_{\text{old}}$
is set to $\mathbf{false}$. Line~\ref{algo:eq_add} updates the equivalence condition to $\phi_{\text{old}} \lor \phi_{\text{new}}$. Note that after adding a conditional equivalence edge, the \texttt{AddEq} function does not recursively call itself to propagate the equivalence further. This is because $[term_1] = [term_l]$ and $[term_2] = [term_r]$ (Line~\ref{algo:tran_start}), so no additional iterations are required to establish equivalence between $[term_1]$ and $[term_2]$, thus ensuring termination.

Since the equivalence condition has changed, Line~\ref{algo:inherit} invokes the \texttt{InheritST} function to perform
\textbf{superterm inheritance}, updating the superterms/subterms of $term_l$ and $term_r$ along with their associated conditions. That is, for every $(term_1, (\phi_1, f_1, k_1), term_l) \in E_{\succ}$ and $(term_l, (\phi_2, f_2, k_2), term_2) \in E_{\succ}$, add or update $(term_1, (\phi_1 \wedge \phi_{\text{new}}, f_1, k_1), term_r)$ and $(term_r, (\phi_2 \wedge \phi_{\text{new}}, f_2, k_2), term_2)$ to $E_{\succ}$. Moreover, for any other superterm $term_s$ of $term_r$ with function $f_1$, the superterm edges between $term_1$ and $term_s$ are updated accordingly; similarly, for any other superterm $term_s$ of $term_2$ with function $f_2$, the superterm edges between $term_r$ and $term_s$ are updated accordingly. Finally, the superterms and subterms of $term_r$ are similarly propagated to $term_l$. In the example from Section~\ref{section2:example}, Figure~\ref{fig:solve_EUFn_sub3} shows the resulting graph after executing \texttt{InheritST}($x$, true, $y$). It can be seen that the superterm $z$ of $x$ also becomes a superterm of $y$, and the conditional superterm relation between $z$ and $y$'s original superterm $k$ is also added. Similarly, $y$'s superterm $k$ is inherited by $x$ in the same manner.

In the while loop from Lines~\ref{algo:begin_while}–\ref{algo:eq_end}, Line~\ref{algo:worklist_pop} pops an unprocessed edge from $W$.
Lines~\ref{algo:Congruence}–\ref{algo:eq_end} then perform the \textbf{conditional congruence closure}, deriving new
equivalences between congruent terms via the \texttt{CongruenceClosure} procedure, whose implementation is described
in Section~\ref{section3:congruence}. The resulting conditional equivalences are collected in the edge set $\widehat{E}_{\equiv}$
(Line~\ref{algo:Congruence}) and propagated via the \texttt{EqProp} function. In Section~\ref{section2:example}, after processing
$x=y$, we have $\widehat{E}_{\equiv} = \{(z, m=n \land m>0 \land n>0, k)\}$. As a result, a conditional equivalence edge between
$z$ and $k$ is added via \texttt{EqProp}, as illustrated in Figure~\ref{fig:solve_EUFn_sub4}. Furthermore, the triple
$(z, m=n \land m>0 \land n>0, k)$ is pushed into $W$. Note that the equivalence relations added to the \emph{CCG} via \texttt{AddEq}
(Lines~\ref{algo:tran_start}--\ref{algo:tran_end}), which are of the form $(term_1, \phi_{\text{new}} \land \phi_1 \land \phi_2, term_2)$
where $term_1 \in [term_l] \setminus \{term_l\}$ and $term_2 \in [term_r] \setminus \{term_r\}$, are not considered for further
congruence closure.
% Note that the equivalence relations added to the CCG via \texttt{AddEq}
% (Lines~\ref{algo:tran_start}--\ref{algo:tran_end}) are not considered for further congruence closure. These relations are of
% the form $(term_1, \phi_{\text{new}} \land \phi_1 \land \phi_2, term_2)$, where $term_1 \in [term_l] \setminus \{term_l\}$ and
% $term_2 \in [term_r] \setminus \{term_r\}$.
This is because $term_1$ (respectively, $term_2$) shares identical superterms with
$term_l$ (respectively, $term_r$). Consequently, for the equivalent condition
$\phi \vee \phi_{\text{new}} \vee (\phi_{\text{new}} \land \phi_1 \land \phi_2)$ between congruence terms obtained by
considering both edges, the condition introduced by the edge $(term_1, \phi_{\text{new}} \land \phi_1 \land \phi_2, term_2)$ is redundant.
(For brevity, we omit the superterm conditions since the conclusion holds when they are considered.)

The algorithm then processes the disequalities in $\Psi$ (Lines~\ref{algo:neq_start}--\ref{algo:neq_end}). For each disequality
$term_l \neq term_r$ (Line~\ref{algo:neq_start}), the algorithm identifies from the \emph{CCG} the condition $\phi$ under which
$term_l$ and $term_r$ become equivalent (Line~\ref{algo:find_eq}). Since a conflict arises if $\phi$ holds, its negation $\neg \phi$ is
conjoined to the global condition $\Phi$ (Line~\ref{algo:neq_end}). After processing all disequalities, $\Phi$ encapsulates the constraints
necessary to prevent conflicts. Thus, if $\Phi$ is satisfiable, the original $\mathrm{EUF}^n$ formula is satisfiable; otherwise, no satisfying
assignment exists (Line~\ref{algo:solve_cond}). In the example from Section~\ref{section2:example}, when processing $z \neq k$, the equivalence
condition between $z$ and $k$ is found to be $\phi = (m=n \land m>0 \land n>0) \vee (m=0 \land n=0)$ (Figure~\ref{fig:solve_EUFn_sub4}, which omits
some conditions for brevity as noted earlier). Since $\neg \phi$ is satisfiable, the $\mathrm{EUF}^n$ formula in Section~\ref{section2:example}
is satisfiable. Note that $\Phi$ can be expressed as an existential sentence in Presburger arithmetic with divisibility, a theory with
decidable satisfiability. A formal justification for this characterization is provided in Section~\ref{section3:congruence} and
Section~\ref{section3:correct}.

\subsection{Initialization} \label{section3:init}
This subsection details the initialization function \texttt{Initialize} (Line~\ref{algo:init}) in Algorithm~\ref{algo:solve}. During initialization, all terms from the EUF$^n$ formula $\Psi$ are collected into set $V$, which serves as nodes of the \emph{CCG}. The superterm edge set $E_{\succ}$ is initialized by defining relations where one term can be obtained by applying functions to another (\emph{e.g.}, $f^m(x)$ becomes a superterm of $f^n(x)$ via $(m \text{-} n)$ applications of $f$, under condition $m>n$). The sets $V$ and $E_{\succ}$ are initialized as follows:
% \begin{alignat*}{2}
%     &  V &\  = \  &\bigl\{ term \mid term \in \Psi \bigr\} \\
%     &  E_{\succ} & \ =\  &\bigl\{ (f^{k_2}(term), (k_2 > k_1, f, k_2 - k_1), f^{k_1}(term)) \\
%     & && \mid term \in V, f \in \mathbb{F}, f^{k_1}(term) \in V, f^{k_2}(term) \in V \bigr\} \\
%     &  E_{\equiv} & \ =\  &\bigl\{ (term_1, k_1 = k_2, term_2) \\
%     & && \mid term, term_1, term_2 \in V,  f \in \mathbb{F}, \\
%     & &&  (term_1, (k_1 > 0, f, k_1),term) \in E_{\succ}, \\
%     & &&  (term_2, (k_2 > 0, f, k_2),term) \in E_{\succ}  \bigr\}
% \end{alignat*}
$$
    \frac{
         term \in \Psi
    }{
        V \leftarrow V \cup \{term\}
    }
\quad \quad
    \frac{
        f^{k_1}(term), f^{k_2}(term), term \in V, \quad f \in \mathbb{F}
    }{
    %  (f^{k_2}(term), (k_2 > k_1, f, k_2 - k_1), f^{k_1}(term)) \in E_{\succ}
     E_{\succ} \leftarrow E_{\succ} \cup \{(f^{k_2}(term), (k_2 > k_1, f, k_2 - k_1), f^{k_1}(term))\}
    }
$$
This rule states that if the conditions above the line hold, then the operations below the line are executed. Since the syntax of EUF$^n$ permits function composition ($f = f_1 \circ \cdots \circ f_n$; Section~\ref{section2:formal_syntax}), each addition of a superterm edge in the \emph{CCG} (\emph{e.g.}, for function relation $f_i$) triggers an immediate check. This check determines if new superterm edges (\emph{e.g.}, for $f$) can be added based on function composition rules and the edge's predecessor/successor relationships, with corresponding superterm conditions and composition counts derived from intermediate edges. For example, if $x = f_1(y)$, $y = f_2(z)$, and $f = f_1 \circ f_2$, then a superterm edge representing $x = f(z)$ is added. This composition check is intrinsic to all superterm edge additions in subsequent algorithm steps.

Finally, the conditional equivalence edge set $E_{\equiv}$ is initialized with conditional equivalence relations between terms that
are derivable from identical base terms through identical function applications of matching depth (\emph{i.e.}, according to Theorem~\ref{theorem:Gen_Congruence} and the congruence axiom~\ref{axiom4}). The initialization rules are as follows:
\[
    \frac{
        \begin{array}{c}
            (term_1, (k_1 > 0, f, k_1),term) \in E_{\succ}, \quad term, term_1 \in V, \\
            (term_2, (k_2 > 0, f, k_2),term) \in E_{\succ}, \quad term_2 \in V, f \in \mathbb{F}
        \end{array}
    }{
        % (term_1, k_1 = k_2, term_2) \in E_{\equiv}
        E_{\equiv} \leftarrow E_{\equiv} \cup \{(term_1, k_1 = k_2, term_2)\}
    }
\]
After initializing the equivalence edges in $E_{\equiv}$, no further congruence closure processing is required. This is because
the superterms of $term_1$ and $term_2$ are also superterms of $term$, and their conditional equivalence relations are already
included in $E_{\equiv}$.

\input{section/sec3_con.tex}

%% file: algorithms/algo_solve.tex
\begin{algorithm}[ht!]
    \caption{Solve}\label{algo:solve}
    \KwIn{A finite conjunction $\Psi$ of $\mathrm{EUF}^n$-literals.}
    \KwOut{A solution or UNSAT.}
    \SetKwFunction{FSubFunc}{AddEq}
    \SetKwFunction{FEqFunc}{EqProp}
    \SetKwProg{Fn}{Function}{:}{}

    Let \emph{CCG} $(E, E_{\equiv}, E_{\succ} ) \leftarrow \texttt{Initialize}(\Psi) $    \tcp*[f]{\textcolor{silvergrey}{Section \ref{section3:init}}}   \label{algo:init}\\
    Let $W \leftarrow \emptyset $    \tcp*[f]{\textcolor{silvergrey}{A global stack-based worklist} }   \label{algo:worklist_init}\\
    \ForEach{$term_{l} = term_{r} \in \Psi$}{                                             \label{algo:eq_start}
        % \If{$ \texttt{Identical} (term_{l}, term_{r} )$}{                                              \label{algo:identical}
        %     continue;
        % }
        \tcp{\textcolor{silvergrey}{Equivalence propagation}}
        $ \FEqFunc(term_{l}, \mathbf{true}, term_{r})$                                               \label{algo:call_EqProp}  \\
        % \ForEach{ $(term_{l}, \phi_1, term_1) \in E_{\equiv}, (term_{r}, \phi_2, term_2) \in  E_{\equiv} $      \label{algo:tran_start} }
        % { $ \FSubFunc (term_1, \ \phi_1 \wedge \phi_2,\ term_2) $               \label{algo:tran_end}  }

        \While{$W \neq \emptyset$}{                                          \label{algo:begin_while}
            $ (term^{\prime}_{l}, \phi^{\prime}, term^{\prime}_{r}) \leftarrow W.pop() $        \label{algo:worklist_pop}          \\
            \tcp{\textcolor{silvergrey}{Conditional congruence closure}}
            $ \widehat{E}_{\equiv} \leftarrow \texttt{CongruenceClosure}(term^{\prime}_{l}, \phi^{\prime}, term^{\prime}_{r} )$
            \tcp*[f]{\textcolor{silvergrey}{Section \ref{section3:congruence}}}   \label{algo:Congruence}          \\
            \ForEach{ $(term^{\prime}_1, \phi^{\prime\prime}, term^{\prime}_2) \in \widehat{E}_{\equiv}$}{     \label{algo:AddEq_congruence}
                    $\FEqFunc(term^{\prime}_1, \phi^{\prime\prime}, term^{\prime}_2) $                             \label{algo:eq_end}  }
        }
      }
    $\Phi \leftarrow \mathbf{true}$

    \ForEach{$ \neg(term_{l}=term_{r}) \in \Psi $}{ \label{algo:neq_start}
      \If{$(term_{l}, \phi, term_{r}) \in E_{\equiv}$}{ \label{algo:find_eq}
        $\Phi \leftarrow \Phi \wedge (\neg \phi)$  \label{algo:neq_end}
      }
    }
    $result \leftarrow$  \texttt{SolvePresburgerWithD}($\Phi$)                                          \label{algo:solve_cond}  \\
    \KwRet $result$
    \BlankLine

    \Fn{\FEqFunc {$term_{l}, \phi_{\textup{new}}, term_{r}$}}{
        $ \FSubFunc (term_{l}, \phi_{\text{new}}, term_{r}) $                           \label{algo:call_AddEq} \\
        $W.push((term_{l}, \phi_{\text{new}}, term_{r} ))$                                 \label{algo:push_W}       \\
        \tcp{\textcolor{silvergrey}{$term_1 \in [term_{l}] \setminus \{term_l\}, term_2 \in [term_{r}] \setminus \{term_r\}$}}
        \ForEach{ $(term_{l}, \phi_1, term_1) \in E_{\equiv}, (term_{r}, \phi_2, term_2) \in  E_{\equiv} $     \label{algo:tran_start} }
               { $ \FSubFunc (term_1, \ \phi_{\text{new}} \land \phi_1 \wedge \phi_2,\ term_2) $     \label{algo:tran_end}   }
    }
    \Fn{\FSubFunc {$term_{l}, \phi_{\textup{new}}, term_{r}$}}{                                    \label{algo:AddEq}
        % \If{\texttt{Identical}$\textup{(}term_{l}, term_{r}\textup{)}  $                                             \label{algo:identical} }
        %     {\textbf{return}}
        \If{$(term_{l}, \phi_{\text{old}}, term_{r}) \in E_{\equiv}$}{            \label{algo:find}
              % \tcp*[f]{\textcolor{silvergrey}{$\phi_{\text{old}} \gets \mathbf{false}$ if no such edge exists.}}        \label{algo:find} \\

              % \If{$(\phi_{\text{old}} \rightarrow \phi_{\text{new}})$                                    \label{algo:imply}  }
              %    {return; }                                                                            \label{algo:ret}
              $E_{\equiv}  \leftarrow  E_{\equiv}  \backslash  \{(term_{l}, \phi_{\text{old}}, term_{r})\}  \cup \{(term_{l}, \phi_{\text{old}} \vee \phi_{\text{new}}, term_{r} )\} $    \label{algo:eq_add} \\
              \tcp{\textcolor{silvergrey}{Superterm inheritance}}
              \texttt{InheritST}($term_{l}, \phi_{\textup{new}}, term_{r} $)        \label{algo:inherit}
        \\}                                                                                                   \label{algo:AddEq_end}
  }
\end{algorithm}

%% file: section/sec3_con.tex
% \section{Congruence Relations Processing} \label{sec:section4}
% This section presents the implementation details of the \texttt{CongruenceClosure} function. We begin by introducing the concept of
% \emph{periodic equivalence}, followed by a discussion on the processing of congruence relations, with particular attention to handling
% congruence closure under periodic equivalence. Finally, we prove the correctness and analyze the complexity of Algorithm~\ref{algo:solve}.

\subsection{Periodic Equivalence} \label{section3:periodic_equivalence}
Before introducing the function \texttt{CongruenceClosure} (Line~\ref{algo:Congruence}) that handles congruence relations in Algorithm~\ref{algo:solve}, we first introduce a phenomenon we refer to as \emph{periodic equivalence}. In the standard congruence closure algorithm for EUF formulas, after performing the \textit{Union}$(x, y)$ operation, the algorithm recursively merges congruent terms by proceeding upward and executing \textit{Union}$(f(x), f(y))$, \textit{Union}$(f(f(x)), f(f(y)))$, and so on. The resulting equivalence classes can be categorized into two types based on whether a superterm relationship exists between $x$ and $y$, as illustrated in Figure~\ref{fig:periodic_eq_intro}, where shaded nodes belong to the same equivalence class. Figure~\ref{fig:periodic_subfig1} depicts the case where no superterm relationship exists between $x$ and $y$; that is, the upward search for congruent terms starting from $x$ does not reach $y$ (and vice versa). We refer to the equivalence classes formed in this case as having a \textbf{hierarchical structure}. Figure~\ref{fig:periodic_subfig2} illustrates the case where a superterm relationship exists between $x$ and $y$, such as in \textit{Union}$(x, f^2(x))$. Here, the upward search from $x$ eventually reaches $f^2(x)$, causing equivalence classes from different levels to merge, forming what we call a \textbf{partitioned structure}.

\begin{figure}[ht]
    \centering
    \input{figure/periodic_eq.tex}
    \vspace{-8pt}
    \caption{Impact of superterm relationships on equivalence class formation.}
    \label{fig:periodic_eq_intro}
    \vspace{-8pt}
\end{figure}
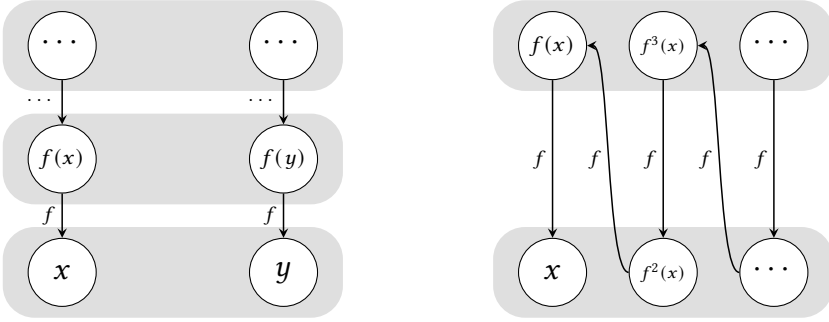

In EUF, where the depth of function applications is bounded, both cases in Figure~\ref{fig:periodic_eq_intro} can be handled uniformly
by recursively merging until a fixed point is reached. However, in EUF$^n$, where the composition depth of functions may be
represented by variables, this recursive approach fails. Thus, we must treat the two cases separately. For instance, in the
hierarchical structure (Figure~\ref{fig:periodic_subfig1}), we have $x = y \land m = n \rightarrow f^n(x) = f^m(y)$, while in the
partitioned structure (Figure~\ref{fig:periodic_subfig2}), we have $m \equiv_{2} n \rightarrow f^n(x) = f^m(x)$ (for simplicity, we
use $m \equiv_{2} n$ to denote $m \equiv n \mod 2$, and adopt this notation throughout). We refer to the equivalence relation
arising from a superterm relationship between $x$ and $y$ as \emph{periodic equivalence}. This section elaborates on the
treatment of such congruence relations in solving EUF$^n$ formulas.

% Periodic Equivalence
\begin{myDef} [Periodic Equivalence]         \label{def:define_periodic}
    Given an equivalence pair $(term_l, term_r)$ where $term_l$ and $term_r$ are distinct terms,
    if
    $$ \exists f \in \mathbb{F}. \ term_l \in f \text{-} term_r \vee term_r \in f \text{-} term_l $$
    holds, then $(term_l, term_r)$ is called a periodic equivalence pair and $term_l$ and $term_r$ are considered periodically equivalent.
\end{myDef}
% When $term_l$ and $term_r$ are periodically equivalent, there exists a term $t$ such that
% $ term_l \in f \text{-} t \wedge term_r \in f \text{-} t$, and we also refer to this situation as a periodic
% equivalence relation with respect to $t$.

\begin{myExample}  \label{example1}
   Suppose the equation to be processed is $x = f^{2}(x)$. By definition, $x$ and $f^{2}(x)$ are periodically equivalent, as shown
   in Figure~\ref{fig:periodic_subfig2}, this equivalence relation partitions the set $f$-$x$ into two equivalence classes:
   $\{x, f^{2}(x), \cdots\}$ and $\{f(x), f^{3}(x), \cdots\}$, each containing infinitely many elements. Therefore, after adding
   the equivalence edge $(x, \text{True}, f^{2}(x))$ to the \emph{CCG}, for any two elements $f^{n}(x) \in f\text{-}x$ and
   $f^{m}(x) \in f\text{-}x$ that belong to the same equivalence class, an equivalence edge $(f^{n}(x), m \equiv_{2} n, f^{m}(x))$
   should be added.
\end{myExample}
\begin{myExample} \label{example2}
   Assume we are processing the equation $f^{6}(x)=f^{8}(x)$, having previously processed the equation $x = f^{5}(x)$.
   Consequently, an equivalence edge with condition $6 \equiv_{5} 8$ already exists between nodes $f^{6}(x)$ and $f^{8}(x)$ in
   the \emph{CCG}. Since $\neg(6 \equiv_{5} 8)$ holds, the equivalence class derived from the relation $(x, f^{5}(x))$ is distinct
   from that derived from $(f^{6}(x), f^{8}(x))$. Given
    \begin{equation*}
    \left(
        \begin{aligned}
            x = f^{5}(x) & \Rightarrow x = f^{5}(x) = f^{10}(x) \\
            f^{6}(x) = f^{8}(x) & \Rightarrow f^{6}(x) = f^{8}(x) = f^{10}(x)
        \end{aligned}
    \right)
    \Rightarrow
    x = f^{5}(x) = f^{6}(x) \Rightarrow x = f (x)
\end{equation*}
   where the first implication follows from applying $f^5$ to $x = f^5(x)$, yielding $f^5(x) = f^{10}(x)$, the second follows
   similarly, and the last from $f(x)=f(f^5(x))$. We therefore conclude that the equivalence classes derived from the
   relations $(x, f^{5}(x))$ and $(f^{6}(x), f^{8}(x))$ are identical to those derived from $(x, f(x))$. A more general
   and formal treatment will be provided in Section~\ref{section3:congruence}.
\end{myExample}

\subsection{Congruence Relations Processing} \label{section3:congruence}
In this subsection, we explain how to handle congruence relations after adding an equivalence edge, \emph{i.e.}, the \texttt{CongruenceClosure} function (Line~\ref{algo:Congruence}) in Algorithm~\ref{algo:solve}. Suppose the added equivalence edge is $\boldsymbol{( \mathit{term}_l, \phi, \mathit{term}_r )}$. We will discuss the cases based on whether a \emph{periodic equivalence} exists between $term_l$ and $term_r$. During this process, we use $\widehat{E}_{\equiv}$ to record the conditional equivalence relations derived from handling the congruence relations. The computation rules presented later in this subsection state that if the conditions above the line hold, then the operations below the line are executed.

\begin{description}
    \item[$term_{l}$ and $term_{r}$ are not periodically equivalent.]

    We first consider the case where the newly added equivalence relation is not a \emph{periodic equivalence}, meaning there is no superterm relation between $term_{l}$ and $term_{r}$. Since $term_{l}$ and $term_{r}$ are equivalent, according to the congruence axiom~\ref{axiom4} and Theorem~\ref{theorem:Gen_Congruence}, for any $f \in \mathbb{F}$, we need to identify terms within $f \text{-} term_{l}$ and $f \text{-} term_{r}$ that can be derived as equivalent. The specific rule for adding equivalence relations is as follows. The condition $\phi_1 \wedge \phi_2$ ensures that both superterm edges can hold simultaneously, prevents an infinite upward search for
    congruent terms, and thereby guarantees algorithm termination. The validity of the equivalence condition between terms is ultimately checked during the final satisfiability check of the formula.
    \begin{equation} \label{formula1}
        \frac{
            \begin{array}{c}
                (term_l, \phi, term_r) \in E_{\equiv}, \quad
                (term_1, (\phi_1, f, \mathit{k}_1), term_l) \in E_{\succ}, \\
                (term_2, (\phi_2, f, \mathit{k}_2), term_r) \in E_{\succ}, \quad
                (term_1, \phi_{\text{old}}, term_2) \in E_{\equiv}, \quad
                \widehat{\phi} = \phi_1 \land \phi_2 \ \text{is satisfiable}.
            \end{array}
        }{
            % (term_1, (\phi \land \widehat{\phi} \land \mathit{k}_1 = \mathit{k}_2) \lor \phi_{\text{old}}, term_2) \in \widehat{E}_{\equiv}
            \widehat{E}_{\equiv} \leftarrow \widehat{E}_{\equiv} \cup \{(term_1, (\phi \land \widehat{\phi} \land \mathit{k}_1 = \mathit{k}_2) \lor \phi_{\text{old}}, term_2)\}
        }
    \end{equation}

    \item[$term_{l}$ and $term_{r}$ are periodically equivalent.]
    When a superterm relationship exists between $term_l$ and $term_r$ (assuming $term_r = f^k(term_l)$ for $k > 0$), the addition
    of an equivalence relation partitions the set $f\text{-}term_{l}$ into $k$ equivalence classes. First, for elements outside
    the set $f \text{-} term_{l}$, we handle them according to Rule~(\ref{formula1}); that is, for any $g \in \mathbb{F}$ with
    $g \neq f$, the elements in $\widehat{E}_{\equiv}$ follow Rule~(\ref{formula1}) with the function symbol $f$ replaced by $g$.
    Next, we discuss the elements in the set $f \text{-} term_{l}$, as discussed in Examples \ref{example1} and \ref{example2},
    we will consider two cases: whether a \emph{periodic equivalence} involving $term_l$ and $term_r$ has already been processed.

    \begin{enumerate}
    \item \textbf{ $(term_{l}, term_{r})$ is the first added \emph{periodic equivalence} pair with respect to $f \text{-} term_l$.} \\
        For elements in $f \text{-} term_{l}$, any two elements that belong to the same equivalence class are considered equivalent
        and processed according to the following rules:
        \begin{equation} \label{formula3}
            \frac{
                \begin{array}{c}
                    (term_l, \phi, term_r) \in E_{\equiv}, \quad
                    (term_{r}, (\phi_1, f, k), term_{l})  \in   E_{\succ}, \\
                    (term_{1}, (\phi_2, f, k_1), term_{l})  \in   E_{\succ}, \quad
                    (term_{2}, (\phi_3, f, k_2), term_{l})  \in   E_{\succ}, \\
                    (term_{1}, \phi_{\text{old}}, term_{2})   \in   E_{\equiv}, \quad
                    \widehat{\phi} = \phi_1 \wedge \phi_2 \wedge \phi_3 \ \text{is satisfiable}.
                \end{array}
            }{
            %   (term_1, ( \phi \wedge \widehat{\phi} \wedge k_1 \equiv_{k} k_2  )  \vee \phi_{\text{old}}, term_2) \in \widehat{E}_{\equiv}
              \widehat{E}_{\equiv} \leftarrow \widehat{E}_{\equiv} \cup \{(term_1, ( \phi \wedge \widehat{\phi} \wedge k_1 \equiv_{k} k_2  )  \vee \phi_{\text{old}}, term_2)\}
            }
        \end{equation}

    Note that, modular arithmetic can be expressed with an existential sentence of Presburger arithmetic with divisibility, as shown below:
    \begin{align*}
         m \equiv_{k} n  \Leftrightarrow  \exists c.\ k \mid (m-c) \wedge k \mid (n-c) \wedge c < k \\
         \neg ( m \equiv_{k} n ) \Leftrightarrow  \exists c. \ k \mid (m-c) \wedge \neg (k \mid (n-c)) \wedge c < k
    \end{align*}

    \item  \textbf{The \emph{periodic equivalence} relation $(\mathbf{term_{l^{\prime}}, term_{r^{\prime}}})$
    with respect to $f \text{-} term_{l}$ has been processed.} In this case, an equivalence edge already exists between $term_{l}$ and
    $term_{r}$. The corresponding \emph{periodic equivalence} pair $(term_{l'}, term_{r'})$ responsible for introducing this edge
    can be directly retrieved from the algorithm's execution trace.
    Consequently, the newly introduced \emph{periodic equivalence} $(term_{l}, term_{r})$ will further refine the equivalence classes
    partitioned by $(term_{l'}, term_{r'})$. In Theorem~\ref{theorem:th1}, we prove that multiple \emph{periodic equivalences} can
    be represented by a single \emph{periodic equivalence}, and both yield identical equivalence class partitions.
    Without loss of generality, we assume $term_{l'} = f^{k_1}(term_{l})$ and $term_{r'} = f^{k_2}(term_{l})$, with $k_2 > k_1$.
    The specific processing rules are as follows, where $\gcd$ denotes the greatest common divisor function:

    \vspace{-6pt}
    \begin{equation} \label{formula4}
        \frac{
            \begin{array}{c}
                (term_l, \phi, term_r) \in E_{\equiv}, \quad
                (term_{r}, (\phi_1, f, k), term_{l})  \in   E_{\succ}, \\
                (term_{l^{\prime}}, (\phi_2, f, k_1),term_{l})  \in   E_{\succ}, \quad
                (term_{r^{\prime}}, (\phi_3, f, k_2), term_{l})  \in   E_{\succ}, \\
                (term_{r^{\prime}}, (\phi_4, f, k_2-k_1), term_{l^{\prime}}) \in   E_{\succ}, \quad
                (term_{1}, (\phi_5, f, k_3), term_{l})  \in   E_{\succ}, \\
                (term_{2}, (\phi_6, f, k_4), term_{l})  \in   E_{\succ}, \quad
                (term_{1}, \phi_{\text{old}}, term_{2})   \in   E_{\equiv}, \\
                \widehat{\phi} = \phi_1 \wedge \phi_2 \wedge \phi_3 \wedge \phi_4  \wedge \phi_5  \wedge \phi_6 \ \text{is satisfiable}.
            \end{array}
        }{
        %   (term_1, (\phi \wedge \widehat{\phi} \wedge k_3 \equiv_{ \mathsf{gcd} (k,\ (k_2-k_1)) } k_4)
        %   \vee \phi_{\text{old}}, term_2) \in \widehat{E}_{\equiv}
          \widehat{E}_{\equiv} \leftarrow \widehat{E}_{\equiv} \cup \{(term_1, (\phi \wedge \widehat{\phi} \wedge k_3 \equiv_{ \mathsf{gcd} (k,\ (k_2-k_1)) } k_4)  \vee \phi_{\text{old}}, term_2)\}
        }
    \end{equation}

    The correctness of this rule is established by Theorem~\ref{theorem:th1} in Section~\ref{section3:correct}.
    Noted that, according to Bézout's identity~\cite{rosen1999discrete}, the $\mathsf{gcd}$ function can be expressed
    using an existential sentence of Presburger arithmetic with divisibility, as shown below:
        $$a = \mathsf{gcd}(m,n) \Leftrightarrow  \exists b.\  a \mid m \wedge a \mid n \wedge m \mid b \wedge n \mid (b-a) $$

    \end{enumerate}
\end{description}

Above, we described the handling of \texttt{CongruenceClosure()}, which is invoked after each addition of an equivalence
relation through the \texttt{EqProp} function. The resulting conditional equivalences are stored in $\widehat{E}_{\equiv}$
and subsequently added to the \emph{CCG}.

\begin{myThe}[Decidability] \label{theorem:decidability}
    For any finite conjunctive formula $\Psi$ of EUF$^n$ equalities and disequalities, its satisfiability is decidable via the \emph{CCG} algorithm.
\end{myThe}

\begin{proof}
    For an arbitrary formula $\Psi$, the \emph{CCG} algorithm transforms it into $\Phi$ such that $\Psi$ is satisfiable if and only if
    $\Phi$ is satisfiable. The formula $\Phi$ is an existential sentence of Presburger arithmetic with divisibility. This is because
    Algorithm~\ref{algo:solve} introduces existential quantifiers only through the $\mathsf{modulo}$ and $\mathsf{gcd}$ operations.
    Although Line~\ref{algo:neq_end} involves negation, it introduces no universal quantifiers, as both $m \equiv_{a} n$ and its negation
    are expressible as existential sentences. The $\mathsf{gcd}$ function, having no negation semantics, similarly avoids universal
    quantification. Therefore, the formula $\Phi$ produced by Algorithm~\ref{algo:solve} is an existential sentence of Presburger
    arithmetic with divisibility. Since this logic is decidable~\cite{lechner2015complexity}, the satisfiability of $\Psi$ is also decidable.
\end{proof}

\subsection{Correctness and Complexity}  \label{section3:correct}
% \subsubsection{Correctness}
To prove the correctness of Algorithm \ref{algo:solve}, we first demonstrate its termination.
\begin{myThe}[Termination] \label{theorem:termination}
    Algorithm \ref{algo:solve} terminates on any finite input conjunct $\Psi$, where $\Psi$ is a conjunction of EUF$^n$ equalities and disequalities.
\end{myThe}
\begin{proof}
    Since the size of the $\Psi$ is finite, it suffices to consider whether
    the processing of each individual formula terminates. Lines~\ref{algo:neq_start}–\ref{algo:neq_end} handle disequalities by collecting the
    associated equivalence conditions of the involved terms, which is clearly a terminating process.

    For equalities, the processing involves three steps. The first step, \textbf{equivalence propagation} (Line~\ref{algo:call_EqProp}),
    updates the conditional equivalence relations between elements in $[term_l]$ and $[term_r]$. As the number of nodes in $[term_l]$
    and $[term_r]$ is finite, this step terminates. Following the insertion of this equivalence edge, the second step,
    \textbf{superterm inheritance}, is performed. This step terminates because the number of superterms and subterms of $term_l$ and
    $term_r$ is finite. The third step, \textbf{conditional congruence closure}, derives new equivalences based on the congruence axiom~\ref{axiom4} and the generalized congruence theorem (Theorem~\ref{theorem:Gen_Congruence}), which may in turn trigger additional congruence-based deductions. However, since the number of nodes is finite and no term can be a superterm of itself (as such conditions are unsatisfiable), no infinite superterm chains can arise. Thus, this process is guaranteed to terminate. Therefore, we conclude that
    Algorithm~\ref{algo:solve} terminates.
\end{proof}

We now prove the soundness and completeness of Algorithm~\ref{algo:solve}, namely, that the input EUF$^n$ formula $\Psi$ is unsatisfiable
if and only if the existential Presburger formula $\Phi$ produced by the algorithm is unsatisfiable. To establish this, we
first prove Lemma~\ref{lemma:correctness}, as stated below:
\begin{myLem} \label{lemma:correctness}
    After Algorithm~\ref{algo:solve} completes, for every equivalence edge $(term_l, \phi, term_r)$ in the \emph{CCG} and any assignment
    $\sigma: \mathsf{N}\to\mathbb{N}$, the following invariant holds: $\phi$ is satisfiable under $\sigma$
    if and only if $term_l$ and $term_r$ belong to the same equivalence class in the congruence closure of $\Psi[\sigma]$.
\end{myLem}
\begin{proof}
    We prove by induction that after processing each equation, the \emph{CCG} and the congruence closure of $\Psi[\sigma]$ obtained
    after this processing satisfy the invariant stated in Lemma~\ref{lemma:correctness}. The detailed proof can be found in Appendix~\ref{sec:appendix_correct}; below we outline the main ideas.

    First, during Algorithm~\ref{algo:solve}'s initialization (Section~\ref{section3:init}), equivalence edges are generated solely
    from terms obtained by applying the same function identically to identical terms (congruence axiom~\ref{axiom4} and
    Theorem~\ref{theorem:Gen_Congruence}). The initialization rules (\emph{i.e.}, the condition $k_1 = k_2$) define the initial
    congruence closure for all $\Psi[\sigma]$, hence the invariant stated in Lemma~\ref{lemma:correctness} holds at initialization.

    Assume the invariant of Lemma~\ref{lemma:correctness} holds after the first $k$ equations. For the $(k{+}1)$-th equation,
    the algorithm adds equivalence edges via \texttt{EqProp}, which first adds an edge based on its parameters (Lines~\ref{algo:call_EqProp}
    and~\ref{algo:eq_end}), then adds edges for transitivity (Line~\ref{algo:tran_end}). Parameters come from two sources:
    the $(k{+}1)$-th equation itself, and edges from \texttt{CongruenceClosure} (Line~\ref{algo:eq_end}). Edges from the equation satisfy
    the invariant; those from \texttt{CongruenceClosure} use Rules~(\ref{formula1})--(\ref{formula4}), whose correctness is
    shown later. The proof for transitivity edges (added by \texttt{AddEq}) is given in Appendix~\ref{sec:appendix_correct}.
\end{proof}

The \texttt{CongruenceClosure} function derives new equivalence edges using Rules~(\ref{formula1})--(\ref{formula4}).
While Rules~(\ref{formula1})--(\ref{formula3}) trivially satisfy the invariant in Lemma~\ref{lemma:correctness} (\emph{e.g.}, Rule~(\ref{formula1})
directly encodes the congruence axiom in solving $\Psi[\sigma]$), we now prove that Rule~(\ref{formula4}) also preserves this invariant.

\begin{myLem}   \label{lemma1}
    Assume $n > m$, $ term = f^{m}(term) = f^{n}(term) \leftrightarrow term = f^{\mathsf{gcd}(m,n)}(term) $.
\end{myLem}

\begin{proof}
    The detailed proof can be found in Appendix~\ref{sec:appendix_th1}.
\end{proof}

% \begin{proof}
%     As the implication from right to left is obvious, we now prove the implication from left to right,
%     \emph{i.e.}, $ term = f^{m}(term) = f^{n}(term) \rightarrow term = f^{\mathsf{gcd}(m,n)}(term) $. Here we assume $m<n$.
%     By Theorem~\ref{theorem:Gen_Congruence}, the identity $term = f^m(term)$ implies $f^{n-m}(term) = f^n(term)$, and thus
%     by the transitivity of the equivalence relation:
%         $$ term = f^{n-m}(term) = f^{m}(term) = f^{n}(term) $$
%     If $ n-m \neq m $, we can can repeatedly apply Theorem~\ref{theorem:Gen_Congruence} until $ term = f^{a}(term) = f^{b}(term) \ s.t. \ a=b $ is derived.
%     This process utilizes a variant of the Euclidean algorithm~\cite{rosen2011elementary} to find the greatest common divisor of $m$ and
%     $n$. Thus, $ a = b = \mathsf{gcd}(m, n) $, and the formula on the right $ term = f^{\mathsf{gcd}(m,n)}(term) $ is derived.
% \end{proof}

\begin{myThe}   \label{theorem:th1}
Given two periodic equivalence relations $ term = f^{k}(term) $ and $ f^{m}(term) = f^{n}(term) $.
Then
$$ term = f^{k}(term) \wedge f^{m}(term) = f^{n}(term) \leftrightarrow term = f^{\mathsf{gcd}(k, (n-m))}(term). $$
\end{myThe}
\begin{proof}
       The detailed proof can be found in Appendix~\ref{sec:appendix_th1}.
\end{proof}

Assuming that all current \emph{CCG} edges satisfy the invariant in Lemma~\ref{lemma:correctness}, Theorem~\ref{theorem:th1} ensures that
Rule~(\ref{formula4}) precisely characterizes the equivalence class partitioning induced by multiple periodic equivalences.
Therefore, the addition of Rule~(\ref{formula4}) also preserves the invariant in Lemma~\ref{lemma:correctness}.

\begin{myThe}[Soundness and Completeness of Algorithm \ref{algo:solve}] \label{theorem:correctness}
    For any finite conjunctive formula $\Psi$ of EUF$^n$ equalities and disequalities, let $\Phi$ be the existential formula
    in Presburger arithmetic with divisibility generated by Algorithm~\ref{algo:solve}. Then
    \[
        \Psi \text{ is satisfiable } \iff \Phi \text{ is satisfiable}.
    \]
\end{myThe}

Lemma~\ref{lemma:correctness} establishes that the \emph{CCG} precisely captures all equivalence relations in $\Psi[\sigma]$. This
directly implies the correctness of Algorithm~\ref{algo:solve}, as stated in Theorem~\ref{theorem:correctness}.

We now analyze the complexity of Algorithm~\ref{algo:solve}.
\begin{myThe}[Complexity] \label{theorem:complexity}
    Algorithm~\ref{algo:solve} decides the satisfiability of a conjunction $\Psi$ of
    $\mathrm{EUF}^n$ equalities and disequalities of length $n$ in \textup{2NEXPTIME}.
\end{myThe}

\begin{proof}
    For an EUF$^n$ formula $\Psi$ and its instantiation yielding the disjunctive form of infinitely many EUF formulas $\Psi'$ (\emph{e.g.}, Equation~(\ref{equality:EUFn_EUF})), each EUF$^n$ term with a \emph{parametric composition depth} corresponds to infinitely many EUF terms. However, the number of possible superterm orderings is bounded by $n!$. For example, three terms $f^{m_1}(x)$, $f^{m_2}(x)$, and $f^{m_3}(x)$ yield $3!$ possible orderings of $m_1$, $m_2$, and $m_3$, each inducing a distinct superterm chain. For each ordering, a chain of length $n$ requires at most $(n-1)$ \textit{Union} operations, each corresponding to one \texttt{EqProp} execution. Consequently, Algorithm~\ref{algo:solve} executes $O(n! \cdot n)$ \texttt{EqProp} operations, where each \texttt{EqProp} invocation calls \texttt{AddEq} $O(n^2)$ times, and each \texttt{AddEq} call updates $O(n^2)$ superterm relations at Line~\ref{algo:inherit}.
    The cumulative cost of these internal operations is $O(n! \cdot n^5)$, which is asymptotically dominated by the subsequent solving of the generated formula and thus does not affect the overall complexity bound.

    % Each ordering is associated with corresponding constraints (\emph{e.g.}, $m_1 > m_2 > m_3$), which are incorporated into the equivalence conditions between terms. Upon termination of Algorithm~\ref{algo:solve}, these conditions may form an existential Presburger formula with divisibility of length $O(n!)$. Since solving such formulas is NEXPTIME~\cite{lechner2015complexity}, the complexity of solving the resultant formula is 2NEXPTIME (\emph{i.e.}, solvable in $O(2^{2^{p(n)}})$ time on a nondeterministic
    % Turing machine, for some polynomial $p$).

    % Each ordering is associated with a corresponding constraint (\emph{e.g.}, $m_1 > m_2 > m_3$), which is incorporated into the equivalence conditions between terms. Upon termination of Algorithm~\ref{algo:solve}, these conditions may give rise to an existential Presburger formula with divisibility whose length is $L = O(n!)$. The satisfiability problem for such formulas is in NEXPTIME~\cite{lechner2015complexity}; that is, for a formula of length $L$, it can be solved in time $2^{q(L)}$ for some polynomial $q$. Substituting $L = O(n!)$ yields a running time of \[2^{q(O(n!))}.\] Since \[n! = 2^{\Theta(n\log n)},\] and $q$ is a polynomial, we have \[q(O(n!)) = 2^{O(n\log n)}.\] Hence the overall running time is bounded by \[2^{2^{O(n\log n)}} \subseteq 2^{2^{p(n)}}\] for some polynomial $p$. Therefore, the resulting decision procedure runs in $2$NEXPTIME.

    Each ordering is associated with corresponding constraints (\emph{e.g.}, $m_1 > m_2 > m_3$), which are incorporated into the equivalence conditions between terms. Upon termination of Algorithm~\ref{algo:solve}, these conditions may form an existential Presburger formula with divisibility of length $O(n!)$. The satisfiability problem for such formulas is in NEXPTIME~\cite{lechner2015complexity}, meaning that, on a nondeterministic Turing machine, a formula of length $L$ can be decided in time $2^{p(L)}$ for some polynomial $p$. In our case, $L = O(n!)$, so the solving time becomes $2^{p(O(n!))} = 2^{O((n!)^c)}$ for some constant $c$. Since $n! = 2^{O(n \log n)}$, we have $(n!)^c = 2^{O(n \log n)}$. Thus, the total time is $2^{2^{O(n \log n)}}$, \emph{i.e.}, $2^{2^{p(n)}}$ for some polynomial $p$. Hence, the complexity of solving the resultant formula is 2NEXPTIME (\emph{i.e.}, solvable in $O(2^{2^{p(n)}})$ time on a nondeterministic Turing machine, for some polynomial $p$). Additionally, each call to \texttt{CongruenceClosure} (Line~\ref{algo:Congruence}) must verify the satisfiability of superterm conditions. Due to superterm inheritance, these conditions incorporate term equivalence relations, potentially involving divisibility. However, checking only the satisfiability of the ordering relation along the superterm chain (\emph{e.g.}, $m_1>m_2>\cdots$) suffices to ensure that no term becomes its own superterm, thereby guaranteeing termination. Since each such check requires solving an existential Presburger formula (NP-complete) and $O(n! \cdot n)$ checks are required, this step is in EXPTIME. Thus, the overall complexity of Algorithm~\ref{algo:solve} is 2NEXPTIME.
\end{proof}

%% file: figure/periodic_eq.tex
\tikzset{
    squarednode/.style = {rectangle, rounded corners, draw = black, fill=gray!5, very thick, text centered, align=center},
    originnode/.style = {circle, rounded corners, draw = black, fill=white!5,  text centered, align=center, minimum size=9mm, font=\bfseries\Large},
    arrow/.style       = {->, >=stealth, draw = black, line width=0.2mm},
    line/.style        = {-, draw = black, >=stealth,},
    dashed arrow/.style = {arrow, dashed},
    dashed line/.style  = {line, dashed},
    on background layer/.style={execute at begin scope={\scoped[on background layer]}},
    shadownode/.style={on background layer, draw=none, fill=lightgray!50!white, rounded corners=12pt, thick}
}

\begin{subfigure}{0.46\linewidth}
    \centering
    {\scriptsize
        \begin{tikzpicture}
            \begin{scope}[shift={(0,0)}]
                \node[shadownode, minimum width=45mm, minimum height=12mm] (shadowA) {};

                \node[originnode] (x) at ([xshift=8mm]shadowA.west)     {$x$};
                \node[originnode] (y) at ([xshift=-8mm]shadowA.east)    {$y$};
            \end{scope}

            \begin{scope}[shift={(0,1.5)}]
                \node[shadownode, minimum width=45mm, minimum height=12mm] (shadowB) {};

                \node[originnode, font=\footnotesize] (fx) at ([xshift=8mm]shadowB.west)     {$f(x)$};
                \node[originnode, font=\footnotesize] (fy) at ([xshift=-8mm]shadowB.east)    {$f(y)$};
            \end{scope}

            \begin{scope}[shift={(0,3)}]
                \node[shadownode, minimum width=45mm, minimum height=12mm] (shadowC) {};

                \node[originnode] (ffx) at ([xshift=8mm]shadowC.west)     {$\cdots$};
                \node[originnode] (ffy) at ([xshift=-8mm]shadowC.east)    {$\cdots$};
            \end{scope}

            \draw[arrow] (fx.south) -- node[left, align=center]{$f$} (x.north);
            \draw[arrow] (ffx.south) -- node[left, align=center]{$\cdots$} (fx.north);
            \draw[arrow] (fy.south) -- node[left, align=center]{$f$} (y.north);
            \draw[arrow] (ffy.south) -- node[left, align=center]{$\cdots$} (fy.north);

            % \draw[line] (x.east) --  (y.west);
            % \draw[line] (fx.east) --  (fy.west);
            % \draw[line] (ffx.east) --  (ffy.west);
        \end{tikzpicture}
    }
    \caption{Without superterm: Hierarchical structure}
    \label{fig:periodic_subfig1}
\end{subfigure}
\begin{subfigure} {0.46\linewidth}
    \centering
        {\scriptsize
        \begin{tikzpicture}
            \begin{scope}[shift={(0,0)}]
                \node[shadownode, minimum width=45mm, minimum height=12mm] (shadowA) {};

                \node[originnode] (x) at ([xshift=8mm]shadowA.west)     {$x$};
                \node[originnode, font=\tiny] (ffx) at (shadowA.center)    {$f^2(x)$};
                \node[originnode] (fnx) at ([xshift=-8mm]shadowA.east)    {$\cdots$};
            \end{scope}

            \begin{scope}[shift={(0,3)}]
                \node[shadownode, minimum width=45mm, minimum height=12mm] (shadowB) {};

                \node[originnode, font=\footnotesize] (fx) at ([xshift=8mm]shadowB.west)     {$f(x)$};
                \node[originnode, font=\tiny] (fffx) at (shadowB.center)    {$f^3(x)$};
                \node[originnode] (fnnx) at ([xshift=-8mm]shadowB.east)    {$\cdots$};
            \end{scope}

            \draw[arrow] (fx.south) -- node[left, align=center]{$f$} (x.north);
            \draw[arrow] (ffx.west) .. controls +(left:3mm) and +(right:3mm) ..  node[left, align=center]{$f$} (fx.east);
            \draw[arrow] (fffx.south) -- node[left, align=center]{$f$} (ffx.north);
            \draw[arrow] (fnx.west) .. controls +(left:3mm) and +(right:3mm) ..  node[left, align=center]{$f$} (fffx.east);
            \draw[arrow] (fnnx.south) -- node[left, align=center]{$f$} (fnx.north);
        \end{tikzpicture}
    }
    \caption{With superterm: Partitioned structure}
    \label{fig:periodic_subfig2}
\end{subfigure}

%% file: section/sec4.tex
% \section{Analysis of EUF$^n$'s Practical Encoding Capacity}  \label{sec:section4}
\section{Applications of EUF$^n$} \label{sec:section4}
This section demonstrates the applications of EUF$^n$ through two case studies: the single-source single-target
interleaved Dyck reachability problem and the verification of uninterpreted programs.

\input{section/sec4_Dyck.tex}

\subsection{Verification of Uninterpreted Programs}
In this subsection, we characterize the class of uninterpreted programs that are encodable in $\mathrm{EUF}^n$ and, on this basis,
identify a decidable subclass from the perspective of solving $\mathrm{EUF}^n$ formulas.

\subsubsection{Background of Uninterpreted Program Verification}
Uninterpreted programs~\cite{mathur2019decidable} are a class of programs that work over infinite domains and relate closely to EUF.
Functions in these programs are represented solely by function symbols, serving as abstractions for actual programs
~\cite{1997Automatic, CLU, lee2014unbounded, fan2024leveraging}, enabling quick verification of program properties.
Despite the abstraction, the verification of uninterpreted programs remains undecidable~\cite{mathur2019decidable}.
Recently, a decidable subclass known as \emph{coherent} programs~\cite{mathur2019decidable} has been identified, which is applicable to
verifying heap manipulation programs~\cite{mathur2019deciding}. However, very few real-world programs meet the
stringent \emph{coherence} requirements.

\noindent\textbf{Syntax of Uninterpreted Program.}
Let $ \Sigma_{UP} $ be a finite symbolic set,  $\Sigma_{f} \subseteq \Sigma_{UP} $ represents function set, $ \Sigma_{V} \subseteq \Sigma_{UP}$ represents variable set. The syntax definition of the uninterpreted program is shown in Figure \ref{UP syntax}, where $\mathtt{f} \in \Sigma_{f}$ denotes a function, $ \mathtt{x},\mathtt{y} \in \Sigma_{V} $ denotes a variable, and $\vec{\mathtt{z}}$ is a tuple of variables.
\begin{wrapfigure}{r}{0.55\textwidth}
    \begin{center}
    \vspace{6pt}
    {
    \begin{tabular}{rcl}
    $\mathtt{stmt} $ & $::=$ & $\mathtt{skip} \mid \mathtt{x:=y} \mid \mathtt{x:=f}(\vec{\mathtt{z}})$ \\
    &&$\mid \mathtt{assume(cond)} \mid \mathtt{assert(cond)}$ \\
    &&$\mid \mathtt{stmt} ; \mathtt{stmt}$\\
    &&$\mid {\mathtt{if}}\ \mathtt{(cond)}\ {\mathtt{then}}\ \mathtt{\{stmt;\}}\ {\mathtt{else}}\ \mathtt{\{stmt;\}}$ \\
    &&$\mid {\mathtt{while}}\ \mathtt{(cond)}\ \mathtt{\{stmt;\}}$ \\
    $\mathtt{cond}$ & $::=$ & $\mathtt{x=y} \mid \mathtt{cond \land cond} \mid \mathtt{\neg cond  }
    $ \\
    \end{tabular}
    }
    \end{center}
    % \vspace{-10pt}
    \caption{The syntax of uninterpreted programs.}
    \label{UP syntax}
    % \vspace{-12pt}
\end{wrapfigure}

% \begin{figure}[h]
%     \begin{center}
%     {
%     \begin{tabular}{rcl}
%     $\stdef{\mathtt{stmt}} $ & $::=$ & $\mathtt{skip} \mid \mathtt{x:=y} \mid \mathtt{x:=f}(\bar{\mathtt{z}})$ \\
%     &&$\mid {\mathtt{assume}}(\stdef{cond}) \mid {\mathbf{assert}}(\stdef{cond})$ \\
%     &&$\mid \stdef{\mathtt{stmt}} ; \stdef{\mathtt{stmt}}$\\
%     &&$\mid {\bf{if}}\ (\stdef{cond})\ {\bf{then}}\ \stdef{\mathtt{stmt}}\ {\bf{else}}\ \stdef{\mathtt{stmt}}$ \\
%     &&$\mid {\bf{while}}\ (\stdef{cond})\ \stdef{\mathtt{stmt}}$ \\
%     &&\\
%     $\stdef{cond}$ & $::=$ & $x=y %\mid {\bf{\neg}} \stdef{cond}
%      \mid \stdef{cond} {\bf{\wedge}} \stdef{cond} \mid {\bf{\neg}} \stdef{cond}
%     $ \\
%     \end{tabular}
%     }
%     \end{center}

%     \caption{The syntax of uninterpreted programs.}
%     \label{UP syntax}
% \end{figure}

In the above definition, a $\mathtt{skip}$ statement performs no operation and proceeds directly to the next statement.
The program includes two types of assignment statements: variable assignment and function assignment. As well as assume statements,
assert statements, and three program structures: sequential structure, $\mathtt{if}$-branch structure, and $\mathtt{While}$-loop
structure. Conditional statements include equivalence relations, as well as corresponding conjunctions and negations. Definitions
of the semantics and verification problem for uninterpreted programs can be found in Appendix~\ref{sec:appendix_background}.

% \begin{figure}[!ht]
%       \scalebox{0.8}{\input{figure/encoding_UP.tex}}
%     \caption{An example program and the formula encoding by $\text{EUF}^{n}$.}
%     \label{fig:exam_program}
%     \vspace{-10pt}
% \end{figure}

%Original Example
%%=========================================================================================================================== %%

\subsubsection{An Example of Verifying Uninterpreted Programs Using EUF$^n$}
EUF$^n$ can be used to encode verification problems for partially uninterpreted programs, including some that do not satisfy the
\emph{coherence} requirement. We now present an example of verifying a \emph{non-coherent} program using EUF$^n$. As shown in
Figure~\ref{fig:exam_program}, the goal is to verify whether the assertion at the end of the program holds.
The program contains two loops with identical behavior, they can be seen as two versions of the same code, differing only by
variable renaming—specifically, $\mathtt{x}$ to $\mathtt{y}$ and $\mathtt{z}$ to $\mathtt{k}$. The assertion at the end of the
program expresses the functional equivalence of these two code segments. However, since terms computed in the first loop are
recomputed in the second without being stored in variables, the program is \emph{non-coherent}~\cite{mathur2019decidable} and
cannot be verified using the methods in~\cite{mathur2019decidable}.

% The program can also be viewed as an abstraction of
% traversing a linked list starting from different pointers, where $\mathtt{x}$ and $\mathtt{y}$ initially point to the same node.
% The assertion states that after the same number of $\mathtt{next}$ traversals, both pointers still refer to the same node.

We now encode the program using $\mathrm{EUF}^n$. First, for each program variable, we introduce a formula variable to represent its initial value, such as $x_0$, $y_0$, $z_0$, and $k_0$.
Next, for assignment statements, the left-hand variable is rewritten, so we introduce new formula variables to represent the modified
values. For example, for the statement $\mathtt{x := y; z := k}$, we encode it as $x_1 = y_0 \land z_1 = k_0$.
We encode the loops in the example program by introducing a natural number variable $n$ for the number of iterations. The encoding
consists of three parts: (1) \textbf{Entry conditions}: Introduce a universal quantifier $i$ for the $i$-th entry
condition. The condition $\mathtt{x!=z}$ can be encoded as $i < n \rightarrow next^i(x_1) \neq add^i(z_1)$. (2) \textbf{Loop summary}:
In the first loop, the statements $\mathtt{x := next(x); z = add(z)}$ are executed $n$ times, yielding $x_2 = next^n(x_1)$ and
$z_2 = add^n(z_1)$. (3) \textbf{Exit condition}: After $n$ iterations, when the loop exits and the loop condition is not satisfied,
we encode it as $x_2 = z_2$. The encoding for the second loop is similar.

\begin{wrapfigure}{l}{0.63\textwidth}
%   \vspace{-5mm}
  \centering
  \adjustbox{scale=0.8}{
    \input{figure/encoding_UP.tex}
  }
  \vspace{-8pt}
  \caption{An example program and the formula encoding by $\text{EUF}^{n}$.}
  \label{fig:exam_program}
  \vspace{-10pt}
\end{wrapfigure}

This example program lacks branching structures. For programs with branches, we can encode each branch according to the rules
for sequential statements and take the disjunction of these encodings. The assertion statement $\mathtt{assert(x = y)}$ encodes
the variables $x$ and $y$ as $x_2$ and $y_1$, resulting in the formula $x_2 = y_1$. Thus, this EUF$^n$ formula encodes the
existence of a trace that violates the assertion. If the formula is satisfiable, the values of $n$ and $m$ correspond to such a violating
trace, suggesting the program is incorrect. Conversely, if unsatisfiable, it indicates no violating traces, confirming
the program's correctness.

% As previously mentioned, encoding loop conditions will introduces universal quantifiers. As shown in the right-hand formula of
% Figure \ref{fig:exam_program}, although this formula contains universal quantifiers, the \emph{CCG} algorithm for solving
% EUF$^n$ formulas (Algorithm \ref{algo:solve}) is fundamentally a reduction algorithm. Therefore, the EUF$^n$ formula with
% universal quantifiers in Figure \ref{fig:exam_program} can be transformed into a quantified Presburger arithmetic formula:

As previously mentioned, encoding loop conditions introduces formulas of the form $\forall i. i < n \rightarrow next^i(x_1) \neq add^i(z_1)$, where the bounded universal quantifier $\forall i. i<n$ is not supported in the syntax of EUF$^n$ defined in Figure~\ref{fig:syntax}. However, a simple modification to the \emph{CCG} algorithm suffices to solve such formulas. For instance, the formula $\forall i. i < n \rightarrow next^i(x_1) \neq add^i(z_1)$ can be expressed in \emph{CCG} as the equivalence of nodes $next^i(x_1)$ and $add^i(z_1)$ under the condition $\forall i. i<n$. Therefore, the EUF$^n$ formula with universal quantifiers in Figure \ref{fig:exam_program} can be transformed into a quantified Presburger arithmetic formula:
\begin{align*}
    \forall i,j. & \neg[(i<n \wedge i=m) \vee (i<n \wedge i=n)] \\
                 &\wedge \neg[(j<m \wedge j=m) \vee (j<m \wedge j=n)] \wedge \neg(m=n).
\end{align*}
The unsatisfiability of this formula formally establishes the correctness of the original program. Notably, this program
fails to preserve intermediate computation results (specifically, terms computed in the first loop are redundantly
recalculated in the second loop without variable retention), thereby violating the coherence criterion. Consequently,
such programs cannot be verified as correct under the methodology proposed by Mathur \emph{et al.}~\cite{mathur2019decidable}.

% %% ========================================================================================================================== %%

\subsubsection{Decidability of Uninterpreted Programs}
We now present a class of uninterpreted programs whose verification problems can be encoded as EUF$^n$ formulas (Theorem~\ref{theorem:UP_encode}),
and then, based on this class, we identify a decidable subclass of uninterpreted programs (Theorem~\ref{theorem:UP_decidable}).

\begin{myThe}[EUF$^n$ Encodability of Uninterpreted Programs] \label{theorem:UP_encode}
    The verification problem for uninterpreted programs without nested loops, conditional branches within loops, or multi-ary
    functions in loop bodies can be encoded as the satisfiability problem of EUF$^n$ formulas.
\end{myThe}
\begin{proof}
     The detailed proof is provided in Appendix~\ref{sec:appendix_encoding_UP}.
\end{proof}

\begin{myThe}[Decidability of the Uninterpreted Programs Verification Problem]\label{theorem:UP_decidable}
    There exists a decidable class $\mathcal{P}$ of uninterpreted programs where: (1) Every program $P \in \mathcal{P}$ can be
    encoded into an EUF$^n$ formula $\phi_P$; (2) For any assume statement of the form $\mathtt{assume(x = y)}$
    in $P$, the superterm edges set $E_{\succ}$ contains no element $e$ satisfying $e = (x, \phi, y)$ with $\phi$ is satisfiable.
\end{myThe}

\begin{proof}
    Encoding loops in uninterpreted programs introduces universal quantifiers. According to Algorithm~\ref{algo:solve} and the
    discussion in Section~\ref{section3:congruence}, if \textbf{periodic equivalence} relations do not appear in the formula, then no
    divisibility predicates are introduced in the conditions of equivalence edges. Theorem~\ref{theorem:UP_decidable}, specifically
    Condition~(2), ensures the absence of such \textbf{periodic equivalence} relations. As a result, the EUF$^n$ formula obtained
    from encoding this class of uninterpreted programs can be transformed into a quantified Presburger arithmetic formula.
    The satisfiability of the original formula is preserved under this transformation by Algorithm~\ref{algo:solve}, and the resulting
    formula is decidable. Therefore, the verification problem for this class of uninterpreted programs is decidable.
\end{proof}

For any given uninterpreted program, determining whether it satisfies Condition~(1) is straightforward. After encoding, Condition~(2)
can be verified during the \emph{CCG}-based solving process. The decidable subclass of uninterpreted programs established in
Theorem~\ref{theorem:UP_decidable} has a non-empty intersection with the existing class of coherent decidable programs
\cite{mathur2019decidable, mathur2019deciding}, but neither class subsumes the other. For example, Figure~\ref{fig:exam_program}
illustrates a non-coherent program; on the other hand, uninterpreted programs with nested loop structures, while potentially
coherent, cannot be encoded in EUF$^n$.
Furthermore, \cite{gulwani2007assertion} abstract programs using uninterpreted functions, where conditional branches are replaced by nondeterministic choice, equality guards are omitted, and the verification mainly targets equality assertions. \cite{godoy2009invariant} proposed a decidable class of uninterpreted programs called Sloopy programs, which similarly abstracts control flow with nondeterministic choice and reduces assertion conditions to Presburger formulas, an idea closely related to the verification approach in this paper. The decidable classes in these two works both have a non-empty intersection with the class defined in this paper, but neither is subsumed by the other. Their classes support multi-argument uninterpreted functions inside loops, whereas our approach can directly model branch conditions in programs.

The limitations of EUF$^n$ in encoding nested loops and branching within loops stem from the absence of variables representing
inner iteration counts and the inability to handle compositions of multi-arity functions at \emph{parametric depth}, thereby restricting
loops to unary functions. Despite this, this class of uninterpreted programs remains practically useful: it can model heap operations
(e.g., reads/writes~\cite{mathur2019deciding}), and multi-arity functions can be abstracted into unary functions with vector arguments,
thereby falling within this class. Moreover, in control property verification, function semantics are often irrelevant in that the
order and frequency of specific operations are what truly matter~\cite{bueno2021software}. For instance, in the OpenSSL vulnerability
CVE-2015-0291, a deallocated state variable \texttt{shared\_sigalgs} was not synchronized with its length field \texttt{shared\_sigalgslen},
leading to a null dereference. Such bugs can be abstracted as constraints on update counts of related variables, \emph{e.g.}, using
counter variables with updates like \texttt{count\textsubscript{i} = add(count\textsubscript{i})}, and enforcing relations such as
$count_i = count_j$ or $count_i = add^n(count_j) \land n > 0$.

%% file: section/sec4_Dyck.tex
\subsection{Single-Source Single-Target Interleaved Dyck Reachability}
In this subsection, we formalize the encoding of a subclass of the single-source single-target interleaved Dyck reachability
problem (\textsc{IDRP}) in \textsc{EUF}$^n$ and prove that this subclass is decidable.

\subsubsection{Background} We first introduce the concepts of IDRP and motivate our EUF$^n$ encoding.

\noindent\textbf{Regular Expression.}
  A regular expression over an alphabet $\Sigma_{\mathcal{R}}$ is defined recursively as:
  \begin{align*}
    \mathcal{R} ::= \epsilon  \mid a \mid (\mathcal{R}) \mid \mathcal{R}_1 \cdot \mathcal{R}_2 \mid \mathcal{R}_1 + \mathcal{R}_2  \mid  \mathcal{R} \text{\Kleene}
  \end{align*}
  where $\epsilon$ denotes the empty string and $a \in \Sigma_{\mathcal{R}}$ represents a terminal symbol, $(\mathcal{R})$ denotes
  expression grouping to explicitly specify precedence or grouping, $\mathcal{R}_1 \cdot \mathcal{R}_2$ represents concatenation,
  $\mathcal{R}_1 + \mathcal{R}_2$ signifies alternation, $\mathcal{R} \text{\Kleene}$ indicates zero or more repetitions.
  Operator precedence (descending): Kleene star ($^*$), concatenation ($\cdot$), alternation ($+$).

\noindent\textbf{Dyck Language.}
The Dyck language, a context-free language (CFL) class for bracket-matching, is formally defined by the production rule
$ S \rightarrow SS \mid \llparenthesis_{i} S \rrparenthesis_{i} \mid \epsilon $ over alphabet
$\Sigma = \{ \llparenthesis_{i}, \rrparenthesis_{i} \mid i = 1, \ldots, k \}$. Here $\llparenthesis_{i}$ are termed \emph{opening labels}
and $\rrparenthesis_{i}$ \emph{closing labels}. The Dyck reachability problem determines existence of a path between graph nodes where the
concatenated edge labels form a valid Dyck string. This formalism models critical program behaviors including function call/return
pairs~\cite{sridharan2005demand,spath2019context} and field access sequences~\cite{spath2019context}, establishing its broad utility in program analysis.

\noindent\textbf{Single-Source Single-Target IDRP.}
Given Dyck languages $D_{\alpha}$ and $D_{\beta}$, defined over alphabets $\Sigma_{\alpha}$ and $\Sigma_{\beta}$, respectively,
a string $s$ is belongs to an interleaved Dyck language, denoted as $ s \in D_{\alpha} \odot D_{\beta}$,
if its projections on $\Sigma_{\alpha}$ and $\Sigma_{\beta}$ belong to $ D_{\alpha} $ and $ D_{\beta} $ respectively
(\emph{i.e.}, $  s \downharpoonright \Sigma_{\alpha} \in D_{\alpha} \wedge s \downharpoonright \Sigma_{\beta} \in D_{\beta} $).
Given a graph $G = (V, E)$ where each edge $e \in E$ carries a label $\lambda(e) \in \Sigma_\alpha \cup \Sigma_\beta$, two
nodes $s, t \in V$ are \textit{$D_{\alpha} \odot D_{\beta}$-reachable} iff there exists a path from $s$ to $t$ whose edge label
concatenation forms a string $p \in D_{\alpha} \odot D_{\beta}$. For instance, in Figure~\ref{fig:interleaved_examp}, nodes $a$
and $d$ are connected by a path labeled $\llparenthesis_{1}\llparenthesis_{2}\llbracket_{1}\rrparenthesis_{2}\rrparenthesis_{1}\rrbracket_{1}$,
demonstrating $D_{\alpha} \odot D_{\beta}$-reachability, where
$\Sigma_{\alpha}=\{\llparenthesis_{1}, \llparenthesis_{2}, \rrparenthesis_{1}, \rrparenthesis_{2}\}$ and
$\Sigma_{\beta}=\{\llbracket_{1}, \rrbracket_{1}\}$.

% % \vspace{-10pt}
% \begin{figure}[ht]
%     \centering
%     \input{figure/interleaved_ex.tex}
%     \caption{An example of interleaved Dyck reachability.}
%     \label{fig:interleaved_examp}
% \end{figure}
% % \vspace{-10pt}

Furthermore, reachability evaluated solely between two specified points is termed single-source single-target interleaved Dyck
reachability. While the interleaved Dyck reachability problem is undecidable~\cite{reps2000undecidability},
its importance in static analysis applications such as context-sensitive data-dependence analysis~\cite{reps2000undecidability}
and context- and field-sensitive data-flow analysis~\cite{spath2019context} has driven substantial research into
approximation techniques~\cite{spath2019context,li2023single,zhang2017context,conrado2025program}.
However, existing methods universally neglect the structural properties of connecting paths between reachable nodes.

\begin{wrapfigure}{r}{0.5\textwidth}
  \vspace{-2pt}
  \centering
  \input{figure/interleaved_ex.tex}
  \caption{An example of interleaved Dyck reachability.}
  \label{fig:interleaved_examp}
  \vspace{-10pt}
\end{wrapfigure}
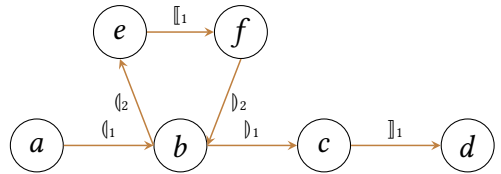

\noindent\textbf{Path Expressions.} Given a graph $G$, all possible paths between any two nodes can be represented using a
regular expression~\cite{tarjan1981fast}. Such expressions can be computed in $O(|E| \cdot \alpha(|E|))$ time, where $|E|$ is
the number of edges and $\alpha$ is the inverse Ackermann function. For example, in Figure~\ref{fig:interleaved_examp}, the path
expression between nodes $a$ and $d$ is represented by the regular expression
$ \llparenthesis_{1} (\llparenthesis_{2}\llbracket_{1}\rrparenthesis_{2})^{*} \rrparenthesis_{1} \rrbracket_{1} $.

\noindent\textbf{Observation.}
We identify two structural patterns under which reachability is decidable: (1) path expressions contain no Kleene star operators; and
(2) Kleene star operators are present, but they are neither:
  \begin{itemize}
    \item Nested within another Kleene star operator (\emph{e.g.}, $((a)^*b)^*$), nor
    \item Combined with the alternation operator (\emph{e.g.}, $(a + b)^*$)
  \end{itemize}

The decidability of the single-source single-target IDRP is straightforward in case (1). For case (2), this paper establishes its
decidability by encoding the problem using $\mathrm{EUF}^n$. We further reveal that state-of-the-art approximation
techniques—including SPDS~\cite{spath2019context}, LCL~\cite{zhang2017context}, ILP-based methods~\cite{li2023single}, and
$d$-MCFL($r$)~\cite{conrado2025program}—fail to yield accurate results even under these constraints.

\subsubsection{Decidability of the IDRP}
We first outline the main idea of the encoding scheme. Nodes in the graph are encoded as variables, and labeled edges are encoded using uninterpreted functions.
For example, a sequence such as $a \xrightarrow{ \llparenthesis_{1} } b \xrightarrow{ \rrparenthesis_{1} } c$
is encoded as $a = f(b) \land c = f(b)$, where matching labels (\emph{e.g.} $\llparenthesis_{1}$/$\rrparenthesis_{1}$) are mapped to the
same uninterpreted function. In this way, parenthesis matching is modeled as congruence reasoning, while the relative order of parentheses
is captured by the composition order of uninterpreted functions. This idea closely resembles the encoding of bidirected Dyck
reachability using EUF by Du \emph{et al.}~\cite{du2026euf}.

We restrict our scope to single-source single-target scenarios for two reasons.
First, Interleaved Dyck Reachability is not an equivalence relation.
In the example above, node $a$ can reach node $c$, but not vice versa; yet the encoded formula makes variables $a$ and $c$ equivalent.
Without this restriction, we cannot distinguish whether $a$ reaches $c$ or $c$ reaches $a$.
Second, branching in paths introduces another problem: edges such as $a \xrightarrow{\llparenthesis_{1}} b$ and $c \xrightarrow{\rrparenthesis_{1}} b$ are
directly encoded as $a = f(b) \wedge c = f(b)$, which incorrectly implies $a = c$ even though there is no reachability between them.

% Although Interleaved Dyck Reachability is not an equivalence relation,
% in the single-source single-target setting, the direction of the path is well-defined. For example, given the edge
% $a \xrightarrow{\llparenthesis_{1}} b$, which is encoded as $a = f(b)$, this encoding is fixed and will not be confused with the reverse
% edge $b \xrightarrow{\rrparenthesis_{1}} a$, which is also encoded as $a = f(b)$. Therefore, path regular expressions without Kleene
% stars can be directly encoded via the method described above.

% For each Kleene star structure in the expression, we introduce natural number variables (\emph{e.g.}, $n$) to represent repetition counts, with different values corresponding to paths that traverse the loop different numbers of times, thereby capturing distinct inter-node paths; and we assign two variables (\emph{e.g.}, $v_{in}$ and $v_{out}$) representing the loop's entry and exit points. The paths between $v_{in}$ and $v_{out}$ can be encoded using EUF$^n$ formulas, where the extended expressive power of EUF$^n$ enables the encoding of loops, as established in Theorem~\ref{theorem:dyck-decidable}.

From the above encoding, a branch-free (``linear'') unidirectional path can be encoded as an EUF formula.
When the path contains loops (like Figure~\ref{fig:interleaved_examp}), we can introduce natural number variables (\emph{e.g.}, $n$) to represent repetition counts, with different values corresponding to linear paths that traverse the loop different numbers of times.
By assigning two variables (\emph{e.g.}, $v_{in}$ and $v_{out}$) representing the loop's entry and exit points, the path between $v_{in}$ and $v_{out}$ can be encoded using EUF$^n$ formulas, where the extended expressive power of EUF$^n$ enables the encoding of loops, as established in Theorem~\ref{theorem:dyck-decidable}.

Note that the encoding is performed separately for $D_{\alpha}$- and $D_{\beta}$-reachability between nodes. If the sets of reachable
assignments for the loop count variables corresponding to $D_{\alpha}$- and $D_{\beta}$-reachability intersect, then the combined
assignment of all loop count variables corresponds to a path witnessing the $D_{\alpha}$- and $D_{\beta}$ reachability. We now illustrate
the encoding with the $D_{\alpha} \odot D_{\beta}$-reachability problem between nodes $a$ and $d$ in Figure~\ref{fig:interleaved_examp}.
The path expression between $a$ and $d$ is given by
$\mathcal{R} = \llparenthesis_{1} (\llparenthesis_{2}\llbracket_{1}\rrparenthesis_{2})^{*} \rrparenthesis_{1} \rrbracket_{1}$,
which contains a Kleene star structure. We introduce a variable $n$ to represent the number of iterations of the loop. The final encoding
using EUF$^n$ is shown below, with $D_{\alpha}$ and $D_{\beta}$ components separated for clarity:
\begin{align*}
    \text{Encoding} \  D_{\alpha}: & a = f_1(b_{\text{in}}) \wedge b_{\text{in}} = b_{\text{out}} \wedge c = f_1(b_{\text{out}}) \wedge c = d  \\
    \text{Encoding} \  D_{\beta}:  & a^{\prime} = b^{\prime}_{\text{in}} \wedge b^{\prime}_{\text{in}} = g_1^n(b^{\prime}_{\text{out}}) \wedge b^{\prime}_{\text{out}} = c^{\prime} \wedge g_1(c^{\prime}) = d^{\prime}
\end{align*}
The complete encoding procedure is detailed in Appendix~\ref{sec:appendix_encoding_Dyck}. When solving the formula obtained from encoding $D_{\alpha}$ using the \emph{CCG}, the condition for $a$ to be equivalent to $d$ is set to $\mathbf{true}$. When solving the formula from encoding $D_{\beta}$, the condition for $a^{\prime}$ to be equivalent to $d^{\prime}$ is specified as $n = 1$. Since the conjunction $\mathbf{true} \land n = 1$ is satisfiable, there exists a path such that nodes $a$ and $d$ are interleaved Dyck reachable, \emph{i.e.}, a path that enters the loop exactly once.

% Figure~\ref{fig:WellFormedR} defines the syntax of regular expressions under these structural constraints, termed \textit{WellFormed}
% regular expressions, where $a \in \Sigma_{\mathcal{R}}$. In this language, expressions prohibit nested Kleene stars and avoid alternation
% inside Kleene star. We first present an intuitive encoding scheme for \textit{WellFormed} regular expressions using EUF$^n$, and then
% prove that the corresponding IDRP is decidable when all path expressions between the target node pair are \textit{WellFormed}.

\begin{wrapfigure}{r}{0.5\textwidth}
  \centering
  $\begin{array}{rcl}
    WellFormed &::=& Term \mid Term \cdot Term \\
    Term   &::=& Factor \mid (Factor)^* \\
    Factor &::=& Atomic \mid Atomic \cdot Atomic \\
    Atomic &::=& a \mid \epsilon
  \end{array}$
  \vspace{-8pt}
  \caption{The Syntax of \textit{WellFormed} regular expressions.}
  \label{fig:WellFormedR}
  \vspace{-12pt}
\end{wrapfigure}
Formally, Figure~\ref{fig:WellFormedR} defines the syntax of regular expressions under the above structural constraints, termed \textit{WellFormed} regular expressions, where $a \in \Sigma_{\mathcal{R}}$.
In this language, expressions prohibit nested Kleene stars and avoid alternation inside a Kleene star.
If the path expression between nodes can be represented by a finite number of \textit{WellFormed} expressions, \emph{e.g.}, $\mathcal{R}_1 + \mathcal{R}_2$, we can encode $\mathcal{R}_1$ and $\mathcal{R}_2$ separately and take the disjunction of the resulting formulas.
Having presented an intuitive encoding scheme for \textit{WellFormed} regular expressions using EUF$^n$, we now prove that the corresponding IDRP is decidable when all path expressions between the target node pair are \textit{WellFormed}.

\begin{myThe}[Decidability of the IDRP] \label{theorem:dyck-decidable}
  Let $\mathcal{P}$ be the set of regular expressions encoding paths between two nodes $u$ and $v$ in a single-source single-target
  Interleaved Dyck reachability problem. If each expression $\mathcal{R} \in \mathcal{P}$ is \textit{WellFormed}, then $\mathcal{P}$ can be
  encoded as a set of EUF$^n$ formulas, thus establishing the decidability of the reachability problem.
\end{myThe}

% \begin{proof}
%   The detailed proof and the specific encoding algorithm are provided in Appendices~\ref{sec:appendix_encoding_Dyck}.
% \end{proof}

\begin{proof}
  We first demonstrate that any \textit{WellFormed} regular expression $\mathcal{R}$ can be encoded as an EUF$^n$ formula.
  As shown in the example in Figure~\ref{fig:interleaved_examp}, the primary challenge is automatically encoding Kleene star
  structures. We subsequently show how EUF$^n$ formulas can express such structures, with the complete encoding algorithm for $\mathcal{R}$ provided in Appendix~\ref{sec:appendix_encoding_Dyck}.

  Given a Kleene star expression $(s)^*$ and a natural number variable $n$, we remove the Dyck-matching pairs from $s$ to obtain
  a simplified string $s'$. It can be shown that the encodings for $(s)^*$ and $(s')^*$ are equivalent. It is straightforward to detect
  whether $s'$ contains any invalid pattern. The possible valid forms of $s'$ fall into three categories: (1) a sequence of only
  \emph{opening labels}, (2) a sequence of only \emph{closing labels}, or (3) a sequence of \emph{closing labels} followed by
  \emph{opening labels}. Cases (1) and (2) can be encoded easily, and the final formula takes the form:
  \[
  v_{in}/v_{out} = (f_1 \circ \cdots \circ f_k)^n (v_{out}/v_{in}),
  \]
  where $k$ is the length of $s'$.

  Suppose in case (3), the string contains $k_1$ \emph{closing labels} followed by $k_2$ \emph{opening labels}. When $n = 1$, the functional relationship between
  $v_{in}$ and $v_{out}$ is easy to compute. When $n > 1$, assume without loss of generality that $k_1 > k_2$ (the case $k_1 \leq k_2$ follows analogously).
  If the $k_2$ \emph{opening labels} do not match the $k_2$-length prefix of the $k_1$ \emph{closing labels}, then only one valid path exists for $n = 1$.
  Otherwise, repeating $s'$ $n$ times yields a path that begins with $k_1$ \emph{closing labels}, followed by $(k_1 - k_2)$ repeated \emph{closing labels} for
  each additional iteration, and ends with $k_2$ \emph{opening labels}. After introducing an intermediate variable $v_{\mathit{temp}}$, the corresponding
  functional relationship among $v_{\mathit{temp}}$, $v_{in}$, and $v_{out}$ can be derived accordingly.

  After encoding, the resulting EUF$^n$ formula contains no periodic equivalences and consists solely of a conjunction of equalities. Solving this formula
  using the \emph{CCG}-based method yields the equivalence conditions between the variables corresponding to the source and target nodes. These conditions
  form a quantifier-free Presburger formula, and its satisfiability problem is NP-complete. If this formula is satisfiable, then $D_{\alpha}$- or
  $D_{\beta}$-reachability holds between the nodes. If the intersection of their solution sets is nonempty, then $D_{\alpha} \odot D_{\beta}$-reachability is established.
\end{proof}

%%===========================================%%

However, existing methods—such as $d$-MCFL($r$)-based reachability \cite{conrado2025program}, SPDS \cite{spath2019context},
LCL \cite{zhang2017context}, and ILP-based approaches \cite{li2023single}—may still produce false positives even under
the constraints of \textit{WellFormed} regular expressions. First, consider the under-approximation technique based on $d$-MCFL($r$)
reachability~\cite{conrado2025program}, where parameters
$d$ and $r$ control approximation precision (with larger values yielding tighter approximations for IDRP).
% This method only
% considers valid paths comprising at most $2d$ concatenated segments.
While this strategy eliminates false positives, it
inevitably introduces false negatives: for any fixed $d$ and $r$, there exist Interleaved Dyck strings unrecognizable
by $d$-MCFL($r$), rendering it incapable of resolving certain IDRP instances even in loop-free cases. We now present
counterexamples for the over-approximation algorithm under the same constraints.

\vspace{8pt}
\noindent{\bf{SPDS~\cite{spath2019context}}} The synchronized pushdown system (SPDS) approximates IDRP by
independently computing the reachability of the two Dyck languages involved in the interleaving. For example, given $L=D_\alpha \odot D_\beta$,
SPDS considers $s$ and $t$ to be $L$-reachable if there are paths from
$s$ to $t$ that are both $D_\alpha$-reachable and $D_\beta$-reachable. However, this approximation may lead to false positives if the
$D_\alpha$-reachable and $D_\beta$-reachable paths are not the same, and there is no path that is both $D_\alpha$-reachable and
$D_\beta$-reachable, \emph{i.e.}, no path that is interleaved Dyck reachable. For instance, suppose the path labels between nodes
$s$ and $t$ form the following regular expression with only one layer of looping is
$$ \llparenthesis_{1} ( \rrparenthesis_{1} \llbracket_{1} )^{*} \llbracket_{1} \rrbracket_{1}, $$
where $\Sigma_\alpha = \{ \llparenthesis_{1}, \rrparenthesis_{1} \}$ and $ \Sigma_\beta = \{ \llbracket_{1}, \rrbracket_{1} \} $.
Since there exists a path
$ \llparenthesis_{1} \rrparenthesis_{1} \llbracket_{1} \llbracket_{1} \rrbracket_{1} \downharpoonright \Sigma_{\alpha} \in D_{\alpha} $
and a path $ \llparenthesis_{1} \llbracket_{1} \rrbracket_{1} \downharpoonright \Sigma_{\beta} \in D_{\beta} $, SPDS incorrectly
concludes that $s$ and $t$ are interleaved Dyck reachable. However, upon further analysis, it becomes clear that only the path entering
the loop once is $D_{\alpha}$-reachable, and this path is not $D_{\beta}$-reachable. Thus, no interleaved Dyck reachable path exists between
$s$ and $t$.

The following explains how to solve this example using EUF$^n$ encoding. First, we define a variable $n$ to represent the number of loop
executions in the path, and then proceed with the encoding for $D_{\alpha}$, resulting in the following:
$$ s=f(v_{1\_\text{in}}) \wedge v_{1\_\text{out}} = f^n(v_{1\_\text{in}}) \wedge v_2 = v_{1\_\text{out}}  \wedge t = v_2 $$
The condition for $s$ and $t$ to be equivalent when processing this formula is $n=1$, meaning the path that executes the loop once
is reachable via $D_{\alpha}$. Then proceed with the encoding for $D_{\beta}$, resulting in the following:
$$ s^{\prime}=v^{\prime}_{1\_\text{in}} \wedge v^{\prime}_{1\_\text{in}} = g^n(v^{\prime}_{1\_\text{out}}) \wedge v^{\prime}_{1\_\text{out}} = g(v^{\prime}_2) \wedge t^{\prime} = g(v^{\prime}_2) $$
The condition for $s$ and $t$ to be equivalent when processing the formula is $n=0$, meaning the path that does not enter the loop
is reachable via $D_{\beta}$. By combining $D_{\alpha}$ and $D_{\beta}$, the formula has no solution, indicating that there is no
path between $s$ and $t$ that is reachable via both $D_{\alpha}$ and $D_{\beta}$.

An improved version of the SPDS method is the mutual refinement algorithm~\cite{ding2023mutual, conrado2024better}, which removes edges
that do not contribute to the current round of Dyck reachability solving. However, such approaches do not fundamentally improve upon
the SPDS method; thus, counterexamples satisfying the constraints of Theorem~\ref{theorem:dyck-decidable} but still
causing false positives can be constructed.

\vspace{8pt}
\noindent{\bf{LCL-reachability~\cite{zhang2017context}}} A Linear Conjunctive Language (LCL) describes the language formed by the intersection of linear
context-free languages, which provide a precise formulation of Interleaved Dyck languages. However, its reachability solving algorithm is based on
finite summaries, which can lead to imprecision due to the limitations of finite summaries and the inability to distinguish
different paths between nodes. For example, if the paths between two nodes $s$ and $t$ include $p_1a$ and $bp_2$, the summary for
$s$ and $t$ will be derived from the summaries of $p_1$ and $p_2$, even if $p_1$ and $p_2$ may not on the same path. In the following
example, $\Sigma_{\alpha} = \{ \llparenthesis_{1}, \rrparenthesis_{1} \}$, $ \Sigma_{\beta} = \{ \llbracket_{1}, \rrbracket_{1} \} $
and the path expression between nodes $s$ and $t$ is:
$$ \llparenthesis_{1} \llbracket_{1} (\rrparenthesis_{1} \rrbracket_{1})^{*}  \rrbracket_{1} $$
there exists a path $p_1 \cdot \rrbracket_{1} $ where $p_1 = \llparenthesis_{1} \llbracket_{1} \rrparenthesis_{1} \rrbracket_{1}$,
and a path  $\llparenthesis_{1} \cdot p_2$ where $p_2 = \llbracket_{1} \rrbracket_{1}$.
Both $p_1$ and $p_2$ are feasible paths, and their summaries are the same, indicating the existence of an interleaved Dyck
reachable path. However, according to the summary generation rules of LCL-reachability algorithm~\cite{zhang2017context}, two reachable
summaries also generate a summary indicating reachability. In this case, there is no interleaved Dyck reachable path between $s$ and
$s$, so LCL-reachability algorithm draws an incorrect conclusion.

The following explains how to solve this example using EUF$^n$. Similarly, we first define a variable $n$ to represent the number of
loop executions, and then proceed with the encoding for $D_{\alpha}$, resulting in the following:
$$ s=f(v_1) \wedge v_1 = v_{2\_\text{in}} \wedge v_{2\_\text{out}} = f^{n} (v_{2\_\text{in}}) \wedge t = v_{2\_\text{out}}.$$
By solving the formula, we obtain that $s$ and $t$ are equivalent when $n=1$. And then proceed with the encoding for
$D_{\beta}$, resulting in the following:
$$ s^{\prime}=v^{\prime}_1 \wedge v^{\prime}_1 =g(v^{\prime}_{2\_\text{in}}) \wedge v^{\prime}_{2\_\text{out}} = g^{n} (v^{\prime}_{2\_\text{in}}) \wedge t^{\prime} = g(v^{\prime}_{2\_\text{out}}).$$
By solving the formula, we obtain that $s$ and $t$ are equivalent when $n=0$. By combining $n=1$ and $n=0$, we obtain no solution,
thus concluding that there is no interleaved Dyck reachable path between $s$ and $t$.

\vspace{8pt}
\noindent{\bf{ILP-based~\cite{li2023single}}} When solving the single-source single-target IDRP, \cite{li2023single} derived path expressions between nodes and used ILP to count the number of parentheses and distinguish paths
with different loop executions. However, their approach only considered the Parikh images of the corresponding languages and
could not encode the relative order of parentheses. Although \cite{li2023single}  used LCL reachability to check the order
of parentheses, as mentioned earlier, LCL reachability cannot distinguish between different paths. Therefore, this check cannot fully
prevent errors in parentheses order.

For example, given a Dyck reachability instance where $\Sigma = \{ \llparenthesis_{1}, \llparenthesis_{2}, \rrparenthesis_{1}, \rrparenthesis_{2} \}$,
the path expression between nodes $s$ and $t$ is:
$$ \llparenthesis_{2} ( \llparenthesis_{1} )^{*} \rrparenthesis_{2} \rrparenthesis_{1} $$
In this case, the ILP method assigns a variable $n$ to represent the number of loop executions and requires that the counts of $\llparenthesis_{1}$
and $\rrparenthesis_{1}$ be equal, resulting in the constraint $n=1$. Similarly, it requires the counts of $\llparenthesis_{2}$ and $\rrparenthesis_{2}$
to be equal, resulting in the constraint $1=1$, which is satisfiable, yielding $n=1$. However, the corresponding path's parentheses order
is incorrect, $ \llparenthesis_{1} $ cannot match with $ \rrparenthesis_{2} $. Since LCL cannot distinguish between different paths between nodes,
it will generate a valid summary as long as there exists any path with a valid parentheses order (\emph{e.g.}, $ \rrparenthesis_{2} \rrparenthesis_{1} $ ).
As a result, LCL fails to identify the unique path with matching parentheses counts and incorrectly reports that $s$ and $t$ are reachable,
despite the invalid parentheses order.

The following explains how to solve this example using EUF$^n$. First, assign a loop execution count $n$ to the loop,
and then proceed with the encoding for $D_{\alpha}$, resulting in the following:
$$ s=f_2(v_{1\_\text{in}}) \wedge v_{1\_\text{in}} = f_1^{n}( v_{1\_\text{out}}) \wedge v_2 = f_2(v_{1\_\text{out}}) \wedge t = f_1( v_2 ).$$
By solving the formula, we find that no conditional equivalence edge exists between $s$ and $t$, thus concluding that there is no interleaved
Dyck reachable path between them.

%% file: figure/interleaved_ex.tex
\tikzstyle{squarednode} = [rectangle, rounded corners, draw = gray, fill=gray!5, very thick, text centered, align=center]%, minimum size=5mm]
\tikzstyle{originnode} = [circle, rounded corners, draw = black, fill=white!5,  text centered, align=center, minimum size=7mm, font=\bfseries\Large]
\tikzstyle{arrow} = [->,>=stealth, draw=brown, line width=0.2mm]
\tikzstyle{line} = [-,>=stealth, draw=red]

{\scriptsize
\begin{tikzpicture}

    \node[originnode] (a){$a$};
    \node[originnode, right of=a, node distance=19mm] (b){$b$};
    \node[above of=b, node distance=15mm] (b_above){};
    \node[originnode, left of=b_above, node distance=8mm] (e){$e$};
    \node[originnode, right of=b_above, node distance=8mm] (f){$f$};
    \node[originnode, right of=b, node distance=19mm] (c){$c$};
    \node[originnode, right of=c, node distance=19mm] (d){$d$};

    \draw[arrow] (a.east) --  node[above, align=center]{$ \llparenthesis_{1} $} (b.west);
    % \draw[arrow] (b.west) .. controls +(up:10mm) and +(down:-10mm) ..  node[above, align=center]{$ \llbracket_{1} $} (b.east);
    \draw[arrow] (b.west) --  node[left, align=center]{$ \llparenthesis_{2} $} (e.south);
    \draw[arrow] (e.east) --  node[above, align=center]{$ \llbracket_{1} $} (f.west);
    \draw[arrow] (f.south) --  node[right, align=center]{$\rrparenthesis_{2}$} (b.east);
    \draw[arrow] (b.east) --  node[above, align=center]{$ \rrparenthesis_{1}  $} (c.west);
    \draw[arrow] (c.east) --  node[above, align=center]{$ \rrbracket_{1} $} (d.west);

\end{tikzpicture}
}

%% file: figure/encoding_UP.tex
% \begin{align*}
%     & x := y; \\
%     & z := k; \\
%     & while (x != z) \{   \\
%     & \quad  x := f(x);      \\
%     & \quad   z := g(z);      \\
%     &    \};        \\
%     & while (y != k) \{   \\
%     &  \quad   y := f(y);     \\
%     &  \quad   k := g(k);    \\
%     &  \};    \\
%     & assert(x = y);
%    \end{align*}
\begin{subfigure}[t]{0.4\textwidth}
\begin{align*}
  & \mathtt{x := y;} \\
  & \mathtt{z := k;} \\
  & \mathtt{while \; (x != z) \{}   \\
  & \quad  \mathtt{x := next(x);}      \\
  & \quad  \mathtt{z := add(z);}      \\
  & \mathtt{\}}        \\
  & \mathtt{while \; (y != k) \{}   \\
  & \quad  \mathtt{y := next(y);}     \\
  & \quad  \mathtt{k := add(k);}    \\
  & \mathtt{\}}    \\
  & \mathtt{assert(x = y)}
\end{align*}
\end{subfigure}
\begin{subfigure}[t]{0.4\textwidth}
  \begin{align*}
    % &  \exists x_1 \exists x_2  \exists y_0 \exists y_1 \dots \exists n \exists m \forall i \forall j.  \\
    &  \forall i \forall j. \\
    &  x_1 = y_0 \\
    &  \wedge z_1 = k_0  \\
    &  \wedge [ i < n \rightarrow next^{i}(x_1) \neq add^{i}(z_1) ]  \\
    &  \wedge x_2 = next^{n}(x_1) \wedge z_2 = add^{n}(z_1)  \\
    &  \wedge x_2 = z_2 \\
    &  \wedge [ j < m \rightarrow next^{j}(y_0) \neq add^{j}(k_0) ]  \\
    &  \wedge y_1 = next^{m}(y_0) \wedge k_1 = add^{m}(k_0) \\
    &  \wedge y_1 = k_1 \\
    &  \wedge x_2 \neq y_1
   \end{align*}
\end{subfigure}

%% file: section/sec5.tex
\vspace{-3pt}
\section{Related Work} \label{sec:section5}
The EUF theory is widely applied in program verification, such as in verifying pipeline processor equivalence~\cite{1997Automatic}
and regression verification for C programs~\cite{2012Regression}. The satisfiability of its quantifier-free fragment is decidable.
Early algorithms by Shostak~\cite{congruence_closure3}, Nelson and Oppen~\cite{congruence_closure2}, and Downey \emph{et al.}~\cite{congruence_closure1} leverage directed acyclic graph (DAG) and \textit{Union-Find} data structure to achieve $O(n \log n)$ congruence closure complexity. Bachmair \emph{et al.}~\cite{bachmair2003abstract, bachmair2000abstract} later unified these algorithms via a rewriting framework. Nieuwenhuis and Oliveras~\cite{2005Proof, 2007Fast} introduced an incremental algorithm with an \emph{Explain} feature, now standard in SMT solvers, which outputs a small subset of constraints implying term equivalence. Congruence closure algorithms are also vital in automated theorem proving~\cite{de2007efficient} and compiler optimization~\cite{tate2009equality}. The Egg library~\cite{willsey2021egg}, implementing equality saturation~\cite{tate2009equality}, enhances performance via a \emph{rebuilding} mechanism that amortizes consistency maintenance costs.

In extending EUF, Bryant \emph{et al.}~\cite{PositiveEquality} introduced Positive EUF (PEUF), which defines positive terms. While PEUF does not increase EUF's expressive power, it significantly improves solving efficiency. Additionally, Bryant \emph{et al.}~\cite{CLU} proposed CLU (Counter arithmetic with Lambda expressions and Uninterpreted functions), which extends EUF by allowing integer-valued domains for functions and predicates. This extension is orthogonal to our work, which introduces \emph{parametric composition depth} for unary functions. Additionally, Zakhour \emph{et al.}~\cite{zakhour2025dis} extended e-graphs (an implementation of congruence graphs) with a unified framework for handling both equalities and disequalities. Deepak~\cite{kapur2019conditional} introduced conditional congruence closure using Horn equations, where target equalities hold only under conditions expressed as conjunctions of EUF equalities.

On the other hand, the satisfiability definition of EUF$^n$ formulas (Definition~\ref{thm:eufn-sat-equivalence}) implies that enumerating all
assignments to the natural number variables yields an infinite set of EUF formulas, with the EUF$^n$ formula being satisfiable
if and only if at least one formula in this set is satisfiable. This is closely related to the concept of \emph{disjunctive satisfiability}
introduced by Mathur \emph{et al.}~\cite{mathur2025decision}, which studies the existence of a model satisfying at least one formula in
a given infinite formula set. While their work focuses on regular first-order theories with formula sets described by regular tree languages,
EUF$^n$ can express non-regular terms such as $f^n(g^n(h^n(x)))$, which cannot be captured by regular tree languages. Thus, our work
addresses the \emph{disjunctive satisfiability} of a non-regular EUF theory, extending beyond the regular
framework of~\cite{mathur2025decision}.

%% file: section/sec6.tex
\vspace{-6pt}
\section{Conclusion} \label{sec:section6}
This paper introduces EUF$^n$, a decidable extension of EUF theory that supports \emph{parametric composition depth} for unary
functions (\emph{e.g.}, $f^n(x)$), significantly enhancing the theory's expressive power. We also present a \emph{CCG} algorithm
designed to solve quantifier-free EUF$^n$ formulas by reducing them to existential sentence of Presburger arithmetic with divisibility.
This satisfiability-preserving transformation establishes the decidability of the satisfiability of quantifier-free EUF$^n$ formulas,
as the corresponding Presburger fragment is known to be decidable. Furthermore, we demonstrate the applications of EUF$^n$ through
two practical scenarios: the single-source single-target IDRP and the verification of uninterpreted
programs, for which we identify novel decidable subclasses.

\section*{Acknowledgments}
 This research was supported by the National Key R\&D Program of China (No. 2022YFB4501903) and the NSFC Program (No. 62172429 and U2341212). We also thank the anonymous reviewers for their valuable feedback.
% \titlenote{This research was supported by the National Key R\&D Program of China (No. 2022YFB4501903) and the NSFC Program (No. 62172429 and U2341212). We also thank the anonymous reviewers for their valuable feedback.}

%% file: section/appendix.tex
% \newpage
% \section{Appendix}\label{sec:appendix}

%%=============================================================== axioms ===============================================================
\section{Axioms of Presburger Arithmetic~\cite{mendelson2009introduction}} \label{sec:appendix_presburger}
Let the language $\mathcal{A}_0 = \{ 0, S, + \}$, where $0$ is a constant symbol, $S$ is a unary function symbol,
$+$ is a binary function symbol. The axioms of Presburger arithmetic are defined as the universal closures of the following formulas:

\begin{enumerate}
    \item $x=y \rightarrow (x=z \rightarrow y=z)$
    \item $x=y \rightarrow S(x)=S(y)$
    \item $\neg( S(x) = 0)$
    \item $S(x)=S(y) \rightarrow x = y$
    \item $x+0 = x$
    \item $ x+S(y) = S(x+y) $
    \item The above axioms define the properties of equality (axioms (1) and (2)), the successor function, and the
    addition function. Additionally, Presburger arithmetic includes an infinite axiom schema, the induction axiom:
    Let $P(x)$ be a first-order formula in the language $\mathcal{A}_0$, where $x$ is a free variable of $P$. The
    induction axiom is formulated as follows:
    \[
        (P(0) \land \forall x \, (P(x) \rightarrow P(S(x)))) \rightarrow \forall y \, P(y).
    \]
\end{enumerate}

The less-than predicate $<$ is definable in the language $\mathcal{A}_0$ by the formula $x < y \Leftrightarrow \exists z\, (\neg(z=0) \land x+z=y)$.

%%=============================================================== correctness proof  ===============================================================
\section{Details of the proof of Lemma~\ref{lemma:correctness}.}   \label{sec:appendix_correct}
\begin{lemma}(Lemma~\ref{lemma:correctness})
    During Algorithm~\ref{algo:solve}, for every equivalence edge $(term_l, \phi, term_r)$ in the \emph{CCG} and any assignment
    $\sigma: \mathsf{N}\to\mathbb{N}$, the following invariant holds: $\phi$ is satisfiable under $\sigma$
    if and only if $term_l$ and $term_r$ belong to the same equivalence class in the congruence closure of $\Psi[\sigma]$.
\end{lemma}
\begin{proof}
    Given a conjunct of an EUF$^n$ formula, we prove by induction that after processing each equation, the \emph{CCG} and the
    congruence closure of $\Psi[\sigma]$ obtained at that step satisfy the invariant stated in Lemma~\ref{lemma:correctness}.

    \begin{itemize}
        \item \textbf{ Base Step  }  \\
        The initialization rules (Section~\ref{section3:init}) show that equivalence edges in $E_{\equiv}$ are of the form
        $(term_1, k_1 = k_2, term_2)$, where $term_1 = f^{k_1}(term)$ and $term_2 = f^{k_2}(term)$. When $\sigma(k_1) = \sigma(k_2)$,
        the terms $term_1$ and $term_2$ are identical and thus belong to the same equivalence class in the congruence closure of
        $\Psi[\sigma]$. Conversely, in the EUF formula $\Psi[\sigma]$ obtained under any assignment $\sigma$, terms of the form
        $f^{c}(term)$ (where $c$ is a constant) are treated identically during initialization. Setting $k_1 = c$ and $k_2 = c$ yields
        a satisfying assignment for $k_1 = k_2$.
        \item \textbf{ Inductive Step}  \\
        % Assume the invariant of Lemma~\ref{lemma:correctness} holds after the first $k$ equations. For the $(k{+}1)$-th equation,
        % the algorithm adds equivalence edges via \texttt{EqProp}, \texttt{EqProp} adds equivalence edges to the \emph{CCG} in either
        % Line~\ref{algo:call_EqProp} or Line~\ref{algo:eq_end}. Specifically, \texttt{EqProp} uses the \texttt{AddEq} function to
        % add equivalence edges to the \emph{CCG}. In Line~\ref{algo:call_AddEq}, \texttt{AddEq} adds the parameters of \texttt{EqProp}
        % to the \emph{CCG}; these parameters originate from two sources: (1) the $(k{+}1)$-th equality itself (Line~\ref{algo:call_EqProp}),
        % which clearly satisfies the invariant of Lemma~\ref{lemma:correctness}; and (2) the return value of \texttt{CongruenceClosure}, whose
        % correctness will be proven subsequently.

        Assume the invariant of Lemma~\ref{lemma:correctness} holds after the first $k$ equations. For the $(k{+}1)$-th equation,
        the algorithm adds equivalence edges via \texttt{EqProp}, which may be invoked multiple times per equation. We therefore
        prove by induction on the number of \texttt{EqProp} executions that the invariant still holds after processing the
        $(k{+}1)$-th equation.

        Initially, when starting the $(k{+}1)$-th equation, the invariant holds by the induction hypothesis. Now assume it holds
        after the first $k'$ executions of \texttt{EqProp} for this equation. When executing the $(k'{+}1)$-th \texttt{EqProp},
        it adds equivalence edges to the \emph{CCG} in either Line~\ref{algo:call_EqProp} or Line~\ref{algo:eq_end}.
        In Line~\ref{algo:tran_end}, to maintain transitivity of the equivalence relation, \texttt{AddEq} adds the equivalence
        edge $(term_1, \phi_{\text{old}} \vee (\phi_{\text{new}} \land \phi_1 \land \phi_2), term_2)$ to the \emph{CCG}
        (from Lines~\ref{algo:tran_end},~\ref{algo:find}, and~\ref{algo:eq_add}), where $term_1 \in [term_l]$ and $term_2 \in [term_r]$.
        By the induction hypothesis, the edges $(term_l, \phi_1, term_1)$, $(term_r, \phi_2, term_2)$, and
        $(term_1, \phi_{\text{old}}, term_2)$ satisfy the invariant of Lemma~\ref{lemma:correctness}. As established previously,
        $(term_l, \phi_{\text{new}}, term_r)$ also satisfies the invariant of Lemma~\ref{lemma:correctness} (\emph{i.e.}, the two parameter sources
        of \texttt{EqProp} in Lines~\ref{algo:call_EqProp} and~\ref{algo:eq_end}). By the transitivity of the equivalence relation,
        $term_1$ and $term_2$ are also equivalent under $\phi_{\text{new}} \land \phi_1 \land \phi_2$. Hence, the updated edge
        $(term_1, \phi_{\text{old}} \vee (\phi_{\text{new}} \land \phi_1 \land \phi_2), term_2)$ satisfies the invariant of Lemma~\ref{lemma:correctness}.
    \end{itemize}
\end{proof}

%%=============================================================== th1 proof  ===============================================================
\section{Details of the proof of Lemma~\ref{lemma1} and Theorem~\ref{theorem:th1}.}   \label{sec:appendix_th1}
\begin{lemma}(Lemma~\ref{lemma1})
    Assume $n > m$, $ term = f^{m}(term) = f^{n}(term) \leftrightarrow term = f^{\mathsf{gcd}(m,n)}(term) $.
\end{lemma}
\begin{proof}
    As the implication from right to left is obvious, we now prove the implication from left to right,
    \emph{i.e.}, $ term = f^{m}(term) = f^{n}(term) \rightarrow term = f^{\mathsf{gcd}(m,n)}(term) $. Here we assume $m<n$.
    By Theorem~\ref{theorem:Gen_Congruence}, the identity $term = f^m(term)$ implies $f^{n-m}(term) = f^n(term)$, and thus
    by the transitivity of the equivalence relation:
        $$ term = f^{n-m}(term) = f^{m}(term) = f^{n}(term) $$
    If $ n-m \neq m $, we can can repeatedly apply Theorem~\ref{theorem:Gen_Congruence} until $ term = f^{a}(term) = f^{b}(term) \ s.t. \ a=b $ is derived.
    This process utilizes a variant of the Euclidean algorithm~\cite{rosen2011elementary} to find the greatest common divisor of $m$ and
    $n$. Thus, $ a = b = \mathsf{gcd}(m, n) $, and the formula on the right $ term = f^{\mathsf{gcd}(m,n)}(term) $ is derived.
\end{proof}

\begin{theorem}(Theorem~\ref{theorem:th1})
Given two periodic equivalence relations $ term = f^{k}(term) $ and $ f^{m}(term) = f^{n}(term) $.
Then
$$ term = f^{k}(term) \wedge f^{m}(term) = f^{n}(term) \leftrightarrow term = f^{\mathsf{gcd}(k, (n-m))}(term). $$
\end{theorem}
\begin{proof}
We will prove the theorem from two directions. Without loss of generality, we assume that $n > m$.
    \begin{itemize}

        \item Right to Left: \\
            Let $c_1 = \mathsf{gcd}(k,(n-m))$. Since $c_1$ is a divisor of $k$, we have
            $$ term = f^{c_1}(term) \rightarrow  term = f^{k}(term), $$
            and since $c_1$ is also a divisor of $(n-m)$, we have
            $$ term = f^{c_1}(term) \rightarrow term = f^{(n-m)}(term). $$
            Applying $f$ repeatedly for $m$ times to both sides of the above equation, we obtain
            $$ f^{m}(term) = f^{n}(term). $$
            Thus, the right-to-left direction is established.

        \item Left to Right: \\
            From the formula on the left, we can obtain
            % $$ term = f^{k}(term) \;\wedge\; f^{m}(term) = f^{n}(term), $$
            $$ term = f^{c_1k}(term), \quad s.t. c_1=0,1,\cdots, $$
            $$ f^{m+c_2}(term) = f^{n+c_2}(term), \quad s.t. c_2=0,1,\cdots. $$
            For any given $n$ and $k$, when $c_1$ is sufficiently large, there must exist a $c_2 > 0$ such that
            $c_1 k = n + c_2$, that is, there exist $c_1$, $c_2$ such that
            $$ term = f^{k}(term) = f^{c_1 k}(term) = f^{n+c_2}(term) = f^{m+c_2}(term). $$
            By Lemma \ref{lemma1}, it follows that $$ term = f^{\gcd(k, m+c_2)}(term). $$
            From the equation $c_1 k = n + c_2$, it follows that this equality can also be expressed as
            $$ term = f^{\gcd(k,\, c_1k-(n-m))}(term). $$
            Since $\gcd(m,n) = \gcd(m,n-m)$ and $\gcd(m,n) = \gcd(m,-n),$
            we have $$ \gcd(k,\, c_1k-(n-m)) = \gcd(k, n-m). $$
            Hence, the left-to-right direction holds, and the proof is thus complete.
    \end{itemize}
\end{proof}

%=============================================================== arXiv  ===============================================================

% %% =============================================================== Theorem and EncodingIDR  ===============================================================
\section{Encoding single-source single-target Interleaved Dyck Reachability with EUF$^n$}   \label{sec:appendix_encoding_Dyck}

\input{algorithms/algo_encodingIDR.tex}
In this section, we will illustrate Algorithm~\ref{algo:encoding_regex} using the example in Figure~\ref{fig:interleaved_examp},
which verifies $D_{\alpha} \odot D_{\beta}$-reachability between nodes $a$ and $d$. The path expression
$ \llparenthesis_{1} (\llparenthesis_{2}\llbracket_{1}\rrparenthesis_{2})^{*} \rrparenthesis_{1} \rrbracket_{1} $ is first derived and then encoded
using the algorithm. In Line~\ref{algo3:init}, the EUF$^n$ signature $\mathsf{\Sigma}_{\text{EUF}^n} = (\mathsf{U} \cup \mathsf{N}, \mathsf{F}, \mathsf{C}, \{ =, \neq \}) $
is initialized, where $\mathsf{U}$ and $\mathsf{N}$ represent variables of uninterpreted and natural number sorts, respectively,
and $\mathsf{F}$ contains only unary functions. In Line~\ref{algo3:init_phi}, the EUF$^n$ formula $\Psi$ (initialized to $\mathbf{true}$)
accumulates the encoding results, while the variables $v_{\text{old}}$ and $v'_{\text{old}}$ are used to track variable name in EUF$^n$ during
the encoding process. In Line~\ref{algo3:def_pi}, the mappings $\pi: \Sigma_\alpha \rightarrow \mathsf{F} \cup \{\mathrm{Id}\}$ and
$\pi^{\prime}: \Sigma_\beta \rightarrow \mathsf{F} \cup \{\mathrm{Id}\}$ assign terminal symbols to uninterpreted functions. Paired terminal symbols
(\emph{e.g.}, $\llparenthesis_{1}/\rrparenthesis_{1} \in \Sigma_\alpha$) are mapped to the same function, and $\epsilon$ is mapped to the identity
function. Line~\ref{algo3:fun_call} invokes the \texttt{Encoding} function to recursively encode $\mathcal{R}$, and finally,
Line~\ref{algo3:return} returns the encoded result $\Psi$.

Lines~\ref{algo3:start}-\ref{algo3:end} recursively encode $\mathcal{R}$, with $\Psi_r$ in Line~\ref{algo3:ini_phi_r} recording the formula generated
in each recursive call. For conciseness, syntactically similar \texttt{if}-branches are merged and distinguished via black/gray visual coding.
During the initial \texttt{Encoding} call, assume $v_{\text{old}}=a$, $v^{\prime}_{\text{old}}=a^{\prime}$, and
$\mathcal{R} = \llparenthesis_{1} (\llparenthesis_{2}\llbracket_{1}\rrparenthesis_{2})^{*} \rrparenthesis_{1} \rrbracket_{1}$, leading to execution of
Lines~\ref{algo3:recu1}-\ref{algo3:recu2} that recursively process \texttt{Encoding}$(\llparenthesis_{1})$ and
\texttt{Encoding}$((\llparenthesis_{2}\llbracket_{1}\rrparenthesis_{2})^{*} \rrparenthesis_{1} \rrbracket_{1})$.

When processing \texttt{Encoding}$(\llparenthesis_{1})$, the terminal symbol triggers Line~\ref{algo3:terminal_start}, updating $v_{\text{new}}$ and
$v^{\prime}_{\text{new}}$ in Line~\ref{algo3:v_new} (hypothetically to $b_{\text{in}}$/$b^{\prime}_{\text{in}}$). As $\llparenthesis_{1}$ is an opening label
in $\Sigma_{\alpha}$, Line~\ref{algo3:opening} activates and Line~\ref{algo3:encode_opening} updates $\Psi_r$ with two components encoding
$D_{\alpha}$/$D_{\beta}$ reachability. Given $\llparenthesis_{1} \in \Sigma_{\alpha}$, $\pi(a)$ becomes an uninterpreted function while
$\pi^{\prime}(a)$ simplifies to an identity function (denoted as equality, omitting explicit identity notation). Gray-shaded code handles
the $\Sigma_{\beta}$ opening label case analogously. Finally, Line~\ref{algo3:update_v_old} updates $v_{\text{old}}$/$v^{\prime}_{\text{old}}$ to
$b_{\text{in}}$/$b^{\prime}_{\text{in}}$, tracking variable's name during encoding.

When processing \texttt{Encoding}$((\llparenthesis_{2}\llbracket_{1}\rrparenthesis_{2})^{*} \rrparenthesis_{1} \rrbracket_{1})$, the
algorithm activates the Line~\ref{algo3:recu1} branch and computes \texttt{Encoding}$((\llparenthesis_{2}\llbracket_{1}\rrparenthesis_{2})^{*})$
and \texttt{Encoding}$(\rrparenthesis_{1} \rrbracket_{1})$ in Line~\ref{algo3:recu2}. For the Kleene Star operator in
\texttt{Encoding}$((\llparenthesis_{2}\llbracket_{1}\rrparenthesis_{2})^{*})$, Line~\ref{algo3:kleene_star} triggers with Lines~\ref{algo3:select_n}-\ref{algo3:encoding_star}
introducing a bound $n$ and encoding via the \texttt{EncodingStar} function using EUF$^n$ formulas. Per Theorem~\ref{theorem:dyck-decidable}, removing matched
$D_{\alpha}$/$D_{\beta}$ substrings from $\mathcal{R}^n_1$ leaves three EUF$^n$-encodable forms:
(1) entirely opening labels;
(2) entirely closing labels;
(3) closing label prefixes followed by opening label suffixes.
In this case, $(\llparenthesis_{2}\llbracket_{1}\rrparenthesis_{2})^{*}$ yields $\epsilon$ for $D_{\alpha}$ and $\llbracket^n_{1}$ for $D_{\beta}$ after
match removal, resulting in the formula $b_{\text{in}} = b_{\text{out}} \land b^{\prime}_{\text{in}} = g_1^n(b^{\prime}_{\text{out}})$.
We omit \texttt{EncodingStar}'s straightforward details. Finally, \texttt{Encoding}$(\rrparenthesis_{1} \rrbracket_{1})$ handles closing labels via
Lines~\ref{algo3:closing}-\ref{algo3:encode_closing} using logic analogous to opening label encoding.

\section{Semantics and Verification of Uninterpreted Programs}  \label{sec:appendix_background}
The semantics of an uninterpreted program are defined over a given data model $ \mathcal{M}=(\mathcal{U}, \mathcal{I}) $,
where $\mathcal{U}$ represents the set of all elements, and $ \mathcal{I}  $ provides an interpretation for the symbols in
$ \Sigma $. For example, $\mathcal{I}$ maps functions in $\Sigma_{F}$ to relations in $\mathcal{U}$. Thus, the
semantics of an uninterpreted program can be defined as a state transition graph over such a data model, where each state is defined as
a mapping $ \mathcal{S}:\Sigma_{V} \rightarrow \mathcal{U} $, which maps variables in the program to elements in the data model.
In the initial state, this mapping is empty, denoted as $ \emptyset $, and a state for assertion failure is introduced,
denoted as $Fail$. Each state can be updated to a new state by the program's statements.

The verification problem of an uninterpreted program involves checking the satisfiability of each assert statement
$ assert(\stdef{cond})$ in the program. This requires that the condition $ \stdef{cond} $ holds in any data model during
the execution of the program, i.e., $ \neg \stdef{cond} $  is not satisfiable. Therefore, the verification problem of a program
$ \mathcal{P} $ can be reduced to a reachability problem, defined as follows (where $ \not\rightsquigarrow_{\mathcal{P}} $ denotes
an execution of program $ \mathcal{P} $):
$$ \forall \mathcal{M}. \stdef{ \mathcal{M}, \emptyset } \not\rightsquigarrow_{\mathcal{P}} \stdef{ \mathcal{M}, Fail } . $$
Here, the tuple $ \stdef{ \mathcal{M} , \mathcal{S} } $ represents the current data model $\mathcal{M}$ and the current program
state $\mathcal{S}$. In the equation, $ \emptyset $ denotes the initial state, and $Fail$ denotes the failure state. Thus,
the equation signifies that during the execution of the program, the $Fail$ state is unreachable in any data model $\mathcal{M}$.
Generally, the reachability problem for uninterpreted programs is undecidable. Recent work \cite{mathur2019decidable}, however,
has shown that the verification problem for uninterpreted programs that satisfy a property called $coherence$ is decidable,
and the complexity of the decision algorithm is PSPACE-complete.

%%===============================================================  Encoding UP  ===============================================================
\section{Encoding Uninterpreted Program Verification with EUF$^n$}  \label{sec:appendix_encoding_UP}
This section formalizes the encoding method of these programs into \textsc{EUF}$^n$ formulas, with loop encoding representing
the most challenging aspect. While loop summaries can be readily identified in specific cases (\emph{e.g.}, Figure~\ref{fig:exam_program}),
they often require non-trivial derivation, as illustrated in subsequent examples (\emph{e.g.}, Example~\ref{example3}). We will start
by presenting the following fundamental theorem:

\begin{theorem}[Theorem~\ref{theorem:UP_encode}]
  The verification problem for uninterpreted programs without nested loops, conditional branches within loops, or multi-ary
  functions in loop bodies can be encoded as the satisfiability problem of EUF$^n$ formulas.
\end{theorem}

\begin{proof}
     We now prove the theorem by presenting an encoding method for such programs.
\end{proof}

We now present an approach for encoding uninterpreted program verification problems using EUF$^n$. While the encoding of sequential and
branching structures is relatively straightforward—for example, the program statement $\mathtt{x := f(y)}$ can be directly encoded as
$x_0 = f(y_0)$, where $x_0$ and $y_0$ are the EUF$^n$ variables corresponding to the current values of $\mathtt{x}$ and $\mathtt{y}$
at that program point (with formula variables being updated upon variable redefinition)—the encoding of loops is nontrivial. This section
establishes the feasibility of loop encoding within this framework. The complete encoding rules and corresponding algorithm are
provided later. We begin by introducing the notion of a $\mathit{block}$, defined by
the following grammar:
\[
\mathit{block} ::= \ \mathtt{skip} \mid \mathtt{x := y} \mid \mathtt{x := f(y)} \mid \mathit{block} \mathtt{; } \mathit{block}
\]
A $\mathit{block}$ is a code segment consisting of \texttt{skip} statements and two types of basic assignment statements.
Since a $\mathit{block}$ is intended to represent the body of a loop, the assignment statements in this definition involve only unary functions.
Note that \texttt{assume} and \texttt{assert} statements are not included in this definition.

% \vspace{-10pt}
% \begin{figure}[ht!]
%       \begin{center}
%       {
%       \begin{tabular}{rcl}
%       $ $\\
%       \end{tabular}
%       }
%       \end{center}
%       \vspace{-10pt}
%       \caption{The syntax of $\mathit{block}$.}
%       \vspace{-10pt}
%       \label{fig:block-syntax}
% \end{figure}

\begin{myDef}[Program-to-Formula Variable Mapping]\label{def:gamma}
    At each program point $\mathrm{loc}$, the mapping $\gamma_{\mathrm{loc}}: \Sigma_V \to \mathsf{U}$ assigns each program variable
    $\mathtt{var} \in \Sigma_V$ a unique formula variable $v \in \mathsf{U}$ (defined in Section~\ref{section2:EUFn_axiom}), translating
    concrete program states into symbolic variables for constraint expressions.
\end{myDef}

\begin{myDef}[Block Summary Function]\label{def:fun_block}
    Let $BLOCK$ denote the set of basic blocks.
    The function $\zeta : BLOCK \times \Sigma_V \rightarrow \mathsf{F} \times \Sigma_V$ maps
    a block $\mathit{b} \in BLOCK$ and variable $\mathtt{v} \in \Sigma_V$ to a pair $(f, \mathtt{v'})$.
    Denote $\gamma_{\mathrm{init}}$ and $\gamma_{\mathrm{final}}$ as the variable mappings
    before entering and after executing $\mathit{b}$, respectively. Then:
    \begin{align*}
        \gamma_{\mathrm{final}}(\mathtt{v}) = f(\gamma_{\mathrm{init}}(\mathtt{v'}))
    \end{align*}
    Thus, $\zeta$ computes $BLOCK$'s aggregate effect on $\Sigma_V$ variables and $f \in \mathsf{F}$ (defined in Section~\ref{section2:EUFn_axiom})
    is the summary function. For clarity, we write the resulting pair $(f, \mathtt{v'})$ in the form $f(\mathtt{v'})$.
\end{myDef}

\begin{myExample}     \label{example3}
     \begin{figure}[h]
        \vspace{-8pt}
            \centering
            \input{figure/encoding_block.tex}
            \vspace{-8pt}
            \caption{Compositional Effects of $\mathit{block}^n$ on Program Variables.}
            \label{fig:encoding_block}
            \vspace{-8pt}
      \end{figure}
      This example presents a block $\mathit{block}$ with its corresponding $\zeta$ function, along with $\zeta$ functions for
      $\mathit{block}$ repeated $1$ to $4$ times (denoted $\mathit{block}^1, \dots, \mathit{block}^4$). As shown in
      Figure~\ref{fig:encoding_block}, each $\zeta$ computes a block's effect on specific variables. We observe that as
      repetition count increases, $\zeta$'s summary functions for variables like \texttt{x} exhibit periodicity. Specifically,
      compositions alternate between functions $g$ and $h \circ f$, with parameters \texttt{y} and \texttt{z} appearing alternately.
      This behavior stems from the block summary in Figure~\ref{subfig:summary_block}:
      $$
      \zeta(\mathit{block}^4,\mathtt{x}) = f( \zeta(\mathit{block}^3,\mathtt{y}) ) = (f \circ g)( \zeta(\mathit{block}^3,\mathtt{z}) ) = \cdots
      $$
      The finite set of summary functions causes periodically growing compositions during unrolling. Variables \texttt{y} and
      \texttt{z} follow similar patterns. Theorem~\ref{theorem_period} formalizes this phenomenon with proof.
\end{myExample}

% \begin{theorem}   \label{theorem_period}
%       For any $\mathit{block} \in BLOCK$, and for any $ \mathtt{x} \in \Sigma_{V} $, there exist a constant $k$ and a constant $p$,
%       as well as functions $F_j \in \mathsf{F} \ s.t.\  (0 \leq j < p)$, and a variable $y$, such that for any $i \geq k$,
%       $ \zeta( \mathit{block}^{i+cp}, x) = F_{j}^{c}( \zeta( \mathit{block}^{i}, y) ) \ s.t. \ c=0,1,\cdots  $.
% \end{theorem}

\begin{myThe}[Periodic Summary Composition] \label{theorem_period}
    For any basic block $\mathit{block} \in BLOCK$ and variable $\mathtt{x} \in \Sigma_{V}$, there exist constants $k, p \in \mathbb{N}$,
    functions $\{f_j \in \mathsf{F} \mid 0 \leq j < p\}$, and variable $y \in \Sigma_V$,
    such that for all integers $i \geq k$ and $c \geq 0$:
    $$
    \zeta(\mathit{block}^{i+cp}, \mathtt{x}) = (f_{0} \circ \cdots \circ f_{p-1})^{c}\left( \zeta( \mathit{block}^{i}, y) \right)
    $$
    where:
    \begin{itemize}
        \item $k$ is the pre-period constant: the minimal repetitions after which summary functions exhibit periodic composition
        \item $p$ is the period constant: the cycle length of identical function composition in subsequent repetitions
    \end{itemize}
\end{myThe}

\begin{proof}
    For any basic block $\mathit{block} \in BLOCK$ and variable $\mathtt{x} \in \Sigma_V$, there exists a variable $\mathtt{y} \in \Sigma_V$ such that:
    $$
    \zeta(\mathit{block}, \mathtt{x}) = f_{\mathtt{x}}(\mathtt{y})
    $$
    This captures the immediate effect of $\mathit{block}$ on $\mathtt{x}$, as shown in Figure~\ref{subfig:summary_block}. Extending to $\mathit{block}$ repeated $i$ times, we derive:
    $$
    \zeta(\mathit{block}^{i}, \mathtt{x}) = f_{\mathtt{x}}(\zeta(\mathit{block}^{i-1}, \mathtt{y}))
    $$
    Further expanding $\zeta(\mathit{block}^{i-1}, \mathtt{y})$ leads to:
    $$
    \zeta(\mathit{block}^{i}, \mathtt{x}) = f_{\mathtt{x}} \circ f_{\mathtt{y}} \circ \cdots \circ f_{\mathtt{z}}(\mathtt{w})
    $$
    where $\mathtt{w} \in \Sigma_V$ is the innermost parameter.

    Since the set of $\zeta$ functions is finite (bounded by $|\Sigma_V|$) and the composition of summary functions is
    deterministic, the expansion pattern must exhibit periodicity. Therefore, there exists $k \in \mathbb{N}$ where for
    all $i \geq k$, the function compositions become periodic. Let $p$ be the period length, then:
    $$
    \zeta(\mathit{block}^{i+cp}, \mathtt{x}) = (f_0 \circ \cdots \circ f_{p-1})^c (\zeta(\mathit{block}^{i}, \mathtt{y}))
    $$
    Here $\{f_j \in \mathsf{F}_{\mathsf{U}} \mid 0 \leq j < p\}$ represents the cyclic composition pattern observed during expansion.
\end{proof}

The validity of Theorem~\ref{theorem_period} enables the encoding of loops in uninterpreted programs using EUF$^n$.
We now present encoding rules for uninterpreted programs $P \in \mathcal{P}$ using $\text{EUF}^{n}$ formulas. We begin by defining the
signature $\Sigma_{\text{EUF}^n} = (\mathsf{U} \cup \mathsf{N}, \mathsf{F}, \mathsf{C}, \{ =, \neq \})$ where: $\mathsf{U}$ denotes
variables of uninterpreted type, $\mathsf{N}$ denotes variables of natural number type, $\mathsf{F}$ denotes function symbols,
$\mathsf{C}$ denotes constants, $\{=, \neq\}$ denote equality and disequality relations.
The complete encoding rules appear in Figure~\ref{fig:rule}. We specifically omit a separate rule for $\texttt{assert(cond)}$ statements
since they can be reduced to $\texttt{assume(cond)}$ statements and encoded using existing $\texttt{assume}$ rules.

\begin{figure}[h]
      \centering
        \input{figure/encoding_UP_rule.tex}
      \caption{Rules for encoding uninterpreted programs as $\text{EUF}^{n}$ formulas}
      \label{fig:rule}
\end{figure}

% In this section, we introduce the rules for encoding uninterpreted programs $P \in \mathcal{P}$ using $\text{EUF}^{n}$ formulas. First, we present the symbol set
% $ \mathsf{\Sigma}_{EUF} = (\mathsf{F}, \mathsf{U} \cup \mathsf{N}, \mathsf{C}, \{ =, \neq \}) $ for $\text{EUF}^{n}$ formulas, where $ \mathsf{U}$ represents the set of general
% variables, $\mathsf{N}$ represents the set of natural number variables, $ \mathsf{F} $ represents the set of function symbols, $\mathsf{C}$ denote the set of constants,
% and $ = $ and $ \neq $ denote equality and disequality, respectively. The encoding rules for uninterpreted programs are shown
% in Figure \ref{fig:rule}. In these rules, we do not provide an encoding rule for \texttt{assert(cond)} because \texttt{assert(cond)} can be
% transformed into \texttt{assume(cond)}, and thus can be encoded using the encoding rules for \texttt{assume} statement.

These encoding rules use a state tuple $\langle \gamma, pc \rangle$ where $\gamma: \Sigma_{V} \rightarrow \mathsf{U}$ (Definition~\ref{def:gamma})
maps program variables to formula variables, and $pc$ represents the path condition formula from program encoding. The notation $\mathtt{e} \downarrow_{\gamma}$
transforms program statement $\mathtt{e}$ by substituting variables according to $\gamma$ into a formula. For example,
$\mathtt{x := y}\downarrow_{\gamma[\mathtt{x} \mapsto x^{\prime}, \mathtt{y} \mapsto y^{\prime}]}$ yields $x^{\prime} = y^{\prime}$. Here
$\gamma[\mathtt{x} \mapsto x^{\prime}, \mathtt{y} \mapsto y^{\prime}]$ updates $\gamma$ to map $\mathtt{x}$ to $x^{\prime}$ while preserving other variable
mappings. The binary operator $\uplus$ maintains concurrent encodings for both branches of $\mathtt{if}$-statements. The function
$\texttt{EncodingLoop}(\gamma, \mathtt{Loop})$ handles loop encoding and returns the resulting formula.

\input{algorithms/algo_encodingloop.tex}

Algorithm~\ref{algo:encodeing_loop} (\texttt{EncodingLoop}) encodes loops into $\text{EUF}^{n}$ formulas, taking loop code segments as input and
producing encoded formulas. The algorithm first allocates a new variable for loop iteration count (Line~\ref{algo:loop_num}), then processes
the loop body $\mathtt{block}$ by removing $\mathtt{assume}$ statements (Line~\ref{algo:blook}), and initializes $\Psi$ to
$\mathbf{true}$ (Line~\ref{algo:Psi}).

The loop in Line~\ref{algo:forvar} processes each variable $\mathtt{var} \in \Sigma_{V}$, computing a quadruple per Theorem~\ref{theorem_period}.
Here $k_{\mathtt{var}}$ represents repetitions until periodic composition begins, $p_{\mathtt{var}}$ denotes period length of function composition,
$f_{k_{\mathtt{var}}}$ gives pre-period composition (\emph{e.g.}, $f$ for $\mathtt{x}$ in Example~\ref{example3}), and $f_{p_{\mathtt{var}}}$ provides
periodic composition (\emph{e.g.}, $g \circ h \circ f$ for $\mathtt{x}$ in Example~\ref{example3}).

Lines \ref{algo:inloop1}-\ref{algo:inloop2} encode loop entry conditions, focusing on the $\mathtt{x = y}$ case (the $\mathtt{x \neq y}$ case follows similarly).
Since loops may terminate mid-period, Line \ref{algo:inloop2} introduces $f_{s_{\mathtt{x}}}$ and $f_{s_{\mathtt{y}}}$ to represent partial function
compositions after $m$ iterations. For example, when $m=4$ in Example~\ref{example3}, $f_{s_{\mathtt{x}}} = g$.
These partial compositions depend on $m \bmod p_\mathtt{x}$ and $m \bmod p_\mathtt{y}$. Given constant $p_\mathtt{x}$ and $p_\mathtt{y}$ values,
both $f_{s_{\mathtt{x}}}$, $f_{s_{\mathtt{y}}}$ and their input mappings $\gamma(\mathtt{var})$ admit finite representations. Integer division operations are
expressible via existential Presburger arithmetic due to constant period lengths.

Lines \ref{algo:outloop1}-\ref{algo:outloop2} encode loop exit conditions using fresh variables $x^{\prime}, y^{\prime}$ not present in prior $pc$ contexts.
The $\mathtt{assume}$ statement encoding (Line \ref{algo:encodeassume}) mirrors loop condition encoding. For instance, encoding $\mathtt{assume(z = b)}$ after
$\mathtt{block}_1$ combines: (1) The summary of first $(i-1)$ $\mathtt{block}$ executions on $\mathtt{z}$ and $\mathtt{b}$ (2) $\mathtt{block}_1$'s effect on $\mathtt{z}$ and $\mathtt{b}$.

%% file: algorithms/algo_encodingIDR.tex
\begin{algorithm}[htbp]
    \caption{EncodingIDRP} \label{algo:encoding_regex}
    \KwIn{
        \begin{tabular}{l l}
            $\mathcal{R}$ & : A regular expression satisfying the constraint defined in Theorem~\ref{theorem:dyck-decidable}.  \\
            $\Sigma_\alpha, \Sigma_\beta$ & : Alphabets used in the IDRP instance.
        \end{tabular}
    }
    \KwOut{The EUF$^n$ formula $\Psi$ obtained from encoding the IDRP.}

    \SetKwFunction{FSubFunc}{Encoding} % 子函数名
    \SetKwProg{Fn}{Function}{:}{}

    Let $\Sigma_{\text{EUF}^n} = (\mathsf{U} \cup \mathsf{N}, \mathsf{F}, \mathsf{C}, \{=, \neq\})$ denote the EUF$^n$ signature.   \label{algo3:init} \\
    Let $\Psi \leftarrow \mathbf{true}$, let $v_{\text{old}}, v^{\prime}_{\text{old}} \leftarrow$  select two fresh variables from $\mathsf{U}$ \label{algo3:init_phi}. \\
    let $v_{\text{new}} \leftarrow v_{\text{old}}$, $v^{\prime}_{\text{new}} \leftarrow v^{\prime}_{\text{old}} $ . \label{algo3:v_ini} \\
    Let $\pi: \Sigma_\alpha \rightarrow \mathsf{F} \cup \{\mathrm{Id}\}$ and $\pi^{\prime}: \Sigma_\beta \rightarrow \mathsf{F} \cup \{\mathrm{Id}\}$ be
    mappings that assign function symbols to terminal symbols.             \label{algo3:def_pi} \\
    $\Psi \leftarrow \FSubFunc (\mathcal{R}) $   \label{algo3:fun_call} \\
    return ($\Psi$, ($v_{\text{new}}, v_{\text{old}}$), ($v^{\prime}_{\text{new}}, v^{\prime}_{\text{old}}$))       \label{algo3:return}

    \BlankLine
    \Fn{\FSubFunc {$\mathcal{R}$}}{                                \label{algo3:start}
        $ \Psi_r \leftarrow \mathbf{true} $                                 \label{algo3:ini_phi_r}                   \\
        \If{$\mathcal{R}$ is a terminal symbol $\mathbf{a}$   \label{algo3:terminal_start} }
         {
            $v_{\text{new}}, v^{\prime}_{\text{new}} \leftarrow$  select two fresh variables from $\mathsf{U}$.       \label{algo3:v_new} \\
            \If{$\mathbf{a}$ is an opening label in $\Sigma_\alpha$ \textcolor{silvergrey}{/$\Sigma_\beta$}  \label{algo3:opening} }
            { $\Psi_r \leftarrow  v_{\text{old}}=\pi(\mathbf{a})(v_{\text{new}}) \land v^{\prime}_{\text{old}}=v^{\prime}_{\text{new}} $
             \textcolor{silvergrey}{/$\Psi_r \leftarrow v^{\prime}_{\text{old}}=\pi^{\prime}(\mathbf{a})(v^{\prime}_{\text{new}}) \land v_{\text{old}}=v_{\text{new}}$}}  \label{algo3:encode_opening}
            \If{$\mathbf{a}$ is an closing label in $\Sigma_\alpha$ \textcolor{silvergrey}{/$\Sigma_\beta$}   \label{algo3:closing} }
            { $\Psi_r \leftarrow v_{\text{new}}=\pi(\mathbf{a})(v_{\text{old}}) \land v^{\prime}_{\text{old}}=v^{\prime}_{\text{new}} $
             \textcolor{silvergrey}{/$\Psi_r \leftarrow v^{\prime}_{\text{new}}=\pi^{\prime}(\mathbf{a})(v^{\prime}_{\text{old}}) \land v_{\text{old}}=v_{\text{new}} $} }   \label{algo3:encode_closing}
            % \If{$\mathbf{a}$ is an opening label in $\Sigma_\beta$ }
            % { $\Psi_r \leftarrow v^{\prime}_{\text{old}}=\pi(\mathbf{a})(v^{\prime}_{\text{new}}) \land v_{\text{old}}=v_{\text{new}} $ }
            % \If{$\mathbf{a}$ is an closing label in $\Sigma_\beta$ }
            % { $\Psi_r \leftarrow v^{\prime}_{\text{new}}=\pi(\mathbf{a})(v^{\prime}_{\text{old}}) \land v_{\text{old}}=v_{\text{new}} $ }
            \If{$\mathbf{a}=\epsilon$    \label{algo3:epsilon} }
            { $\Psi_r \leftarrow v_{\text{old}}=v_{\text{new}} \land v^{\prime}_{\text{old}}=v^{\prime}_{\text{new}}$             \label{algo3:encode_epsilon}}
            $v_{\text{old}} \leftarrow v_{\text{new}}, v^{\prime}_{\text{old}} \leftarrow v^{\prime}_{\text{new}}$                \label{algo3:update_v_old}
         }
         \If{$\mathcal{R} = \mathcal{R}_1 \cdot \mathcal{R}_2$ \textcolor{silvergrey}{/$\mathcal{R} = \mathcal{R}_1 + \mathcal{R}_2$}  \label{algo3:recu1} }
         {
           $\Psi_r = \FSubFunc(\mathcal{R}_1) \land \FSubFunc(\mathcal{R}_2)$
           \textcolor{silvergrey}{/$\Psi_r = \FSubFunc(\mathcal{R}_1) \lor \FSubFunc(\mathcal{R}_2)$}                                     \label{algo3:recu2}
         }

         \If{$\mathcal{R} = \mathcal{R}_1 \text{\Kleene}$                    \label{algo3:kleene_star} }
         {
            Let $n \leftarrow$  select a fresh variable from $\mathsf{N}$                      \label{algo3:select_n} \\
            Let $\Psi_{r} \leftarrow \texttt{EncodingStar}(\mathcal{R}_1, v_{\text{old}}, n, D_{\alpha}, D_{\beta})$        \label{algo3:encoding_star}
            % Let $v_{t} \leftarrow v_{\text{old}},\ v^{\prime}_{t} \leftarrow v^{\prime}_{\text{old}},\ \Psi_{star} \leftarrow \FSubFunc(\mathcal{R}_1)$ \\
            % \If{ $\Psi_{star} \vdash v_{\text{old}/t} = f_1 \circ \cdots \circ f_k(v_{t/\text{old}}) \land  v^{\prime}_{\text{old}/t} = g_1 \circ \cdots \circ g_m(v^{\prime}_{t/\text{old}})$ \\
            % where $f_1, \cdots, f_k, g_1, \cdots, g_m \in \mathsf{F} \cup \{Id\}$}
            % {
            %     Let $n \leftarrow$  select a fresh variable from $\mathsf{N}$ \\
            %     $\Psi_r \leftarrow v_{\text{old}/t} = (f_1 \circ \cdots \circ f_k)^n (v_{t/\text{old}}) \land v^{\prime}_{\text{old}/t} = (g_1 \circ \cdots \circ g_m)^n (v^{\prime}_{t/\text{old}})$
            % }\Else{
            %     $v_{\text{old}} \leftarrow v_{t} $
            % }
         }
        return $\Psi_r$                                   \label{algo3:end}
    }
\end{algorithm}

%% file: figure/encoding_block.tex
  \begin{tcolorbox}[colback=gray!10, colframe=black!30, boxsep=5pt, left=10pt, right=10pt, top=10pt, bottom=10pt]
    \vspace{-12pt}
    \begin{subfigure}{.28\textwidth}
          \centering
          \begin{align*}
            % & // block:        \\
            & \mathtt{x := f(y);}       \\
            & \mathtt{y := g(z);}       \\
            & \mathtt{z := h(x);}
          \end{align*}
          \vspace{-10pt}
          \caption{Code Structure of $\mathit{block}$.}
      \end{subfigure}
      \hfill
      \begin{subfigure}{.28\textwidth}
          \centering
          \begin{align*}
            % \\
            &  \zeta (\mathit{block}, \mathtt{x}) = f (\mathtt{y})  \\
            &  \zeta (\mathit{block}, \mathtt{y}) = g(\mathtt{z})  \\
            &  \zeta (\mathit{block}, \mathtt{z}) = (h \circ f)(\mathtt{y})
          \end{align*}
          \vspace{-10pt}
          \caption{Summary functions of $\mathit{block}$.}
          \label{subfig:summary_block}
      \end{subfigure}
      \hfill
      \begin{subfigure}{.40\textwidth}
          \centering
          \begin{align*}
            % \\
            &  \zeta (\mathit{block}^2, \mathtt{x}) = (f \circ g) (\mathtt{z})  \\
            &  \zeta (\mathit{block}^2, \mathtt{y}) = (g \circ h \circ f) (\mathtt{y})   \\
            &  \zeta (\mathit{block}^2, \mathtt{z}) = (h \circ f \circ g ) (\mathtt{z})
          \end{align*}
          \vspace{-10pt}
          \caption{Summary functions of $\mathit{block}^2$.}
      \end{subfigure}

      \vspace{-12pt}

      \begin{subfigure}{.40\textwidth}
          \centering
          \begin{align*}
            \\
            &  \zeta (\mathit{block}^3, \mathtt{x}) = (f \circ g \circ h \circ f) (\mathtt{y})   \\
            &  \zeta (\mathit{block}^3, \mathtt{y}) = (g \circ h \circ f \circ g) (\mathtt{z})    \\
            &  \zeta (\mathit{block}^3, \mathtt{z}) = (h \circ f \circ g \circ h \circ f) (\mathtt{y})
          \end{align*}
          \vspace{-10pt}
          \caption{Summary functions of $\mathit{block}^3$.}
      \end{subfigure}
      \hfill
      \begin{subfigure}{.40\textwidth}
            \centering
            \begin{align*}
              \\
              &  \zeta (\mathit{block}^4, \mathtt{x}) = (f \circ g \circ h \circ f \circ g) (\mathtt{z})   \\
              &  \zeta (\mathit{block}^4, \mathtt{y}) = (g  \circ h \circ f \circ g \circ h \circ f) (\mathtt{y})    \\
              &  \zeta (\mathit{block}^4, \mathtt{z}) = (h \circ  f \circ g \circ h \circ f \circ g) (\mathtt{z})
            \end{align*}
            \vspace{-10pt}
            \caption{Summary functions of $\mathit{block}^4$.}
        \end{subfigure}
        \vspace{-8pt}
   \end{tcolorbox}

%% file: figure/encoding_UP_rule.tex
\[
    \frac{
         P=\mathtt{skip;s}   \rightarrow  P=\mathtt{s}
    }{
        \stdef{ \gamma , pc }  \xrightarrow{\mathtt{skip}}  \stdef {\gamma, pc }
    }
\]

\[
    \frac{
         P=\mathtt{x:=e;s}  \rightarrow  P=\mathtt{s} \wedge x^{\prime} \in \mathsf{V}_{\mathsf{U}}  \text{ is fresh in}\ pc
    }{
        \stdef{ \gamma , pc }  \xrightarrow{\mathtt{x:=e}}  \stdef {\gamma [ \mathtt{x} \mapsto x^{\prime} ], pc \wedge x^{\prime} = \mathtt{e} \downarrow_{\gamma} }
    }
\]

\[
    \frac{
         P=\mathtt{assume(cond);s}  \rightarrow  P=\mathtt{s}
    }{
        \stdef{ \gamma , pc }  \xrightarrow{\mathtt{assume(cond)}}  \stdef {\gamma, pc \wedge \mathtt{cond} \downarrow_{\gamma} }
    }
\]

\[
    \frac{
        P = \mathtt{s_1; s_2}   \rightarrow   P = \mathtt{ s_2}
    }{
        \stdef{\gamma, pc}  \xrightarrow{\mathtt{s_1}}  \stdef{\gamma^{\prime}, pc^{\prime}}
    }
\]

\[
    \frac{
         P= \mathtt{if \ (cond) \ then \  s_1\ else \ s_2;s}  \rightarrow  P_1=\mathtt{s_1;s}  \uplus   P_2=\mathtt{s_2;s}
    }{
        \stdef{ \gamma , pc }  \xrightarrow{\mathtt{if \ block}}  \stdef {\gamma, pc \wedge \mathtt{cond} \downarrow_{\gamma} } \uplus  \stdef {\gamma, pc \wedge \neg \mathtt{cond}\downarrow_{\gamma} }
    }
\]

\[
    \frac{
    \begin{array}{c}
        P = \mathtt{while\ (x=y) \{ block_1;assume(z=k); \dots \}; s } \rightarrow P = \mathtt{s} \\
        \wedge  x^{\prime}, y^{\prime}, \dots  \in \mathsf{V}_{\mathsf{U}}  \text{ is fresh in} \ pc \\
    \end{array}
    }{
        \begin{array}{c}
            \stdef{ \gamma, pc} \xrightarrow{\mathtt{while \ block}} \stdef{ \gamma[\mathtt{x} \mapsto x^{\prime}, \mathtt{y} \mapsto y^{\prime}, \dots], pc \wedge pc_{\text{new}}} \\
            pc_{\text{new}} = \texttt{EncodingLoop}(\gamma, \mathtt{while \ (x=y)\  \{  block_1;assume(z=b); \dots \} } )
        \end{array}
    }
\]

% \[
%     \stdef{ \mathbf{skip};s , q, pc}  \rightarrow  \stdef {s, q, pc }
% \]

% \[
%     \frac{
%       \stdef{ e, q } \Downarrow v   \quad    x^{\prime} \in \mathsf{V}_{\mathsf{U}} \ is \ fresh \ in \ pc
% }{
%       \stdef{x:=e, q, pc}  \rightarrow  \stdef{\textbf{ skip }; q [ x \mapsto x^{\prime} ] , pc \wedge x^{\prime}=v}
% }
% \]

% \[
%     \frac{
%       \stdef{ cond, q } \Downarrow v
% }{
%       \stdef{assume(cond);s, q, pc}  \rightarrow  \stdef{\textbf{ skip };s, q, pc \wedge v}
% }
% \]

% \[
%     \frac{
%       \stdef{ s_1, q, pc} \rightarrow \stdef{ {s_1}^{\prime}, q^{\prime}, pc^{\prime}}
% }{
%       \stdef{s_1; s_2, q, pc} \rightarrow  \stdef{{s_1}^{\prime}; s_2, q^{\prime}, pc^{\prime}}
% }
% \]

% \[
%     \stdef{ \mathbf{if} \  (cond)\  \mathbf{then}\  s_1\  \mathbf{else}\  s_2; s, pc}  \rightarrow  \stdef {s_1; s, pc \wedge cond} \uplus  \stdef {s_2; s, pc \wedge \neg cond}
% \]

% \[
%     \frac{
%     \begin{array}{c}
%         x^{\prime}, y^{\prime}, \dots  \in \mathsf{V}_{\mathsf{U}} \ is \ fresh \ in \ pc \\
%     \end{array}
%     }{
%     \begin{array}{c}
%         \stdef{ \mathbf{while} \ (x=y)\  \{  block_1;assume(z=k); \dots \}; s, q, pc} \rightarrow  \stdef{ s, q[x \mapsto x^{\prime}, y \mapsto y^{\prime}, \dots], pc \wedge pc_{new}} \\
%         pc_{new} ::= EncodingLoop(q, \mathbf{while} \ (x=y)\  \{  block_1;assume(z=k); \dots \} )
%     \end{array}
%     }
% \]

%% file: algorithms/algo_encodingloop.tex
\begin{algorithm}[!ht]
    \caption{ EncodingLoop($\gamma, \texttt{Loop}$) } \label{algo:encodeing_loop}
    \KwIn{ $ \texttt{Loop} ::= \mathtt{while \ (x=y)\  \{block_1;assume(z=b);\dots\} }$ }
    \KwOut{The formula derived from encoding $\texttt{Loop}$. }
    Let $m$ be the number of times the while loop executes.     \label{algo:loop_num}    \\
    Let $\mathtt{block} \leftarrow \mathtt{block}_1;\cdots$       \label{algo:blook}      \\
    Let $\Psi \leftarrow \mathbf{true} $      \label{algo:Psi}          \\
    \ForEach{$ \mathtt{var} \in \Sigma_{V} $}  {  \label{algo:forvar}

          %  $ f_{\mathtt{var}}(q(\mathtt{var}_1)) = \zeta(\mathtt{block}, \mathtt{var}) \downarrow_{\gamma} $      \tcp*[f]{  Calculate $\zeta$ functions for each $\mathtt{var}$ }    \\
           $ (k_{\mathtt{var}}, p_{\mathtt{var}}, f_{k_{\mathtt{var}}}, f_{ p_{\mathtt{var}} }) =  \mathtt{CalculatePeriod(block, \mathtt{var})}$
         }
    Let $ k \leftarrow  Max(k_\mathtt{x}, k_\mathtt{y}) $ \\
    $ \Psi \leftarrow \Psi \wedge \forall i. i < m \wedge i < k \rightarrow
    \zeta(\mathtt{block}^{i},\mathtt{x}) \downarrow_{\gamma} =\zeta(\mathtt{block}^{i}, \mathtt{y}) \downarrow_{\gamma} $  \label{algo:inloop1}  \\

    $ \Psi \leftarrow \Psi \wedge \forall i. i < m \wedge i \geq k \rightarrow  f_{k_{\mathtt{x}}}
    \circ f_{p_{\mathtt{x}}} ^ { \lfloor \frac{i-k}{p_{\mathtt{x}}} \rfloor } \circ f_{s_{\mathtt{x}}}( \gamma(\mathtt{var}_1) ) =  f_{k_{\mathtt{y}}}
    \circ f_{p_{\mathtt{y}}} ^ { \lfloor \frac{i-k}{p_{\mathtt{y}}} \rfloor}  \circ f_{ s_{\mathtt{y}} }( \gamma(\mathtt{var}_2) )   $  \label{algo:inloop2} \\

    $\Psi \leftarrow \Psi \wedge   m < k \rightarrow [ x^{\prime} = \zeta(\mathtt{block}^{m},\mathtt{x}) \downarrow_{\gamma} ] \wedge
    [ y^{\prime} = \zeta(\mathtt{block}^{m}, \mathtt{y}) \downarrow_{\gamma} ] \wedge [ x^{\prime} \neq y^{\prime} ] $   \label{algo:outloop1}  \\

    $ \Psi \leftarrow \Psi  \wedge  m \geq k \rightarrow  [ x^{\prime} = f_{k_{\mathtt{x}}}  \circ f_{p_{\mathtt{x}}} ^ { \lfloor \frac{m-k}{p_{\mathtt{x}}} \rfloor }
    \circ f_{ s_{\mathtt{x}} }( \gamma(\mathtt{var}_3) ) ] \wedge [ y^{\prime} = f_{k_{\mathtt{y}}}  \circ f_{p_{\mathtt{y}}} ^ { \lfloor \frac{m-k}{p_{\mathtt{y}}} \rfloor}  \circ f_{s_{\mathtt{y}}}( \gamma(\mathtt{var}_4) ) ]
    \wedge [ x^{\prime} \neq y^{\prime} ] $             \label{algo:outloop2}     \\

    $\Psi \leftarrow \Psi \wedge \mathtt{EncodingAssumeInLoop}(\texttt{Loop}) $   \label{algo:encodeassume}   \\
    return $ \Psi$
\end{algorithm}

%% file: ref.bib
@phdthesis{bueno2021software,
  title={Software Model Checking with Uninterpreted Functions},
  author={Bueno, Denis},
  year={2021}
}

@inproceedings{1997Automatic,
  title={Automatic verification of pipelined microprocessor control},
  author={Burch, Jerry R and Dill, David L},
  booktitle={Computer Aided Verification: 6th International Conference, CAV'94 Stanford, California, USA, June 21--23, 1994 Proceedings 6},
  pages={68--80},
  year={1994},
  organization={Springer},
  doi={10.1007/3-540-58179-0_44}
}

@inproceedings{2012Regression,
  title={Regression verification-A practical way to verify programs},
  author={Strichman, Ofer and Godlin, Benny},
  booktitle={Working Conference on Verified Software: Theories, Tools, and Experiments},
  pages={496--501},
  year={2005},
  organization={Springer},
  doi={10.1007/978-3-540-69149-5_54}
}

@article{congruence_closure1,
  title={Variations on the common subexpression problem},
  author={Downey, Peter J and Sethi, Ravi and Tarjan, Robert Endre},
  journal={Journal of the ACM (JACM)},
  volume={27},
  number={4},
  pages={758--771},
  year={1980},
  publisher={ACM New York, NY, USA},
  doi={10.1145/322217.322228}
}

@article{congruence_closure2,
  title={Fast decision procedures based on congruence closure},
  author={Nelson, Greg and Oppen, Derek C},
  journal={Journal of the ACM (JACM)},
  volume={27},
  number={2},
  pages={356--364},
  year={1980},
  publisher={ACM New York, NY, USA},
  doi={10.1145/322186.322198}
}

@article{congruence_closure3,
  title={Deciding combinations of theories},
  author={Shostak, Robert E},
  journal={Journal of the ACM (JACM)},
  volume={31},
  number={1},
  pages={1--12},
  year={1984},
  publisher={ACM New York, NY, USA},
  doi={10.1145/2422.322411}
}

@article{2007Fast,
  title={Fast congruence closure and extensions},
  author={Nieuwenhuis, Robert and Oliveras, Albert},
  journal={Information and Computation},
  volume={205},
  number={4},
  pages={557--580},
  year={2007},
  publisher={Elsevier},
  doi={10.1016/j.ic.2006.08.009}
}

@inproceedings{2005Proof,
  title={Proof-producing congruence closure},
  author={Nieuwenhuis, Robert and Oliveras, Albert},
  booktitle={International Conference on Rewriting Techniques and Applications},
  pages={453--468},
  year={2005},
  organization={Springer},
  doi={10.1007/978-3-540-32033-3_33}
}

@inproceedings{PositiveEquality,
  title={Exploiting positive equality in a logic of equality with uninterpreted functions},
  author={Bryant, Randal E and German, Steven and Velev, Miroslav N},
  booktitle={Computer Aided Verification: 11th International Conference, CAV’99 Trento, Italy, July 6--10, 1999 Proceedings 11},
  pages={470--482},
  year={1999},
  organization={Springer},
  doi={10.1007/3-540-48683-6_40}
}

@inproceedings{CLU,
  title={Modeling and verifying systems using a logic of counter arithmetic with lambda expressions and uninterpreted functions},
  author={Bryant, Randal E and Lahiri, Shuvendu K and Seshia, Sanjit A},
  booktitle={Computer Aided Verification: 14th International Conference, CAV 2002 Copenhagen, Denmark, July 27--31, 2002 Proceedings 14},
  pages={78--92},
  year={2002},
  organization={Springer},
  doi={10.1007/3-540-45657-0_7}
}

@inproceedings{muller2005checking,
  title={Checking Herbrand equalities and beyond},
  author={M{\"u}ller-Olm, Markus and R{\"u}thing, Oliver and Seidl, Helmut},
  booktitle={Verification, Model Checking, and Abstract Interpretation: 6th International Conference, VMCAI 2005, Paris, France, January 17-19, 2005. Proceedings 6},
  pages={79--96},
  year={2005},
  organization={Springer},
  doi={10.1007/978-3-540-30579-8_6}
}

@inproceedings{govind2021logical,
  title={Logical characterization of coherent uninterpreted programs},
  author={Govind, VK Hari and Shoham, Sharon and Gurfinkel, Arie},
  booktitle={2021 Formal Methods in Computer Aided Design (FMCAD)},
  pages={77--85},
  year={2021},
  organization={IEEE},
  doi={10.34727/2021/isbn.978-3-85448-046-4_16}
}

@inproceedings{lee2014unbounded,
  title={Unbounded scalable verification based on approximate property-directed reachability and datapath abstraction},
  author={Lee, Suho and Sakallah, Karem A},
  booktitle={Computer Aided Verification: 26th International Conference, CAV 2014, Held as Part of the Vienna Summer of Logic, VSL 2014, Vienna, Austria, July 18-22, 2014. Proceedings 26},
  pages={849--865},
  year={2014},
  organization={Springer},
  doi={10.1007/978-3-319-08867-9_56}
}

@article{fan2024leveraging,
  title={Leveraging Datapath Propagation in IC3 for Hardware Model Checking},
  author={Fan, Hongyu and He, Fei},
  journal={IEEE Transactions on Computer-Aided Design of Integrated Circuits and Systems},
  year={2024},
  publisher={IEEE},
  doi={10.1109/tcad.2024.3360022}
}

@article{mathur2019deciding,
  title={Deciding memory safety for single-pass heap-manipulating programs},
  author={Mathur, Umang and Murali, Adithya and Krogmeier, Paul and Madhusudan, P and Viswanathan, Mahesh},
  journal={Proceedings of the ACM on Programming Languages},
  volume={4},
  number={POPL},
  pages={1--29},
  year={2019},
  publisher={ACM New York, NY, USA},
  doi={10.1145/3371103}
}

@article{lipshitz1978diophantine,
  title={The Diophantine problem for addition and divisibility},
  author={Lipshitz, Leonard},
  journal={Transactions of the American Mathematical Society},
  pages={271--283},
  year={1978},
  publisher={JSTOR},
  doi={10.1090/s0002-9947-1978-0469886-1}
}

@article{bel1976decidability,
  title={Decidability of the universal theory of natural numbers with addition and divisibility},
  author={Bel'tyukov, Anatoly Petrovich},
  journal={Journal of Soviet Mathematics},
  volume={14},
  pages={1390--1404},
  year={1980},
  publisher={Springer},
  doi={10.1007/bf01693974}
}

@inproceedings{lechner2015complexity,
  title={On the complexity of linear arithmetic with divisibility},
  author={Lechner, Antonia and Ouaknine, Jo{\"e}l and Worrell, James},
  booktitle={2015 30th Annual ACM/IEEE Symposium on Logic in Computer Science},
  pages={667--676},
  year={2015},
  organization={IEEE},
  doi={10.1109/lics.2015.67}
}

@article{mathur2019decidable,
  title={Decidable verification of uninterpreted programs},
  author={Mathur, Umang and Madhusudan, P and Viswanathan, Mahesh},
  journal={Proceedings of the ACM on Programming Languages},
  volume={3},
  number={POPL},
  pages={1--29},
  year={2019},
  publisher={ACM New York, NY, USA},
  doi={10.1145/3290359}
}

@article{sridharan2005demand,
  title={Demand-driven points-to analysis for Java},
  author={Sridharan, Manu and Gopan, Denis and Shan, Lexin and Bod{\'\i}k, Rastislav},
  journal={ACM SIGPLAN Notices},
  volume={40},
  number={10},
  pages={59--76},
  year={2005},
  publisher={ACM New York, NY, USA},
  doi={10.1145/1094811.1094817}
}

@article{spath2019context,
  title={Context-, flow-, and field-sensitive data-flow analysis using synchronized pushdown systems},
  author={Sp{\"a}th, Johannes and Ali, Karim and Bodden, Eric},
  journal={Proceedings of the ACM on Programming Languages},
  volume={3},
  number={POPL},
  pages={1--29},
  year={2019},
  publisher={ACM New York, NY, USA},
  doi={10.1145/3290361}
}

@article{tarjan1981fast,
  title={Fast algorithms for solving path problems},
  author={Tarjan, Robert Endre},
  journal={Journal of the ACM (JACM)},
  volume={28},
  number={3},
  pages={594--614},
  year={1981},
  publisher={ACM New York, NY, USA},
  doi={10.1145/322261.322273}
}

@article{reps2000undecidability,
  title={Undecidability of context-sensitive data-dependence analysis},
  author={Reps, Thomas},
  journal={ACM Transactions on Programming Languages and Systems (TOPLAS)},
  volume={22},
  number={1},
  pages={162--186},
  year={2000},
  publisher={ACM New York, NY, USA},
  doi={10.1145/345099.345137}
}

@article{li2023single,
  title={Single-Source-Single-Target Interleaved-Dyck Reachability via Integer Linear Programming},
  author={Li, Yuanbo and Zhang, Qirun and Reps, Thomas},
  journal={Proceedings of the ACM on Programming Languages},
  volume={7},
  number={POPL},
  pages={1003--1026},
  year={2023},
  publisher={ACM New York, NY, USA},
  doi={10.1145/3571228}
}

@inproceedings{zhang2017context,
  title={Context-sensitive data-dependence analysis via linear conjunctive language reachability},
  author={Zhang, Qirun and Su, Zhendong},
  booktitle={Proceedings of the 44th ACM SIGPLAN Symposium on Principles of Programming Languages},
  pages={344--358},
  year={2017},
  doi={10.1145/3009837.3009848}
}

@inproceedings{ding2023mutual,
  title={Mutual Refinements of Context-Free Language Reachability},
  author={Ding, Shuo and Zhang, Qirun},
  booktitle={International Static Analysis Symposium},
  pages={231--258},
  year={2023},
  organization={Springer},
  doi={10.1007/978-3-031-44245-2_12}
}

@inproceedings{conrado2024better,
  title={A Better Approximation for Interleaved Dyck Reachability},
  author={Conrado, Giovanna Kobus and Pavlogiannis, Andreas},
  booktitle={Proceedings of the 13th ACM SIGPLAN International Workshop on the State Of the Art in Program Analysis},
  pages={18--25},
  year={2024},
  doi={10.1145/3652588.3663318}
}

@book{rosen2011elementary,
  title={Elementary number theory},
  author={Rosen, Kenneth H},
  year={2011},
  publisher={Pearson Education London}
}

@book{rosen1999discrete,
  title={Discrete mathematics and its applications},
  author={Rosen, Kenneth H and Krithivasan, Kamala},
  volume={6},
  year={1999},
  publisher={McGraw-hill New York}
}

@article{conrado2025program,
  title={Program Analysis via Multiple Context Free Language Reachability},
  author={Conrado, Giovanna Kobus and Kjelstr{\o}m, Adam Husted and van de Pol, Jaco and Pavlogiannis, Andreas},
  journal={Proceedings of the ACM on Programming Languages},
  volume={9},
  number={POPL},
  pages={509--538},
  year={2025},
  publisher={ACM New York, NY, USA},
  doi={10.1145/3704854
  }
}

@inproceedings{hong2021trace,
  title={Trace abstraction-based verification for uninterpreted programs},
  author={Hong, Weijiang and Chen, Zhenbang and Du, Yide and Wang, Ji},
  booktitle={Formal Methods: 24th International Symposium, FM 2021, Virtual Event, November 20--26, 2021, Proceedings 24},
  pages={545--562},
  year={2021},
  organization={Springer},
  doi={10.1007/978-3-030-90870-6_29}
}

@inproceedings{zhang2013fast,
  title={Fast algorithms for Dyck-CFL-reachability with applications to alias analysis},
  author={Zhang, Qirun and Lyu, Michael R and Yuan, Hao and Su, Zhendong},
  booktitle={Proceedings of the 34th ACM SIGPLAN conference on Programming language design and implementation},
  pages={435--446},
  year={2013},
  doi={10.1145/2491956.2462159}
}

@article{shostak1978algorithm,
  title={An algorithm for reasoning about equality},
  author={Shostak, Robert E},
  journal={Communications of the ACM},
  volume={21},
  number={7},
  pages={583--585},
  year={1978},
  publisher={ACM New York, NY, USA},
  doi={10.1145/359545.359570}
}

@book{kroening2008decision,
  title={Decision procedures},
  author={Kroening, Daniel and Strichman, Ofer},
  volume={5},
  year={2008},
  publisher={Springer},
  doi={10.1007/978-3-540-74105-3}
}

@article{mathur2025decision,
  title={The Decision Problem for Regular First Order Theories},
  author={Mathur, Umang and Mestel, David and Viswanathan, Mahesh},
  journal={Proceedings of the ACM on Programming Languages},
  volume={9},
  number={POPL},
  pages={986--1012},
  year={2025},
  publisher={ACM New York, NY, USA},
  doi={10.1145/3704870}
}

@book{mendelson2009introduction,
  title={Introduction to mathematical logic},
  author={Mendelson, Elliott},
  year={2009},
  publisher={Chapman and Hall/CRC},
  doi={10.1201/9781584888772}
}

@article{bachmair2003abstract,
  title={Abstract congruence closure},
  author={Bachmair, Leo and Tiwari, Ashish and Vigneron, Laurent},
  journal={Journal of Automated Reasoning},
  volume={31},
  number={2},
  pages={129--168},
  year={2003},
  publisher={Springer},
  doi={10.1023/b:jars.0000009518.26415.49}
}

@inproceedings{bachmair2000abstract,
  title={Abstract congruence closure and specializations},
  author={Bachmair, Leo and Tiwari, Ashish},
  booktitle={International Conference on Automated Deduction},
  pages={64--78},
  year={2000},
  organization={Springer},
  doi={10.1007/10721959_5}
}

@article{willsey2021egg,
  title={Egg: Fast and extensible equality saturation},
  author={Willsey, Max and Nandi, Chandrakana and Wang, Yisu Remy and Flatt, Oliver and Tatlock, Zachary and Panchekha, Pavel},
  journal={Proceedings of the ACM on Programming Languages},
  volume={5},
  number={POPL},
  pages={1--29},
  year={2021},
  publisher={ACM New York, NY, USA},
  doi={10.1145/3434304}
}

@inproceedings{de2007efficient,
  title={Efficient E-matching for SMT solvers},
  author={De Moura, Leonardo and Bj{\o}rner, Nikolaj},
  booktitle={International Conference on Automated Deduction},
  pages={183--198},
  year={2007},
  organization={Springer},
  doi={10.1007/978-3-540-73595-3_13}
}

@inproceedings{tate2009equality,
  title={Equality saturation: a new approach to optimization},
  author={Tate, Ross and Stepp, Michael and Tatlock, Zachary and Lerner, Sorin},
  booktitle={Proceedings of the 36th annual ACM SIGPLAN-SIGACT symposium on Principles of programming languages},
  pages={264--276},
  year={2009},
  doi={10.1145/1480881.1480915}
}

@article{kapur2019conditional,
  title={Conditional congruence closure over uninterpreted and interpreted symbols},
  author={Kapur Deepak},
  journal={Journal of Systems Science and Complexity},
  volume={32},
  number={1},
  pages={317--355},
  year={2019},
  publisher={Springer},
  doi={10.1007/s11424-019-8377-8}
}

@article{zakhour2025dis,
  title={Dis/Equality Graphs},
  author={Zakhour, George and Weisenburger, Pascal and Cesario, Jahrim Gabriele and Salvaneschi, Guido},
  journal={Proceedings of the ACM on Programming Languages},
  volume={9},
  number={POPL},
  pages={2282--2305},
  year={2025},
  publisher={ACM New York, NY, USA},
  doi={10.1145/3704913}
}

@inproceedings{du2026euf,
  title={EUF-based Solving Dyck-Reachability with Applications to Static Analysis},
  author={Du, Yide and Chen, Zhenbang and Liu, Kunlin and Zhang, Guofeng and Wang, Xudong and Ma, Ke and Dong, Wei and Wang, Ji},
  booktitle={International Symposium on Formal Methods},
  pages={296--316},
  year={2026},
  organization={Springer},
  doi={10.1007/978-3-032-26220-2_15}
}

@inproceedings{gulwani2007assertion,
  title={Assertion checking unified},
  author={Gulwani, Sumit and Tiwari, Ashish},
  booktitle={International Workshop on Verification, Model Checking, and Abstract Interpretation},
  pages={363--377},
  year={2007},
  organization={Springer},
  doi={10.1007/978-3-540-69738-1_26}
}

@inproceedings{godoy2009invariant,
  title={Invariant checking for programs with procedure calls},
  author={Godoy, Guillem and Tiwari, Ashish},
  booktitle={International Static Analysis Symposium},
  pages={326--342},
  year={2009},
  organization={Springer},
  doi={10.1007/978-3-642-03237-0_22}
}
